\documentclass[a4paper,11pt,nocut]{article}
\usepackage[top=2.5cm, bottom=2.5cm, left=2cm, right=2cm]{geometry}
\usepackage{settings_custom}
\allowdisplaybreaks

\begin{document}

\title{Injectivity of ReLU networks: perspectives from statistical physics}
\date{\today}
\author{Antoine Maillard$^{\star, \diamond}$, Afonso S.\ Bandeira$^\star$, David Belius$^{\sharp,\triangleleft}$, Ivan Dokmani\'c$^{\sharp,\flat}$, Shuta Nakajima$^\triangleright$}
\maketitle

{\let\thefootnote\relax\footnote{
    \noindent
$\star$ Department of Mathematics, ETH Z\"urich, Switzerland.\\
$\sharp$ Department of Mathematics and Computer Science, University of Basel, Switzerland.\\
$\flat$ Department of Electrical and Computer Engineering, University of Illinois at Urbana--Champaign, USA.\\
$\triangleright$ Graduate School of Science and Technology, Meiji University, Kanagawa, Japan.\\
$\triangleleft$ Faculty of Mathematics and Computer Science, UniDistance Suisse.\\
$\diamond$ To whom correspondence shall be sent: \href{mailto:antoine.maillard@math.ethz.ch}{antoine.maillard@math.ethz.ch}.
}}
\setcounter{footnote}{0}

\begin{abstract}
    \noindent
    When can the input of a ReLU neural network be inferred from its output? In other words, when is the network injective? 
    We consider a single layer, $x \mapsto \mathrm{ReLU}(Wx)$, with a random Gaussian $m \times n$ matrix $W$, 
    in a high-dimensional setting where $n, m \to \infty$.
    Recent work connects this problem to spherical integral geometry giving rise to a conjectured sharp injectivity threshold for  $\alpha = m/n$ by studying the expected Euler characteristic of a certain random set.
    We adopt a different perspective and show that injectivity is equivalent to a property of the ground state of the spherical perceptron, an important spin glass model in statistical physics. By leveraging the (non-rigorous) replica symmetry-breaking theory, we derive analytical equations for the threshold whose solution is at odds with that from the Euler characteristic. Furthermore, we use Gordon's min--max theorem to prove that a replica-symmetric upper bound refutes the Euler characteristic prediction. 
    Along the way we aim to give a tutorial-style introduction to key ideas from statistical physics in an effort to make the exposition accessible to a broad audience.
    Our analysis establishes a connection between spin glasses and integral geometry but leaves open the problem of explaining the discrepancies.
\end{abstract}

\tableofcontents

\newpage
\section{Introduction}\label{sec:introduction}
We ask the following question: when is a randomly-initialized ReLU neural network injective?
For $n,m \geq 1$ we consider a single layer at initialization, that is the map $\varphi_\bW$ defined as
\begin{align}\label{eq:relu_layer}
    \varphi_\bW(\bx)_\mu &= \sigma\Big[\Big(\frac{\bW \bx}{\sqrt{n}}\Big)_\mu\Big], \hspace{2cm} \mu = 1,\cdots,m,
\end{align}
with $\bx \in \bbR^n$ and $\sigma(x) \coloneqq \max(0,x)$, the ReLU activation. 
The weights at initialization are random; concretely, we let $W_{\mu i} \iid \mathcal{N}(0,1)$ in what follows,
although we expect some of our results to generalize to $\bW$ with independent entries with zero mean, unit variance, and uniformly bounded third moment; see the discussion on universality in Section~\ref{subsec:statphys}.

\myskip
Earlier work studied this question in the proportional growth asymptotics, where $n \to \infty$ and the aspect ratio $\frac{m}{n} \to \alpha > 0$. 
Puthawala et al.\ proved that there exist values $\alpha_l$ and $\alpha_h$, with $\alpha_l < \alpha_h$, such that the probability $p_{m, n}$ that the map $\varphi_\bW$ is injective converges to $1$ for $\alpha > \alpha_h$ and to $0$ for $\alpha < \alpha_l$, 
suggesting that interesting transitions may appear precisely in this proportional scaling \cite{puthawala2022globally} .
Indeed, by studying the expected Euler characteristic of the intersection of a random subspace with a union of orthants with sufficiently many negative coordinates, 
Clum, Paleka, Bandeira, and Mixon conjectured a sharp injectivity threshold at the value $\alpha_\inj^\Eul \simeq 8.34$~\cite{paleka2021injectivity,clum2022topics,clum_paleka_bandeira_mixon}. We adopt this setting and propose an alternative derivation of the injectivity threshold, by making a connection with a spin glass model known as the spherical perceptron.

\subsection{Injectivity and (random) neural networks}

Our focus is on framing injectivity as a statistical physics problem and exploring parallels and discrepancies with the mentioned conjecture based on integral geometry\footnote{If formal injectivity was itself the goal, we could simply replace ReLU by Leaky ReLU and reduce the problem to injectivity of matrices. In that case the interesting quantity to study may be the inverse Lipschitz constant.}.
But a study of injectivity has a variety of motivations in contemporary machine learning. 
Inferring $\bx$ from $\varphi_{\bW}(\bx)$ is an ill-posed problem unless $\varphi_{\bW}$ is injective. The question thus arises naturally when applying neural networks to model forward and inverse maps in inverse problems \cite{puthawala2022globally, arridge2019solving}. There has been considerable interest in inverting generative models on their range to regularize ill-posed inverse problems \cite{bora2017compressed} and in building injective generative models \cite{brehmer2020flows,kothari2021trumpets,ross2021tractable}. Normalizing flows are designed to be invertible with efficiently computable inverses; similar feats can be achieved with injective maps, even with ReLU activations, while retaining favorable approximation-theoretic properties \cite{puthawala2022globally,puthawala2022universal,kothari2021trumpets}. In finite dimension injective maps are (locally) Lipschitz \cite{stefanov2009linearizing}. There is significant work on estimating and controlling the Lipschitz constants of deep neural networks; see for example \cite{fazlyab2019efficient,jordan2020exactly,gouk2021regularisation} and references therein.

\myskip
Applications abound beyond inverse problems: certain injective generative models can provably be trained with sample complexity which is polynomial in image dimension \cite{bai2018approximability}; a message-passing graph neural networks is as powerful as the Weisfeiler--Lehman test, but only if the aggregation function is injective \cite{xu2018powerful}; injective ReLU networks are universal approximators of \textit{any} map with a sufficiently high-dimensional output space \cite{puthawala2022globally} as well as of densities on manifolds \cite{puthawala2022universal}.

\myskip
There is an analogy between random neural networks and random matrices.
Just as results for random matrices help us understand general matrices and have implications throughout  mathematics, physics, engineering, and computer science, random neural networks yield insight into general neural networks.
This is the perspective of recent work on ``nonlinear random matrix theory'' for machine learning \cite{pennington2017nonlinear,louart2018random}.
We mention two other examples from this emerging line of research: neural networks at initialization have been used to theoretically study batch normalization \cite{daneshmand2020batch} and properties of gradients in deep networks \cite{hanin2020products}.

\subsection{Injectivity and random geometry}

\paragraph{Notation --} 
$\bbN^\star = \bbZ_{> 0}$ denotes the positive integers.
We say that an event occurs with high probability (w.h.p.) when its probability is $1 - \smallO(1)$ as the dimension $n \to \infty$. 
We denote $\mu_n$ the uniform probability measure on the Euclidean unit sphere $\mcS^{n-1}$ in $\bbR^n$.
The symbol $\pto$ refers to convergence in probability, and $\mcD \xi$ is the standard Gaussian measure on $\bbR$, as usual in physics.

\myskip
Our first tool is a proposition proved in Appendix~\ref{subsec_app:proof_random_intersection}, stated as Proposition~4.10 in \cite{paleka2021injectivity} and Proposition~37 in \cite{clum2022topics}, which is a simple consequence of Theorem~1 of \cite{puthawala2022globally}.
It connects injectivity to random geometry:
\begin{proposition}[Injectivity and random geometry]
    \label{prop:injectivity_random_intersection}
    \noindent
    The probability $p_{m,n}$ that $\varphi_\bW$ is injective is
    \begin{align}\label{eq:pmn_random_intersection}
       p_{m,n} &= \bbP_V\big[V \cap C_{m,n} = \{0\} \big],
    \end{align}
    where $V$ is a uniformly random $n$-dimensional subspace of $\bbR^m$, and $C_{m,n}$ is the set of vectors in $\bbR^m$ with strictly less than $n$ strictly positive coordinates.
\end{proposition}
\textbf{Remark --} Since $V \cap C_{m,n}$ is a cone, we can equivalently ask in eq.~\eqref{eq:pmn_random_intersection} that $V \cap C_{m,n} \cap \mcS^{m-1}$ be an empty set.

\myskip
Recall that we will study injectivity for large matrices $\bW$ in the proportional growth asymptotics,
\begin{equation*}
    n \to \infty, \quad m / n \to \alpha.
\end{equation*}
In what follows we will only consider the case $m \geq n$ (and therefore $\alpha \geq 1$). For $m < n$ even $\bx \mapsto \bW \bx$ is not injective, implying that $p_{m,n} = 0$.
An analysis of the random subspace--set intersection introduced in Proposition~\ref{prop:injectivity_random_intersection}, 
which is based on the phase transition in the expected Euler characteristic, yields a sharp injectivity threshold prediction of $\alpha_\inj^\mathrm{Euler} \simeq 8.34$ \cite{paleka2021injectivity}, 
see Section~\ref{subsec:related_work}.
Here we refute this prediction and conjecture a new threshold based on a different geometric intuition.

\subsection{Statistical physics and the spherical perceptron}\label{subsec:statphys}

\paragraph{Injectivity as energy minimization --}
The random subspace $V$ of Proposition~\ref{prop:injectivity_random_intersection} is constructed as the column space of the random matrix $\bW$, 
which has dimension $n$ with probability $1$ when $m \geq n$.
If $V' \coloneqq \bW(\mcS^{n-1})$ is the image of the $n$-dimensional unit sphere, 
we have that $\bbP[V \cap C_{m,n} = \{0\}] = \bbP[V' \cap C_{m,n} = \emptyset]$.
Moreover, for any $\bx \in \mcS^{n-1}$ we can define $E_\bW(\bx)$ as the total number of positive coordinates of $\bW \bx$, and $e_\bW(\bx)$ as a normalization of this quantity:
\begin{align}\label{eq:def_energy}
    E_\bW(\bx) &\coloneqq \sum_{\mu=1}^m \theta[(\bW \bx)_\mu], \hspace{2cm} e_\bW(\bx) \coloneqq \frac{E_\bW(\bx)}{n},
\end{align}
where $\theta(x) = \indi(x > 0)$ is the Heaviside step function, with the convention $\theta(0) = 0$.
Since $C_{m,n}$ is the set of all vectors in $\bbR^m$ with strictly less than $n$ (strictly) positive coordinates, 
one has immediately that $\bW \bx \in C_{m,n} \Leftrightarrow E_\bW(\bx) < n$.
Therefore, by Proposition~\ref{prop:injectivity_random_intersection}, $p_{m,n}$ can be rewritten as\footnote{
    The minimum is always reached since $E_\bW(\mcS^{n-1})$ is a finite set.
    }
\begin{align}\label{eq:pmn_minimum}
    p_{m,n} = \bbP_{\bW} \Big[\min_{\bx \in \mcS^{n-1}} E_{\bW}(\bx) \geq n\Big].
\end{align}
Eqs.~\eqref{eq:pmn_random_intersection} and \eqref{eq:pmn_minimum} express two different geometric intuitions. The former one lives in $\bbR^m$ (recall that $m \geq n$) and it is about an intersection of a random $n$-dimensional subspace and a certain nonconvex union of orthants. The latter one lives in $\bbR^n$ and it is about the existence of a halfspace which contains less than $n$ (out of $m$) random vectors. The two intuitions naturally encourage different analytic tools.

\paragraph{Statistical physics of disordered systems --}
The right-hand side of eq.~\eqref{eq:pmn_minimum} is reminiscent of quantities
that theoretical physicists have been tackling since the 1970s, in the field of \emph{physics of disordered systems}.
In these disordered models (also known as \emph{spin glasses}), one wishes to minimize an energy function like $E_\bW$, 
which is itself a function of random interactions (also called \emph{quenched} disorder), represented in our case by $\bW$. We recommend the famous book by Mézard, Parisi, and Virasoro for a beautiful review of the early breakthroughs of the physics of spin glasses  \cite{mezard1987spin}.

\myskip
Given this short description, we can see that eq.~\eqref{eq:pmn_minimum} fits the framework of these studies: 
the energy function given in eq.~\eqref{eq:def_energy} defines a model known in the statistical physics literature as the 
\emph{spherical perceptron} (sometimes referred to as the Gardner--Derrida perceptron \cite{gardner1988optimal} when $E_\bW(\bx)$ is given by eq.~\eqref{eq:def_energy}). 
We consider this perhaps unexpected point of view 
on injectivity of random layers in neural networks.

\paragraph{Cover's theorem and the bound $\alpha_\inj \geq 3$ --}
Cover's theorem \cite{cover1965geometrical} leads to a first natural bound for $\alpha_\inj$. 
It implies that for $\alpha < 2$, there exists with high probability (as $n \to \infty$) $\bx \in \mcS^{n-1}$ s.t.\ $E_\bW(\bx) = 0$ (that is, the constraint satisfaction problem $E_\bW(\bx) = 0$ is satisfiable w.h.p.).
One can easily deduce from this that $\alpha_\inj \geq 3$:
\begin{lemma}[Cover's lower bound for injectivity]\label{lemma:cover}
    \noindent
    Assume $\alpha < 3$. Then as $n,m \to \infty$ the ReLU layer is non injective with high probability, 
    i.e., $\lim_{n \to \infty} p_{m,n} = 0$.
\end{lemma}
Such arguments are classical, and we detail the proof of Lemma~\ref{lemma:cover} for completeness in Appendix~\ref{subsec_app:proof_cover}\footnote{In a nutshell, by Lemma \ref{lemma:cover} there is always an $\bx$ at obtuse angle with the top $2n$ rows of $\bW$. Even if all the remaining $m - 2n$ rows form acute angles with $\bx$, we need at least $n$ such rows for injectivity.}.
Results about the perceptron based on Cover's theorem were greatly extended by Gardner and Derrida \cite{gardner1988space,gardner1988optimal} using non-rigorous tools, and then later rigorously justified by Scherbina and Tirozzi \cite{shcherbina2002volume,shcherbina2003rigorous} and Stojnic \cite{stojnic2013another}. 
In the constraint satisfaction problem (CSP) view on the perceptron, $\alpha = 2$ is sometimes referred to as the \emph{Gardner capacity}, which marks the limit between the satisfiable (SAT) and unsatisfiable (UNSAT) phases.

\paragraph{Thermal relaxation: the Gibbs--Boltzmann distribution --}
Statistical physicists characterize the landscape of the (random) energy function $E_\bW(\bx)$ by considering the 
\emph{Gibbs--Boltzmann} distribution $\bbP_{\beta,\bW}$, defined for any \emph{inverse temperature} $\beta \geq 0$ as
\begin{align}\label{eq:def_Gibbs}
    \rd \bbP_{\beta,\bW}(\bx) &\coloneqq \frac{1}{\mcZ_n(\bW, \beta)} e^{-\beta E_\bW(\bx)}\mu_n(\rd \bx). \hspace{1cm} (\bx \in \mcS^{n-1})
\end{align}
Informally, the parameter $\beta \geq 0$ interpolates between two extremes: the infinite-temperature ($\beta = 0$) regime, in which the Gibbs measure is uniform on the sphere, and 
the zero-temperature ($\beta \to \infty$) limit, in which the Gibbs measure is concentrated on the global minima of the energy function $E_\bW(\bx)$.
Studying the properties of the Gibbs measure for $n \to \infty$ at various $\beta$ (remaining finite when $n \to \infty$) yields deep insight about the landscape of the corresponding energy function \cite{ellis2006entropy}\footnote{The Gibbs distribution is also the invariant measure of stochastic optimization procedures such as Langevin dynamics.}.
In particular, many of our results will be based on an analysis of the large $n$ limit of the \emph{free entropy}, which is defined as the 
the logarithm of the normalization in eq.~\eqref{eq:def_Gibbs}:
\begin{align}\label{eq:def_phi}
    \Phi_n(\bW,\beta) &\coloneqq \frac{1}{n} \log \mcZ_n(\bW, \beta) = \frac{1}{n} \log \int_{\mcS^{n-1}} \mu_n(\rd \bx) \, e^{-\beta E_\bW(\bx)}.
\end{align}

\paragraph{Universality of the free entropy --}
Following classical arguments based on the Lindeberg exchange method \cite{chatterjee2006generalization}, one can
show that the free entropy $\Phi(\alpha, \beta)$ is universal for all matrices $\bW$ 
with independent zero-mean entries with unit variance and uniformly bounded third moment.
In particular, all our conjectures and theorems on the free entropy can be stated in this more general case.
We note that in a recent line of work, similar universality properties have been generalized to matrices with independent rows (see, e.g., 
\cite{montanari2022universality,gerace2024gaussian}
and references therein) under a ``one-dimensional CLT'' condition. 
In particular, \cite{gerace2024gaussian} conjectures that the ground state energy $f^\star(\alpha) = \lim_{n \to \infty} \{\min_\bx e_\bW(\bx)\}$ (shown in Fig.~\ref{fig:chi_estar_T0}) 
is universal with respect to the distribution of $\bW$ in a much wider class than matrices with independent elements: we leave the investigation of this conjecture and its implications on injectivity for future work.

\subsection{Related work}\label{subsec:related_work}

\paragraph{Average Euler characteristic prediction --}
We follow here closely the presentation of \cite{paleka2021injectivity} (see also \cite{clum2022topics}).
Proposition~\ref{prop:injectivity_random_intersection} is reminiscent of the kinematic formulas in integral geometry
\cite{schneider2008stochastic}, that
allow to compute expressions of the type $\EE[F(V \cap C)]$, when $V$ is a uniformly-sampled random $n$-dimensional subspace, and
\begin{itemize}
    \item[$(i)$] $C$ is a finite union of convex cones. 
    \item[$(ii)$] $F$ is an additive function, i.e.,\ it satisfies for any $A, B \subseteq \bbR^{m}$ that $F(A \cup B) + F(A \cap B) = F(A) + F(B)$.
\end{itemize}
Recall that we can write eq.~\eqref{eq:pmn_random_intersection} as $p_{m,n} = \EE[\indi_\mcS(V \cap C_{m,n})]$, with $\indi_\mcS(A) \coloneqq \indi\{A \cap \mcS^{m-1} \neq \emptyset\}$ the indicator function of the 
sphere.
While $C_{m,n}$ is indeed a finite union of orthants (and thus of convex cones), $\indi_\mcS$ is not additive.
However, it follows from Groemer's extension theorem \cite{schneider2008stochastic} that 
the unique additive function defined on finite unions of convex cones to agree with $\indi_\mcS$ on convex cones 
is the (spherical) Euler characteristic $\chi_\mcS(A) \coloneqq \chi(A \cap \mcS^{m-1})$.
A possible heuristic is thus to approximate $p_{m,n} = \EE[\indi_\mcS(V \cap C_{m,n})]$ by
\begin{align}\label{eq:def_qmn}
    q_{m,n} \coloneqq \EE[\chi_\mcS(V \cap C_{m,n})],
\end{align}
in order to apply the kinematic formulas.
We refer to \cite{paleka2021injectivity} for more discussion on the validity of this heuristic.
In particular, let us note that this strategy has also been used to estimate 
the probability of excursions of random fields, 
see \cite{adler2007random}.
Using the kinematic formulas, one can obtain an explicit formula for $q_{m,n}$. 
Estimating its limit as $n, m \to \infty$ is involved, and a non-rigorous calculation performed in \cite{paleka2021injectivity}
leads to the conjecture: 
\begin{align}\label{eq:conj_lim_qmn}
    \begin{dcases}
        \limsup_{n \to \infty} \frac{1}{n} \log q_{m,n} < 0 & \textrm{ for } \alpha < \alpha_\inj^\Eul , \\
        \liminf_{n \to \infty} \frac{1}{n} \log q_{m,n} > 0 & \textrm{ for } \alpha > \alpha_\inj^\Eul,
    \end{dcases}
\end{align}
for a sharp threshold $\alpha_\inj^\Eul \simeq 8.34$, which we will call the average Euler characteristic prediction for injectivity.
Checking the validity of this heuristic approach as a prediction for the behavior of $p_{m,n}$ was one of the motivations of our work.

\paragraph{Physics and mathematics of the perceptron --}
Motivated in particular by the relation of the perceptron to continuous constraint satisfaction problems (e.g.\ to soft sphere packing),
studies of the spherical perceptron in physics and mathematics are numerous. Without aiming at being exhaustive, and rather primarily referring to works relevant for our presentation,
these studies include \cite{gardner1988optimal,gardner1988space,franz2017universality} in the physics literature,
while the spherical perceptron has also been studied with mathematically rigorous techniques \cite{shcherbina2002volume}, \cite[Chapter 3]{talagrand2010mean}, \cite[Chapter 8]{talagrand2011mean}, \cite{stojnic2013another,stojnic2013negative,montanari2024tractability}.
In particular, the satisfiability threshold $\alpha = 2$ has been rigorously determined. 
The techniques however do not apply to the unsatisfiable (UNSAT) regime, which is the one that is relevant in this paper. One reason for this is that the satisfiability question can be formulated in terms of a convex Hamiltonian, while in the unsatisfiable regime one is interested in a Hamiltonian given by the number of half-spaces a point is contained in, which is not convex. This precludes the straightforward use of these rigorous techniques to study the injectivity question. A rigorous sharp characterization of the unsatisfiable phase remains an important open problem.
We refer to \cite{bolthausen2022gardner} for a summary of current advances on the spherical perceptron, from both the physics and the mathematics points of view.

\paragraph{Other related work --}
Puthawala et al.\ derived a suite of results on injectivity of neural networks, including a simple analysis of random ReLU layers \cite{puthawala2022globally}. 
By combining ideas related to Cover's theorem with union bounds over row selections from $\bW$ and concentration of measure, they proved upper and lower bounds on the injectivity threshold, the upper bound being later improved by Paleka \cite{paleka2021injectivity} and Clum \cite{clum2022topics}. 
We summarize them in the following theorem: 
\begin{theorem}[Known bounds for injectivity \cite{puthawala2022globally,paleka2021injectivity,clum2022topics}]\label{thm:known_bounds}
    \begin{equation*}
            \Big(\alpha \leq 3.3 \Rightarrow \lim_{n \to \infty} p_{m,n} = 0\Big) \quad \textrm{and} \quad \Big(\alpha \geq 9.091 \Rightarrow \lim_{n \to \infty} p_{m,n} = 1\Big).
    \end{equation*}
\end{theorem}
By Proposition \ref{prop:injectivity_random_intersection}, the injectivity threshold can be characterized as a phase transition in the probability that a random subspace intersects a certain union of orthants. 
Similar characterizations arise in the study of convex relaxations of sparse linear regression and other high-dimensional convex optimization problems with random data.
Amelunxen et al.\ connect the probability of success of these optimization problems to random convex constraint satisfaction problems, namely the probability that two random convex cones have a common ray \cite{amelunxen2014living}.
They prove that this probability exhibits a sharp phase transition in terms of scalar values known as the \emph{statistical dimension} of the cones.
Unfortunately, these results are limited to convex cones, whereas the union of orthants from Proposition \ref{prop:injectivity_random_intersection} is non-convex.

\subsection{Main results}\label{subsec:summary_results}

Recall that we study a high-dimensional regime in which $n \to \infty$ 
and $m = m(n)$ satisfies $m(n)/n \to \alpha > 0$. 
We will sometimes use the notation $\alpha_n \coloneqq m(n)/n$.
The proofs of the rigorous statements in this section are given in Appendix~\ref{sec_app:proofs}.

\subsubsection{Relating the free entropy to injectivity}

Our starting point is eq.~\eqref{eq:pmn_minimum} in Proposition~\ref{prop:injectivity_random_intersection} (recall that $e_\bW(\bx) = E_\bW(\bx) / n$):
\begin{align*}
    p_{m,n} = \bbP_{\bW} \Big[\min_{\bx \in \mcS^{n-1}} e_{\bW}(\bx) \geq 1\Big].
\end{align*}
Recall the definition of the free entropy in eq.~\eqref{eq:def_phi}.
We immediately have 
\begin{align}\label{eq:bound_Phi_energy}
    -\frac{\Phi_n(\bW,\beta)}{\beta} \geq \min_{\bx \in \mcS^{n-1}} e_{\bW}(\bx),
\end{align}
which formalizes the fact that the Gibbs distribution is a relaxation of the uniform distribution on the global minima of $E_\bW$.
Our strategy is to use eq.~\eqref{eq:bound_Phi_energy} to characterize injectivity.
This involves two challenging steps:
\begin{itemize}
    \item[$(i)$]
    Make the inequality of eq.~\eqref{eq:bound_Phi_energy} as tight as possible:
    as we explain below,
    conjecturally, when taking $n \to \infty$ and then $\beta \to \infty$, eq.~\eqref{eq:bound_Phi_energy} becomes an equality.
    While we are not able to prove this statement, we will use it to conjecture a sharp transition for injectivity in terms of the aspect ratio $\alpha$.
    Moreover, without assuming that this conjecture holds, we will also use eq.~\eqref{eq:bound_Phi_energy} 
    to prove upper bounds on the injectivity threshold. 
    
    \item[$(ii)$]
    Second, computing the large system limit $n \to \infty$ of $\Phi_n(\bW, \beta)$ on the left-hand side of eq.~\eqref{eq:bound_Phi_energy}.
    This is a central object in the physics of disordered systems, and we will provide a conjecture for its limiting value, as well as rigorous upper bounds.
    Our results leverage a long line of work combining probability theory with heuristic predictions of statistical physics.
\end{itemize}

\myskip
The following statement is classical in the theory of disordered systems and a direct consequence of 
celebrated concentration inequalities \cite{boucheron2013concentration}.
It bounds the probability that the free entropy deviates from its mean (with respect to the disorder $\bW$):
\begin{theorem}[Free entropy concentration]\label{thm:free_entropy_concentration}
    \noindent
   For any $\beta \geq 0$ and $n \geq 1$, we have,  
   \begin{align*}
    \bbP_\bW[|\Phi_n(\bW,\beta) - \EE_\bW \Phi_n(\bW,\beta)| \geq t] &\leq 2 \exp \Big\{ - \frac{n t^2}{2 \alpha_n \beta^2} \Big\}.
   \end{align*}
\end{theorem}
Combined with the bound of eq.~\eqref{eq:bound_Phi_energy}, this already allows us to state a sufficient condition for 
non-injectivity with high probability.
We summarize this in the following corollary, proved in Appendix~\ref{subsec_app:proof_cor_sufficient_non_inj}.
\begin{corollary}[Sufficient condition for non-injectivity]\label{cor:sufficient_non_injectivity}
    \noindent
    We denote $\Phi(\alpha,\beta) = \liminf_{n \to \infty} \EE_\bW \Phi_n(\bW, \beta)$.
    It has the following properties: 
    \begin{itemize}
        \item[$(i)$] $\beta \mapsto - \Phi(\alpha,\beta)/\beta$ is a positive non-increasing function of $\beta > 0$. 
        \item[$(ii)$] Its limit as $\beta \to \infty$ satisfies
        \begin{align*}
            \lim_{\beta \to \infty} \Big[- \frac{\Phi(\alpha,\beta)}{\beta} \Big] &< 1 \Rightarrow \lim_{n \to \infty} p_{m,n} = 0,
        \end{align*}
        that is, the limit being smaller than 1 implies non-injectivity w.h.p.\ as $n,m \to \infty$\footnotemark.
    \end{itemize}
\end{corollary}
\footnotetext{The proof actually shows that $p_{m,n}$ goes to zero exponentially fast in $n$, see eq.~\eqref{eq:bound_pmn}.}

\paragraph{Existence of the limit --} 
While we expect the limit of $\EE_\bW \Phi_n(\bW, \beta)$ as $n \to \infty$ to exist, or, in other words, $\Phi(\alpha, \beta)$ to be defined not only as a $\liminf$, this fact is far from trivial.
In the spin glass literature, this has historically been shown using interpolation methods due to Guerra, by showing sub-additivity of the free entropy in the system size 
\cite{guerra2002thermodynamic,talagrand2010mean} for mean-field spin glass models possessing certain convexity properties.
Guerra's technique, however, fails beyond this setting, e.g.\ in bipartite (or other multi-species) spin glass models \cite{panchenko2015free}.
On the other hand, even in some mean-field spin glasses, including spherical $p$-spins,
the existence of the limit was only shown as a corollary of the much stronger asymptotically tight two-sided bound allowing to precisely relate the value of the limit to the Parisi formula, i.e.\ the prediction of statistical physics \cite{talagrand2006free,chen2013aizenman}\footnote{However an approximate sub-additivity property has recently been shown to be enough to deduce the convergence of the free entropy in this case \cite{subag2022convergence}.}.
In the spherical perceptron we consider here, the existence of this limit is, to the best of our knowledge, still a conjecture.

\myskip
Following the statistical physics intuition about the asymptotic tightness of eq.~\eqref{eq:bound_Phi_energy}, we conjecture the following.
\begin{conjecture}[Tightness of the free entropy bound]\label{conj:tightness_criterion}
    \noindent
        The bound of Corollary~\ref{cor:sufficient_non_injectivity} is tight, i.e.,
        \begin{align*}
            \begin{dcases}
            \lim_{\beta \to \infty} \Big[- \frac{\Phi(\alpha,\beta)}{\beta}\Big] &< 1 \Rightarrow \lim_{n \to \infty} p_{m,n} = 0, \\
            \lim_{\beta \to \infty} \Big[- \frac{\Phi(\alpha,\beta)}{\beta}\Big] &> 1 \Rightarrow \lim_{n \to \infty} p_{m,n} = 1.
            \end{dcases}
        \end{align*}
\end{conjecture}
\paragraph{A generalized conjecture --}
Conjecture \ref{conj:tightness_criterion} is a weakened version of a more general conjecture one can make from the definition of $\Phi(\alpha,\beta)$, 
which largely motivates the study of free entropies in statistical physics.
First, assume that the limit defining $\Phi(\alpha,\beta)$ is well defined, so that $\Phi(\alpha,\beta) = \lim_{n \to \infty} \EE_\bW \Phi_n(\beta,\bW)$.
As $\beta \to \infty$, we expect the configurations that have dominating mass under the Gibbs measure of eq.~\eqref{eq:def_Gibbs} 
to have the smallest energy, i.e., to be the ground state configurations.
Therefore, the stronger conjecture that motivates our use of $\Phi(\alpha,\beta)$ to characterize injectivity is that as $\beta \to \infty$, the bound of eq.~\eqref{eq:bound_Phi_energy} is actually an equality.
In a nutshell, this conjecture can be stated as ($\plim$ denotes limit in probability):
\begin{align}\label{eq:relation_Phi_ground_state}
\lim_{\beta \to \infty} -\frac{\Phi(\alpha,\beta)}{\beta} &=
    \plim_{n \to \infty} \Big\{\min_{\bx \in \mcS^{n-1}} e_\bW(\bx) \Big\}.
\end{align}
Note that such a statement also assumes the concentration of the ground state energy on a value independent of $\bW$ as $n \to \infty$.
Generally, the concentration of the intensive energy $e_\bW(\bx)$ under the Gibbs measure at any given $\beta \geq 0$
can be deduced from the existence of the limiting free entropy and its differentiability in $\beta$ \cite{auffinger2018concentration}\footnote{Unfortunately, proving these properties often requires the full power of the so-called Parisi formula for the limit of the free entropy, 
which must first be proven as we discuss later.}.

\paragraph{A remark on discretization --}
A subtlety in establishing eq.~\eqref{eq:relation_Phi_ground_state} arises from the continuous nature of the 
variable $\bx$: one needs to discard the existence of sets with ``super-exponentially'' small volume that might contain the global minima of $e_\bW$.
In discrete models this issue is often not present. For example, replacing $\int_{\mcS^{n-1}} \mu_n(\rd \bx)$ by $2^{-n} \sum_{\bx \in \{\pm 1\}^n}$ in eq.~\eqref{eq:def_phi} yields a model called the \emph{binary (or Ising) perceptron}, for which it is easy to see that
\begin{align}\label{eq:small_temp_discrete}
    \min_{\bx \in \{\pm 1\}^n} e_{\bW}(\bx) \leq - \frac{\Phi_n(\bW, \beta)}{\beta} &\leq \min_{\bx \in \{\pm 1\}^n} e_{\bW}(\bx) + \frac{\log 2}{\beta},
\end{align}
so that the generalized conjecture of eq.~\eqref{eq:relation_Phi_ground_state} follows from the concentration and existence of the limit of the free entropy.
In our spherical model one could hope to approximate $\mcS^{n-1}$ by a sufficiently fine $\varepsilon$-net, so that the value of $\Phi_n(\bW, \beta)$ is well approximated by 
averaging over the points of this net, and such that a two-sided bound similar to eq.~\eqref{eq:small_temp_discrete} holds. 
Let us briefly describe such an approach.
Considering an arbitrary fixed vector $\bx \in \mcS^{n-1}$, it is clear that with high probability there exists $\mu \in [n]$ s.t.\ $|\bW_\mu \cdot \bx| \leq 1$\footnote{Since $\{\bW_\mu \cdot \bx\}_{\mu=1}^n \iid \mcN(0,1)$.}. 
From this, one easily deduces that there exists a small rotation $\by = \bR \bx$ of $\bx$ (in the direction of $\pm \bW_\mu / \|\bW_\mu\|$), with angle $\mcO(1/\sqrt{n})$, such that $(\by \cdot \bW_\mu)(\bx \cdot \bW_\mu) < 0$, while $\|\by - \bx\|_2 \lesssim 1/\sqrt{n}.$
This (very) rough estimation shows that
$\varepsilon \lesssim 1/\sqrt{n}$ is necessary to approximate the minimum of $e_\bW$ over $\mcS^{n-1}$ by the minimum on an Euclidean-distance net.
However it is well known that such a net needs to have cardinality at least $(1 / \varepsilon)^n$ \cite{van2014probability}. 
Thus under this discretization the term $\log 2/\beta$ in the upper bound of eq.~\eqref{eq:small_temp_discrete} becomes $\Omega(\log \varepsilon^{-1} / \beta) = \Omega(\log n/\beta)$. 
Therefore we would need to consider diverging inverse temperatures $\beta = \beta(n) \gtrsim \log n$ in the discretized system for its free entropy to provably approximate the ground state energy.
A rigorous computation of the free entropy on this net with diverging $\beta$ would be challenging: since our results are based on heuristic methods of statistical physics assuming Conjecture~\ref{conj:tightness_criterion}
(with the exception of a rigorous upper bound), we leave the analysis of a possible discretization for future work.

\subsubsection{Predictions of full replica symmetry breaking theory}

Computing $\Phi(\alpha,\beta)$ is in general intractable rigorously.
We will show in Theorem~\ref{thm:bound_Gordon} that we can still derive meaningful rigorous bounds, but before describing that result we first introduce another conjecture, stemming from non-rigorous methods of statistical physics.
This conjecture, which we call a \emph{Parisi formula} as usual in spin glass models, 
and that we derive in Section~\ref{sec:full_rsb} using the non-rigorous \emph{replica method} of statistical physics,
gives us a (heuristic) means to
\textit{exactly} compute $\Phi(\alpha,\beta)$.
\begin{conjecture}[Parisi formula]\label{conj:parisi_formula}
    \noindent
    $\Phi(\alpha,\beta)$ is given by the \emph{full replica symmetry breaking} (FRSB) prediction of statistical physics, discussed in Section~\ref{sec:full_rsb}.
    More precisely, we have $\Phi(\alpha,\beta) = \Phi_\FRSB(\alpha,\beta)$, cf.\ eq.~\eqref{eq:phi_frsb}, with the following interpretation:
    \begin{itemize}
        \item[$(i)$] We have the ``Parisi formula'':
        \begin{align}\label{eq:frsb_general}
          \Phi_\FRSB(\alpha,\beta) = \inf_{q \in \mcF} \mcP[q;\alpha,\beta],  
        \end{align}
        with $\mcF$ the set of non-decreasing functions from $[0, 1]$ to $[0,1]$, and
         $\mcP[q;\alpha,\beta]$ a functional of $q$, whose expression is given in eq.~\eqref{eq:phi_frsb}.
        \item[$(ii)$] The infimum in eq.~\eqref{eq:frsb_general} is attained at a $q^\star \in \mcF$ that is the functional inverse of the CDF of a probability distribution $\rho^\star$ on $[0,1]$, such that for any continuous bounded function $f$ we have
        \begin{align}\label{eq:overlap_distribution}
            \lim_{n \to \infty} \EE_{\bW} \Big[\EE_{(\bx, \bx') \sim \bbP_{\beta,\bW}^{\otimes 2}} f(\bx \cdot \bx')\Big] \, &= \int f(u) \, \rho^\star(\rd u),
        \end{align}
        where $\bbP_{\beta,\bW}$ is the Gibbs measure defined in eq.~\eqref{eq:def_Gibbs}.
    \end{itemize}
\end{conjecture}
Eq.~\eqref{eq:overlap_distribution} shows that in the Parisi formula of eq.~\eqref{eq:frsb_general}, 
the functional parameter $q \in \mcF$ can be interpreted as the \emph{average overlap distribution} of the system.
Intuitively speaking, the ``alignment'' $\bx \cdot \bx'$ of two independent draws $\bx, \bx'$ of the Gibbs measure $\bbP_{\beta, \bW}$ (sharing the same matrix $\bW$) 
will, on average, be distributed according to $\rho^\star$ as $n \to \infty$. The fact that the large-size limit of the system is characterized by this overlap distribution (called therefore an ``order parameter'' in statistical physics)
is one of the most important predictions of the replica symmetry breaking theory of Parisi, and
we will further discuss this theory in the following.

\paragraph{Rigorous approaches --}
While the most general full replica symmetry breaking framework is widely believed to yield exact predictions in the asymptotic limit,
proving these predictions is a field of probability theory in itself. Indeed, significant progress has been made 
in some mean-field spin glass models, see e.g.\ \cite{talagrand2010mean,panchenko2014parisi},
or in the context of inference problems and the study of computational-to-statistical gaps \cite{bandeira2018notes,barbier2019optimal},
but proving the validity of the replica symmetry breaking procedure in more generality remains one of the important open problems in a rigorous description of the physics of disordered systems.
In particular, in the spherical perceptron considered here, the general full-RSB prediction is still a conjecture beyond the satisfiable phase.

\myskip
Based on Conjectures~\ref{conj:tightness_criterion} and \ref{conj:parisi_formula}, we can design a statistical physics program to characterize the injectivity of the ReLU layer:
\begin{itemize}
    \item[$(i)$] For any $\beta \geq 0$, compute $\Phi(\alpha,\beta) = \lim_{n \to \infty}  \EE_\bW \log \mcZ_n(\bW, \beta)/n$, 
    as given by the Parisi formula of Conjecture~\ref{conj:parisi_formula}.
    \item[$(ii)$] Compute analytically the limit $f^\star(\alpha) \coloneqq - \lim_{\beta\to\infty} \Phi(\alpha,\beta)/\beta$.
    As we discussed above, this is a non-decreasing function of $\alpha$, and moreover $f^\star(2) = 0$.
    \item[$(iii)$] $\varphi_\bW$ is (typically) injective if $f^\star(\alpha) > 1$, and non-injective 
    if $f^\star(\alpha) < 1$. In particular, if $f^\star$ is continuous and strictly increasing (which we numerically observe), the injectivity threshold $\alpha_\inj$ is characterized by 
    \begin{align}\label{eq:criterion_alphainj}
        \alpha_\inj &= [f^\star]^{-1}(1). 
    \end{align}
\end{itemize}
We perform this procedure in detail in Section~\ref{sec:full_rsb}, 
and it yields the main result of this section.
\begin{result}[``Full-RSB'' conjecture]\label{result:frsb_result}
    \noindent
    Assume Conjectures~\ref{conj:tightness_criterion} and \ref{conj:parisi_formula} hold.
    Denote by $A$ the event ``$\varphi_\bW$ (cf.\ eq.~\eqref{eq:relu_layer}) is injective'', 
    and let $p_{m,n} = \bbP_\bW[A]$.
    There exists a constant $\alpha_\inj^{\FRSB} \in (6.6979, 6.6982)$, 
    obtained via the Full-RSB prediction of Conjecture~\ref{conj:parisi_formula}, such that: 
    \begin{itemize}[leftmargin=25pt]
        \item[$(i)$] If $\limsup_{n \to \infty} (m/n) < \alpha_\inj^\FRSB$, then $\lim_{n \to \infty} p_{m,n} = 0$.
        \item[$(ii)$] If $\liminf_{n \to \infty} (m/n) > \alpha_\inj^\FRSB$, then $\lim_{n \to \infty} p_{m,n} = 1$.
    \end{itemize}
\end{result}

\subsubsection{Additional bounds}

\paragraph{The replica hierarchy of upper bounds --}
In Conjecture~\ref{conj:parisi_formula}, the FRSB prediction is given as
\begin{align}\label{eq:def_Phi_FRSB}
     \Phi_\FRSB(\alpha,\beta) = \inf_{ q \in \mcF } \mcP[q; \alpha,\beta],
\end{align}
and we saw that the function $q: [0,1] \to [0,1]$ could be interpreted in terms of an \emph{overlap distribution}. 
Restricting the infimum to atomic overlap distributions with $k+1$ atoms (or, equivalently, letting $x \mapsto q(x)$ be a step function with $k+1$ steps)
yields a sequence of upper bounds indexed by $k \geq 0$:
\begin{align}\label{eq:hierarchy_upper_bounds_phi}
     \Phi \overset{\conj}{=} \Phi_\FRSB = \lim_{k \to \infty} \Phi_{k-\RSB} \leq \cdots \leq \Phi_{k-\RSB} \leq \cdots \leq \Phi_{\ORSB} \leq \Phi_{\RS},
\end{align}
in which the ``$k-\RSB$'' functional is given by eq.~\eqref{eq:def_Phi_FRSB}, with the infimum restricted to step functions with $k+1$ steps, and we suppressed the dependence of all quantities on $\alpha$ and $\beta$ to lighten notation. Let $k^\star$ be the smallest $k$ such that $\Phi_\FRSB(\alpha,\beta) = \Phi_{k^\star-\RSB}(\alpha,\beta)$. If $1 \leq k^\star < \infty$, we say that the system is \emph{$k^\star$-th step replica symmetry breaking};
if $k^\star = 0$ the system is called \emph{replica-symmetric} (RS); if $k^\star$ does not exist the system is said to exhibit \emph{full replica symmetry breaking}.
We will clarify the meaning of ``replica symmetry breaking'' in Section~\ref{sec:upper_bounds}.
Finally, note that by using Corollary~\ref{cor:sufficient_non_injectivity} and Conjecture~\ref{conj:tightness_criterion}, 
eq.~\eqref{eq:hierarchy_upper_bounds_phi} transfers into a hierarchy of upper bounds for the injectivity threshold: 
\begin{align}\label{eq:hierarchy_upper_bounds_alpha_inj}
    \alpha_\inj \overset{\conj}{=} \alpha_\inj^\mathrm{FRSB} \leq \cdots \leq \alpha_\inj^\mathrm{(k+1)-RSB} \leq \alpha_\inj^\mathrm{k-RSB} \leq \cdots \leq \alpha_\inj^\mathrm{1RSB} \leq \alpha_\inj^\mathrm{RS},
\end{align}
in which $\alpha_\inj^{k-\RSB}$ is the value of $\alpha$ at which $f^\star_{k-\RSB}(\alpha) \coloneqq \lim_{\beta \to \infty}[-\Phi_{k-\RSB}(\alpha,\beta)/\beta]$ crosses $1$.
In particular, as we will see in Section~\ref{sec:upper_bounds}, one can compute $\alpha_\inj^\RS \simeq 7.65$.
Note that increasing $k$ and in particular going to the full-RSB solution, which conjecturally solves the problem, only takes us further from the Euler characteristic prediction $\alpha_\inj^\mathrm{Euler} \simeq 8.34$.
In Fig.~\ref{fig:chi_estar_T0} we illustrate the predictions at the RS, 1-RSB, and Full-RSB levels.

\paragraph{Proving the replica-symmetric bound --}
We can state another rigorous characterization. 
Using Gordon's min-max theorem \cite{gordon1985some,thrampoulidis2018precise}, 
we can prove that the replica-symmetric prediction is an upper bound on the injectivity threshold:
\begin{theorem}[Replica-symmetric upper bound for the injectivity threshold]\label{thm:bound_Gordon}
    \noindent
    Assume that $\alpha > \alpha_\inj^\RS \simeq 7.65$. Then $p_{m,n} \to 1$ as $n,m \to \infty$, that is, $\varphi_\bW$ is injective w.h.p.
\end{theorem}
Note that for the computation of the Gardner capacity of the so-called ``positive'' perceptron, the replica-symmetric prediction has been shown to be tight also using Gordon's inequality \cite{stojnic2013another}, because 
one can rewrite the associated min-max problem using a \emph{convex function} on \emph{convex sets}.
In the unsatisfiable phase we consider here
the solution is conjecturally full-RSB, and the replica-symmetric bound of Theorem~\ref{thm:bound_Gordon} is not expected to be tight.

\myskip
Theorem~\ref{thm:bound_Gordon} is proved in Appendix~\ref{subsec_app:bound_Gordon}. It improves upon the earlier upper bounds of Theorem~\ref{thm:known_bounds}, and it disproves the Euler characteristic based threshold prediction $\alpha_\inj^\Eul \simeq 8.34$~\cite{paleka2021injectivity,clum2022topics,clum_paleka_bandeira_mixon}. 
In Appendix~\ref{sec_app:mesh}, we provide an alternative proof of Theorem~\ref{thm:bound_Gordon}, developed during the review process of this paper, 
which leverages Gordon's ``escape through a mesh'' 
theorem~\cite{gordon1988milman}.
Due to this chronology, and because Gordon's min-max theorem underpins the ``escape through a mesh'' result while potentially being better suited for future refinements (see the discussion below),
we have retained the original proof alongside the ``escape through a mesh'' approach.
Finally, let us mention two other bounds one can obtain on the injectivity threshold.

\paragraph{The annealed bound --}
A classical approach in statistical physics to upper-bound the free entropy $\Phi(\alpha,\beta)$ is an \emph{annealed} calculation. 
Namely, one uses Jensen's inequality to write 
\begin{align*}
    \EE_\bW \Phi_n(\bW,\beta) &= \frac{1}{n} \EE_\bW \log \mcZ_n(\bW, \beta) \leq \frac{1}{n} \log \EE_\bW \mcZ_n(\bW, \beta) \stackrel{n \to \infty}{\longrightarrow} \Phi_{\annealed}(\alpha,\beta).
\end{align*}
This gives us an additional upper bound $\Phi(\alpha,\beta) \leq \Phi_\annealed(\alpha,\beta)$\footnote{However, it never improves over the replica-symmetric one, since it is a general fact that $\Phi_\RS(\alpha,\beta) \leq \Phi_\annealed(\alpha,\beta)$.} and a corresponding upper bound for the injectivity threshold $\alpha_\inj \leq \alpha_\inj^\annealed$.
We leave to the reader the exercise to show $\Phi_\annealed(\alpha,\beta) = \alpha \log [(1+e^{-\beta})/2]$. In particular, we have $ -\Phi_{\annealed}(\alpha,\beta)/\beta\to 0$ as $\beta \to \infty$ for any $\alpha>0$, and therefore $\alpha_\inj^\annealed = +\infty$: here, the result of the annealed calculation is completely uninformative. 

\paragraph{An additional lower bound --}
In Fig.~\ref{fig:chi_estar_T0} we also show in green a region that is discarded for the injectivity threshold by a non-rigorous lower bound $\alpha_\inj \geq \alpha_{\AT} \simeq 5.32$, 
based on the de Almeida-Thouless criterion \cite{de1978stability} of statistical physics. We detail its origin in Section~\ref{subsec:rs}, and its calculation in Appendix~\ref{sec_app:rs_lower_bound}.
While it is not mathematically rigorous, proving it would follow from a rigorous computation of the free entropy in the ``high-temperature'' (or small $\beta$) phase in which replica symmetry is conjectured to hold. 
In many models, this turned out to be possible to handle more easily than the complete full-RSB conjecture, so we mention it for completeness.
We also note that it improves over the lower bound $\alpha_\inj \geq 3.4$ of Theorem~\ref{thm:known_bounds}.

\begin{figure}
   \centering
\includegraphics[width=\textwidth]{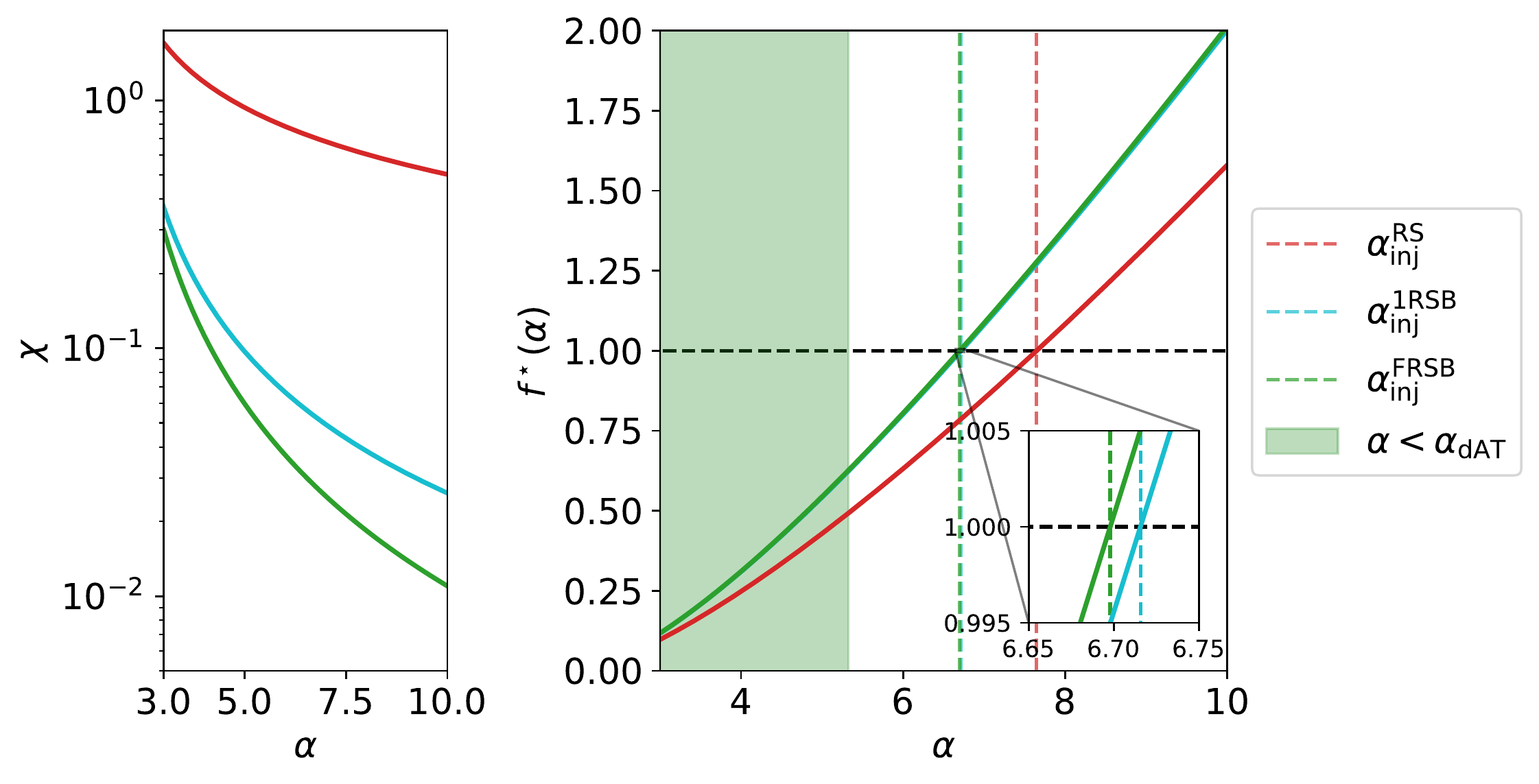}
\caption{$T=0$ limit of the RS, 1RSB and FRSB solutions, as a function of $\alpha$.
   We compare the predictions for the ground state energy $f^\star(\alpha) = \lim_{\beta \to \infty} [-\Phi(\alpha,\beta)/\beta]$ and the zero-temperature susceptibility $\chi$ (see Sections~\ref{sec:upper_bounds} and \ref{sec:full_rsb}).
   The green area is forbidden for $\alpha_\inj$ by the replica-symmetric lower bound of eq.~\eqref{eq:lower_bound_RS_stability}.
\label{fig:chi_estar_T0}}
\end{figure}

\subsection{Structure of the paper and open problems}

Section~\ref{sec:upper_bounds} has in great part a pedagogical purpose, to introduce the unacquainted reader
to the (mostly non-rigorous) results of statistical physics known in the spin glass literature under the umbrella of replica method and replica symmetry breaking. 
We detail there the replica computation in the spherical perceptron and the arising of replica symmetry breaking, 
and derive the replica symmetric and one-step replica symmetry breaking predictions for the injectivity threshold.
In Section~\ref{sec:full_rsb} we discuss the full-replica symmetry breaking prediction for the free entropy, and we derive an efficient algorithmic procedure 
to solve the zero-temperature full-RSB equations. We discuss the numerical behavior of this algorithm, and use it to derive the numerical estimate of the injectivity threshold in Result~\ref{result:frsb_result}.
As mentioned, the proofs of our rigorous results (in particular Theorem~\ref{thm:bound_Gordon}) are given in Appendix~\ref{sec_app:proofs}, 
and other analytical or numerical details and technical arguments will be deferred to the other appendices.

\myskip
Let us finally mention a few open directions that stem from our analysis.

\paragraph{Deep networks --}
A natural extension of our results would be to analyze a composition of multiple ReLU layers. Denoting still by $n$ the input dimension, Theorem~\ref{thm:bound_Gordon} guarantees 
injectivity w.h.p.\ if the size $k_L$ of the $L$-th layer satisfies $k_L > n (\alpha_\inj^\RS)^L$.
However, this is far from optimal: leveraging the structure of the image space of a ReLU layer, \cite{paleka2021injectivity,clum2022topics}
have shown that $k_L \geq n (C_1 + C_2 L \log L)$ (for some constants $C_1, C_2 > 0$) is enough to guarantee injectivity; this may be further improved using arguments based on random projections \cite{puthawala2022globally}.
An interesting open question is whether the techniques we develop here (and in particular the replica symmetry breaking framework) can be extended to predict exact injectivity transitions in the 
multi-layer case.

\paragraph{Stability of the inverse --}
While the injectivity question is limited to non-injective $\sigma$ -- such as ReLU -- in eq.~\eqref{eq:relu_layer}, 
a natural extension would be to estimate the Lipschitz constant of the inverse of $\varphi_\bW$ on its range, either 
in the injective phase we described for $\sigma = \mathrm{ReLU}$, or for any $\alpha > 0$ when $\sigma$ is injective.
Whether this question can be tackled using statistical physics tools similar to the ones we used here is an interesting open direction.

\paragraph{Improvement over Theorem~\ref{thm:bound_Gordon} --}
One can consider a closely-related model called the \emph{negative} perceptron by 
replacing $\theta(x) = \indi\{x > 0\}$ by 
$\indi\{x \geq \kappa\}$ with $\kappa < 0$ in the energy of eq.~\eqref{eq:def_energy}.
In this model, even computing the Gardner capacity conjecturally requires the full-RSB prediction. 
We note that \cite{stojnic2013negative,montanari2024tractability} have
made a refined use of Gordon's inequality to improve over the replica-symmetric upper bound for the capacity.
While similar ideas might be able to improve the upper bound of Theorem~\ref{thm:bound_Gordon}, it is not immediate to implement them, since the 
method used in \cite{montanari2024tractability} relies on the min-max problem being formulated over unit-norm vectors, which is not the case here:
we give more details on this point in Appendix~\ref{subsec_app:improvement_gordon}.
Since Theorem~\ref{thm:bound_Gordon} already allows to disprove the Euler characteristic prediction, we leave such an improvement for later work.

\paragraph{Large deviations of sublevel sets --}
The non-validity of the average Euler characteristic prediction also leads to interesting 
predictions on the energy landscape of the perceptron. Indeed, the quantity $q_{m,n}$ of eq.~\eqref{eq:def_qmn} 
is the mean Euler characteristic of a sublevel set $U$ of the perceptron, more precisely
$q_{m,n} = \EE_\bW[\chi(U)]$, with $U \coloneqq \{\bx \in \mcS^{n-1} \, : \, e_\bW(\bx) \leq 1\}$.
Recall that $\alpha_\inj^\Eul \simeq 8.34$ while $\alpha_\inj^\RS \simeq 7.65$.
According to Theorem~\ref{thm:bound_Gordon}, for all $\alpha \in (\alpha_\inj^\RS,\alpha_\inj^\Eul)$ (and conjecturally in $(\alpha_\inj^\FRSB,\alpha_\inj^\Eul)$) the set $U$ is typically empty: 
however its average Euler characteristic is exponentially large!
A possible explanation for this discrepancy is that there exist large deviations events with probability $\exp(-n I_1)$ in which the set $U$ is not only non-empty, but has Euler characteristic $\exp(n I_2)$. 
A natural conjecture is that $I_2 > I_1$ for $\alpha < \alpha_\inj^\Eul$ and $I_2 < I_1$ for 
$\alpha > \alpha_\inj^\Eul$. 
Exploring further these large deviations could thus explain the error made in the Euler characteristic approach.

\paragraph{Numerical code and reproducibility --}
All figures and numerical results in this paper are fully reproducible. The JAX \cite{jax2018github} code is available in a \href{https://github.com/AnMaillard/Injectivity_ReLu_layer}{GitHub repository} \cite{github_repo}.

\paragraph{Acknowledgments --} 
We would like to thank the anonymous referees for their feedback, which led in particular to the alternative proof of Theorem~\ref{thm:bound_Gordon} presented in Appendix~\ref{sec_app:mesh}.
A.M.\ and A.B.\ thank D.\ Paleka, C.\ Clum, and D.\ Mixon for several discussions related to this paper. A.M.\ is grateful to F.\ Krzakala, L.\ Zdeborov\'a, B.\ Loureiro and P.\ Urbani for insightful discussions. 
I.D.\ acknowledges support by the European Research Council (ERC) Starting Grant 852821---SWING.

\section{The replica hierarchy of upper bounds}\label{sec:upper_bounds}
\subsection{General principles of the replica method}

The replica method is based on the \emph{replica trick}, a heuristic use of the following formula, for any random variable $X > 0$ (assuming that all the moments written hereafter are well-defined):
\begin{align}\label{eq:replica_trick}
    \EE \log X &= \lim_{r \to 0} \frac{\EE X^r - 1}{r} = \frac{\partial}{\partial r} [\log \EE X^r]_{r = 0}.
\end{align}
While the replica trick is most often described as the first equality in eq.~\eqref{eq:replica_trick}, we will here use the second (and equivalent) equality.
Assuming that 
$\Phi(\alpha, \beta) \coloneqq \lim_{n \to \infty} \EE_\bW  \, \Phi_n(\bW, \beta)$ is well defined, we reach:
\begin{align}\label{eq:replica_trick_phi}
    \Phi(\alpha,\beta) &= \lim_{n \to \infty} \frac{\partial}{\partial r} \Big[\frac{1}{n} \log \EE_\bW \big\{\mcZ_n(\bW, \beta)^r\big\}\Big]_{r = 0}.
\end{align}
So far, eq.~\eqref{eq:replica_trick_phi} is not really surprising.
The replica method is based on several heuristics, and leverages the fact that it is often possible to compute the RHS of eq.~\eqref{eq:replica_trick_phi} 
for \emph{integer $r$}. More precisely, the replica method proceeds as follows: 

\myskip
\fbox{\begin{minipage}{0.98\textwidth}
    \textbf{Replica method}
    \begin{itemize}[leftmargin=25pt]
        \item[$(i)$] Assume that the limits $n \to \infty$ and $r \to 0$ can be inverted in eq.~\eqref{eq:replica_trick_phi}, 
        i.e.\ that we have $\Phi(\alpha,\beta) = \partial_r [\Phi(\alpha,\beta;r)]_{r = 0}$, with  
        \begin{align}\label{eq:def_Phir}
            \Phi(\alpha,\beta;r) &\coloneqq \lim_{n \to \infty} \frac{1}{n} \log \EE_\bW \big\{\mcZ_n(\bW, \beta)^r\big\}.
        \end{align}
        \item[$(ii)$] Compute $\Phi(\alpha,\beta;r)$ for \emph{integer $r$}, i.e.\ the asymptotics of the moments of $\mcZ_n(\bW, \beta)$.
        \item[$(iii)$] Use these values to analytically expand $\{\Phi(\alpha,\beta;r)\}_{r \in \bbN}$ to all $r \geq 0$.
        \item[$(iv)$] Compute $\Phi(\alpha,\beta) = \partial_r [\Phi(\alpha,\beta;r)]_{r = 0}$ from the analytic continuation above.
    \end{itemize}
\end{minipage}}

\myskip
Note that step $(i)$, although \emph{a priori} non-rigorous, can sometimes be put on rigorous ground using convexity arguments, cf.\ e.g.\ page 146 of \cite{talagrand2010mean} in the context of the Sherrington-Kirkpatrick (or SK) model.
The arguably ``most heuristic'' step is $(iii)$, as there is in general no guarantee for the uniqueness of the analytic continuation (and it is often not unique!).
The choice of the conjecturally correct continuation was proposed by Parisi in a remarkable series of papers \cite{parisi1979infinite,parisi1980order,parisi1980sequence}, one of the most important contributions for which he earned a Nobel prize in Physics in 2021, and we will describe this choice in the following sections.
In the SK model originally studied by Parisi
his prediction was ultimately proven to be correct by Talagrand \cite{talagrand2006parisi} and generalized by Panchenko \cite{panchenko2014parisi}, 
leveraging notably interpolation techniques that originated with Guerra \cite{guerra2003broken}.
An actual rigorous treatment of the replica method itself remains out of reach, 
and in the spherical perceptron considered here replica predictions have not been proven, with the exception of the satisfiable phase \cite{shcherbina2003rigorous,stojnic2013another}.

\subsection{First steps of the replica method}

Let us now perform step $(ii)$ of the replica method.
From now on, we relax the level of rigor and sometimes adopt notations closer to the theoretical physics literature, since the core of the method is heuristic.
Fixing $r \in \bbN^\star$ we have (recall eq.~\eqref{eq:def_phi}):
\begin{align}\label{eq:Phir}
    \nonumber
   \Phi(\alpha,\beta;r) &= \lim_{n \to \infty}\frac{1}{n} \log \EE_\bW \Bigg\{\Bigg(\int_{\mcS^{n-1}} \mu_n(\rd \bx) \, e^{-\beta E_\bW(\bx)}\Bigg)^r\Bigg\}, \\
                    &= \lim_{n \to \infty}\frac{1}{n} \log \int \prod_{a=1}^r \mu_n(\rd \bx^a) \, \EE_\bW \Bigg\{ \prod_{a=1}^r \exp\Big\{- \beta \sum_{\mu=1}^m \theta\Big[(\bW \bx^a)_\mu\Big]\Big\} \Bigg\}.
\end{align}
We have used Fubini's theorem in eq.~\eqref{eq:Phir}. 
We see appearing a set $\{\bx^a\}_{a=1}^r$ of independent samples from the Gibbs measure $\bbP_{\beta,\bW}$, with the \emph{same} realization of the matrix $\bW$:
we call such independent samples \emph{replicas}, following the statistical physics nomenclature.
The expectation with respect to $\bW$ in eq.~\eqref{eq:Phir} can be performed, since at fixed $\{\bx^a\}$, $\bz^a \coloneqq \bW \bx^a$ are jointly Gaussian vectors
with covariance $\EE[z^a_\mu z^b_\nu] = \delta_{\mu \nu} Q^{ab}$, where we introduced the \emph{overlap matrix} $Q^{ab} \coloneqq \bx^a \cdot \bx^b$ (note that $Q^{aa} = 1$, $Q^{ab}=Q^{ba}$, and that the matrix $\{\bx^a \cdot \bx^b\}$ is almost surely invertible under $\mu_n^{\otimes r}$).
Therefore we have:
\begin{align*}
    \EE_\bW \prod_{a=1}^r e^{- \beta \sum_{\mu=1}^m \theta[(\bW \bx^a)_\mu]} &= I_\beta(\bQ)^m,
\end{align*}
in which we defined
\begin{align}\label{eq:def_Ibeta_Q}
    I_\beta(\bQ) \coloneqq \int_{\bbR^r} \frac{\rd \bz}{(2\pi)^{r/2} \sqrt{\det \bQ}} e^{-\frac{1}{2} \bz^\intercal \bQ^{-1} \bz} e^{-\beta \sum_{a=1}^r \theta(z^a)}.
\end{align}
One can thus write eq.~\eqref{eq:Phir} as:
\begin{align*}
   \Phi(\alpha,\beta;r) &= \lim_{n \to \infty} \frac{1}{n} \log \Bigg[\int \Big\{\prod_{a< b} \rd Q^{ab}\Big\} J(\bQ) \times I_\beta(\bQ)^m \Bigg],
\end{align*}
with $J(\bQ)$ defined as the PDF of the overlap matrix $\bQ(\{\bx^a\})$ (for $\{\bx^a\} \sim \mu_n^{\otimes r}$) evaluated in $\bQ$:
\begin{align}\label{eq:def_JQ}
   J(\bQ) &\coloneqq \int \prod_{a=1}^r \mu_n(\rd \bx^a) \prod_{a < b} \delta(Q^{ab} - \bx^a \cdot \bx^b) = n^{\frac{r(r-1)}{2}} \frac{\int \prod_{a=1}^r \rd \bx^a \prod_{a \leq b} \delta(n Q^{ab} - \bx^a \cdot \bx^b)}{\int \prod_{a=1}^r \rd \bx^a  \,\delta(n - \| \bx^a \|^2)},
\end{align}
in which we used that $Q^{aa} = 1$ and we re-normalized $\bx^a$ by $\sqrt{n}$.
One way to compute the numerator in eq.~\eqref{eq:def_JQ} is to use an exponential tilting method, by the following argument: for any symmetric $\bLambda \in \bbR^{r \times r}$ positive-definite, we have
\begin{align}\label{eq:denominator_JQ}
    \nonumber
    &\frac{1}{n} \log \int \prod_{a=1}^r \rd \bx^a \prod_{a \leq b} \delta(n Q^{ab} - \bx^a \cdot \bx^b) \\
    &= \frac{1}{2} \Tr[\bLambda \bQ] + \frac{1}{n} \log \int \prod_{a=1}^r \rd \bx^a  \, 
    \prod_{a \leq b} \delta(n Q^{ab} - \bx^a \cdot \bx^b)
    \, e^{-\frac{1}{2} \sum_{a,b} \Lambda^{ab} \bx^a \cdot \bx^b}.
\end{align}
The idea is to pick $\bLambda$ so that under the probability distribution $P_\bLambda(\{\bx^a\}) \propto \exp\{-\frac{1}{2} \sum_{a,b} \Lambda^{ab} \bx^a \cdot \bx^b\}$, we have 
with high probability $\bx^a \cdot \bx^b /n \to Q^{ab}$ as $n \to \infty$.
Since $P_\bLambda$ is Gaussian, one finds $\bLambda = \bQ^{-1}$ as the correct choice.
Heuristically, the argument then goes as follows: for $\bLambda = \bQ^{-1}$, the constraint terms in eq.~\eqref{eq:denominator_JQ} are satisfied as $n \to \infty$, so that we can remove the Dirac deltas without affecting the asymptotic value of the integral.
Performing the same calculation in the denominator (for which $\bLambda = \Id_r$ is now the correct choice), one reaches:
\begin{align}\label{eq:JQ_1}
    \frac{1}{n} \log J(\bQ) &= \frac{1}{2} \Tr[\bQ^{-1} \bQ] + \log \int_{\bbR^r} \prod_{a=1}^r \rd x^a \, e^{-\frac{1}{2} \sum_{a,b} (\bQ^{-1})^{ab} x^a x^b} - \frac{r(1+\log 2\pi)}{2} + \smallO_n(1).
\end{align}
Such ``exponential tilting'' arguments can be made rigorous, and are classical e.g.\ in the theory of large deviations \cite{dembo1998large}.
Another (equivalent) way to obtain eq.~\eqref{eq:JQ_1} is to introduce the Fourier transform of the Dirac delta in eq.~\eqref{eq:def_JQ}, and perform a saddle-point method over the parameters of the Fourier integral, see e.g.\ \cite{castellani2005spin} or \cite{urbani2018statistical}.
The Gaussian integral in eq.~\eqref{eq:JQ_1} can be computed:
\begin{align*}
    \frac{1}{n} \log J(\bQ) &= \frac{1}{2} \log \det \bQ + \smallO_n(1).
\end{align*}
This yields:
\begin{align}\label{eq:Phir_before_Laplace}
   \Phi(\alpha,\beta;r) &= \lim_{n \to \infty} \frac{1}{n} \log \Bigg[\int \Big\{\prod_{a< b} \rd Q^{ab}\Big\} \exp \{n F_n(\bQ)\} \Bigg],
\end{align}
with  
\begin{align*}
    F_n(\bQ) &\coloneqq \frac{1}{2} \log \det \bQ + \alpha \log I_\beta(\bQ) + \smallO_n(1).
\end{align*}
It is crucial that in many physical models, the average of the replicated partition function can be written as in 
eq.~\eqref{eq:Phir_before_Laplace}, as a function of a low-dimensional parameter (recall that $\bQ$ is a $r \times r$ matrix, and that $r$ is a fixed positive integer).
In physics, one refers to the overlap matrix $\bQ$ as the \emph{order parameter} of the problem: a low-dimensional quantity that allows to characterize the macroscopic
behavior of our high-dimensional system (similarly to the average magnetization in a ferromagnet for instance).

\myskip
Applying Laplace's method to the integral in eq.~\eqref{eq:Phir_before_Laplace}, we finally reach:
\begin{align}\label{eq:Phir_final}
    \Phi(\alpha,\beta;r) &= \sup_{\bQ} \Big[\frac{1}{2} \log \det \bQ + \alpha \log I_\beta(\bQ)\Big],
\end{align}
where the supremum is over $r \times r$ symmetric positive-definite matrices such that $Q^{aa} = 1$, and 
recall that $I_\beta(\bQ)$ is defined in eq.~\eqref{eq:def_Ibeta_Q}.
Note that we completely removed the high dimensionality of the problem! 
The remaining task is to perform step $(iii)$ of the replica method, i.e.\ to analytically continue $\Phi(\alpha,\beta,r)$ to any $r > 0$.
This is the crucial difficulty of the replica method (and the main reason why it is ill-posed mathematically in general), which was solved by Parisi \cite{parisi1979infinite,parisi1980order,parisi1980sequence}.

\subsection{The replica-symmetric solution}\label{subsec:rs}

The functional in eq.~\eqref{eq:Phir_final} is symmetric: one can permute the different replicas of the systems (and correspondingly swap the rows and columns of $\bQ$) 
without changing the value of the functional. This has led physicists to  first assume that the supremum in eq.~\eqref{eq:Phir_final} is attained by a matrix $\bQ$ that is also invariant under permutations, i.e.\ that satisfies $Q^{ab} = q$ for all $a \neq b$.
This \emph{replica-symmetric} assumption was historically the first one considered to find a solution 
to the SK model \cite{sherrington1975solvable}.
We will see how it allows to complete the final steps of the replica method. 

\myskip 
Note that replica symmetry can be put on a firmer mathematical ground, using the following characterization.

\myskip
\fbox{\begin{minipage}{0.98\textwidth}
    \textbf{Replica symmetry (RS) --} 
    Let $Q(\bx, \bx') \coloneqq \bx \cdot \bx'$, 
    and recall the Gibbs distribution $\bbP_{\beta,\bW}$ of eq.~\eqref{eq:def_Gibbs}.
    Replica symmetry amounts to assuming that the random variable $Q(\bx,\bx')$ concentrates when $\bx,\bx'$ are sampled independently from $\bbP_{\beta,\bW}$ (with the \emph{same} $\bW$), in the following sense:
    \begin{equation}\label{eq:overlap_concentration}
        \lim_{n \to \infty} \EE_\bW \Big[ \EE_{(\bx,\bx') \sim \bbP_{\beta,\bW}^{\otimes 2}} \Big\{(Q(\bx, \bx') - \EE  Q )^2\Big\}\Big] = 0,
    \end{equation}
    with the shorthand $\EE Q \coloneqq \EE_\bW [\EE_{(\bx,\bx') \sim \bbP_{\beta,\bW}^{\otimes 2}}(Q(\bx, \bx'))]$.
\end{minipage}}

\myskip
In particular, under the RS ansatz, we can write the off-diagonal elements of the overlap matrix appearing in eq.~\eqref{eq:Phir_final} 
as $Q^{ab} = q = \EE_\bW [\EE_{(\bx,\bx') \sim \bbP_{\beta,\bW}^{\otimes 2}}(\bx \cdot \bx')]$, in which $\bx, \bx'$ are two independent samples under the Gibbs measure with 
quenched noise $\bW$ (two \emph{replicas} of the system), and $a \neq b$.
. Therefore, we also have 
$q = \EE_\bW [\norm{ \EE_{\bx \sim \bbP_{\beta,\bW}}(\bx)}^2]$, which implies in particular that $q \in [0,1]$.

\myskip 
Let us now finish the replica calculation under a replica symmetric assumption, going back to eq.~\eqref{eq:Phir_final}.
By simple linear algebra calculations, the RS ansatz implies, for all $a \neq b$:
\begin{align*}
    \begin{dcases}
        Q^{-1}_{ab} &= - \frac{q}{(1-q)[1+(r-1)q]},\\
        Q^{-1}_{aa} - Q^{-1}_{ab} &= \frac{1}{1-q},
    \end{dcases}
\end{align*}
and moreover 
\begin{align}\label{eq:detQ_RS}
    \log \det \bQ &= (r-1) \log [1-q] + \log [1+(r-1)q].
\end{align}
Plugging the form of $\bQ^{-1}$ we have:
\begin{align}
\label{eq:Gauss_weight_RS}
\nonumber
    \exp\Big\{-\frac{1}{2} \bz^\intercal \bQ^{-1} \bz\Big\} &= 
\exp\Big\{-\frac{1}{2(1-q)} \sum_{a=1}^r (z^a)^2 + \frac{q}{2 (1-q)[1+(r-1)q]} \Big(\sum_{a=1}^r z^a\Big)^2 \Big\} \\ 
&= \int \mcD \xi \exp\Big\{-\frac{1}{2(1-q)} \sum_{a=1}^r (z^a)^2 + \sqrt{\frac{q}{(1-q)[1+(r-1)q]}}  \Big(\sum_{a=1}^r z^a\Big) \xi \Big\}.
\end{align}
Recall that $\mcD \xi$ is the standard Gaussian measure on $\bbR$, 
and we have used the identity $\exp(x^2/2) = \int \mcD \xi \exp(x \xi)$.
Plugging eqs.~\eqref{eq:detQ_RS} and \eqref{eq:Gauss_weight_RS} in eq.~\eqref{eq:Phir_final} we reach, for $r \in \bbN^\star$: 
\begin{align}\label{eq:phir_RS}
   \Phi_\mathrm{RS}(\alpha,\beta;r) &= \sup_{q \in [0,1]} \Phi_\mathrm{RS}(\alpha,\beta;r, q) = \sup_{q \in [0,1]} \Big[- \frac{\alpha-1}{2} \Big[(r-1) \log [1-q] + \log [1+(r-1)q]\Big] \\ 
   &\nonumber + \alpha \log \int \mcD \xi \Bigg[\int_{\bbR} \frac{\rd z}{\sqrt{2\pi}} e^{-\frac{1}{2(1-q)} z^2 + \sqrt{\frac{q}{(1-q)[1+(r-1)q]}} z \xi -\beta \theta(z)} \Bigg]^r \Big].
\end{align}
One can now begin to see how the replica-symmetric ansatz yields an analytical continuation of $\Phi_\mathrm{RS}(\alpha,\beta;r)$ for all $r > 0$. 
A final non-trivial (and non-rigorous) technicality of the replica method, which we do not detail here, is that when we analytically expand the function above to $r < 1$,
theoretical physicists argue that maximizers of eq.~\eqref{eq:phir_RS} for $r > 1$ are continued into \emph{minima} of $\Phi_\RS(r,q)$ for $r < 1$. 
We refer to \cite{mezard1987spin} for a detailed discussion: 
in the present case this phenomenon can be easily observed numerically, see Fig.~\ref{fig:inversion_min_max}.
\begin{figure}
   \centering
\includegraphics[width=0.5\textwidth]{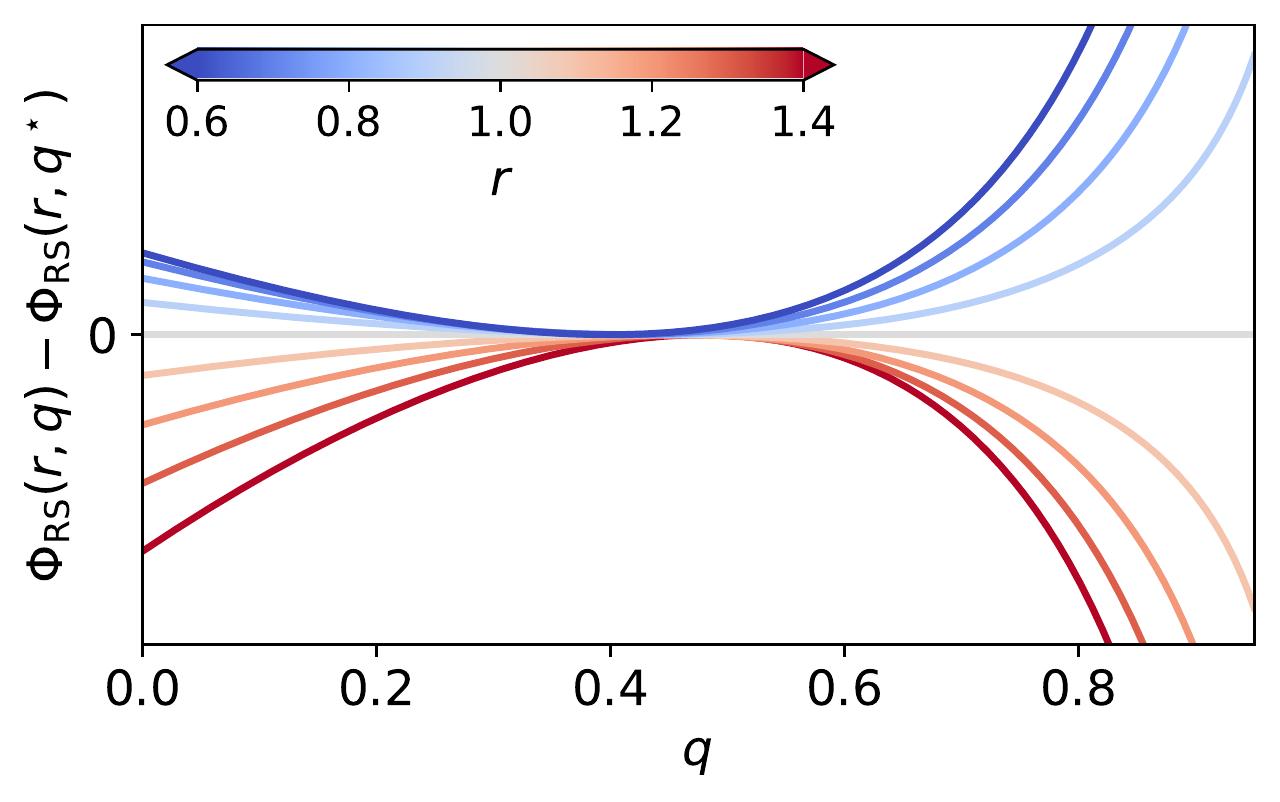}
  \caption{ 
    The function $\Phi_\RS(r, q) - \Phi_\RS(r, q^\star(r))$ as a function 
    of $q \in [0,1]$, for different values of $r$ close to $1$, and $q^\star(r)$ the unique solution to $\partial_q \Phi_\RS(r, q) = 0$. We observe that $q^\star(r)$
    is a global maximum for $r > 1$, and becomes a global minimum for $r < 1$.
    Here $\alpha = 5$ and $\beta = 1$.
        \label{fig:inversion_min_max}}
\end{figure}
Under the replica-symmetric ansatz, we therefore obtain:
\begin{align*}
    \Phi_\mathrm{RS}(\alpha,\beta) &= \partial_r [\Phi_\mathrm{RS}(\alpha,\beta;r)]_{r=0} \nonumber \\ 
    &= \inf_{q \in [0,1]} \Big\{
    - \frac{\alpha-1}{2} \Big[\log [1-q] + \frac{q}{1-q}\Big]
    + \alpha \int \mcD \xi \log \int \frac{\rd z}{\sqrt{2\pi}} e^{-\frac{1}{2(1-q)} z^2 + \frac{\sqrt{q}}{1-q} z \xi -\beta \theta(z)} \Big\}.
\end{align*}
The inner integral is easy to work out:
\begin{align*}
     \int_{\bbR} \frac{\rd z}{\sqrt{2\pi}} e^{-\frac{1}{2(1-q)} z^2 + \frac{\sqrt{q}}{1-q} z \xi -\beta \theta(z)} &= \sqrt{1-q} \, e^{\frac{q}{2(1-q)} \xi^2} \Big[1-(1-e^{-\beta}) H\Big(-\xi \sqrt{\frac{q}{1-q}} \Big) \Big],
\end{align*}
where $H(x) \coloneqq \int_x^\infty \mcD u = [1 - \mathrm{erf}(x/\sqrt{2})]/2$. In particular, $H'(x) = - e^{-x^2/2}/\sqrt{2 \pi}$.
Then:
\begin{align}\label{eq:phi_RS}
   \Phi_\RS(\alpha,\beta) &= \inf_{q \in [0,1]} \Big\{ \frac{1}{2} \Big[\log [1-q] + \frac{q}{1-q}\Big]+ \alpha \int \mcD \xi \log \Big[1-(1-e^{-\beta}) H\Big(\xi \sqrt{\frac{q}{1-q}} \Big) \Big] \Big\}.
\end{align}
For any $\beta \geq 0$, the minimizing $q$ is thus given by the solution to
\begin{align}\label{eq:q_RS_eq_new}
    \frac{q^{3/2}}{\sqrt{1-q}} &= 
    \alpha \int \mcD \xi \frac{(1-e^{-\beta}) \xi H'\Big(\xi \sqrt{\frac{q}{1-q}} \Big)}{1-(1-e^{-\beta}) H\Big(\xi \sqrt{\frac{q}{1-q}} \Big)}
\end{align}
that minimizes the functional of eq.~\eqref{eq:phi_RS}.
The quantity $e^\star(\alpha,\beta) \coloneqq - \partial_\beta \Phi_\RS(\alpha,\beta)$ is called the \emph{average intensive energy}: 
as can be seen from eq.~\eqref{eq:def_phi}, $n e^\star(\alpha,\beta)$ is the average number of negative components of 
$\bx$ when sampled from the Gibbs measure of eq.~\eqref{eq:def_Gibbs}.
At the replica-symmetric level it is given by:
\begin{align}\label{eq:estar_beta}
    e_\RS^\star(\alpha,\beta) &= \alpha e^{-\beta} \int_\bbR \mcD \xi \frac{H\Big(\xi \sqrt{\frac{q}{1-q}} \Big)}{1-(1-e^{-\beta}) H\Big(\xi \sqrt{\frac{q}{1-q}} \Big) }.
\end{align} 
In particular, one sees that for $\beta = 0$ we have $q = 0$ and $e^\star(\alpha,\beta = 0) = \alpha/2$, which is the typical number of negative components of a random $m$-dimensional vector (divided by $n$).

\myskip
\textbf{The zero-temperature limit --}
For any $\alpha > 2$ (i.e.\ in the UNSAT phase), one can check from eq.~\eqref{eq:q_RS_eq_new} that $q \to 1$ as $\beta \to \infty$.
This means that the replica-symmetric ansatz predicts that, as $\beta \to \infty$, the Gibbs measure concentrates on the global minima of $E_\bW(\bx)$, 
and that (at fixed $\bW$) the distance between any two such minima goes to $0$ as $n \to \infty$\footnote{As we will see in Sec~\ref{sec:full_rsb}, while the replica-symmetry assumption turns out to be wrong, this prediction remains correct!}.
One can see also from this equation (cf.\ \cite{gardner1988optimal,franz2017universality}) that the expansion of the solution $q$ is of the type:
\begin{align}\label{eq:q_RS_zerotemp}
    q &= 1 - \frac{\chi_\RS}{\beta} + \mathcal{O}(\beta^{-2}),
\end{align}
where $\chi_\RS$ is the so-called zero-temperature susceptibility.
Plugging this expansion in the equations above, we recover the result of \cite{gardner1988optimal} (we detail the computations in Appendix~\ref{subsec_app:zerotemp_RS}). 
We find that $\chi_\RS$ is the unique solution to:
\begin{align}\label{eq:chi_RS}
   \alpha \int_0^{\sqrt{2\chi_\RS}} \mcD \xi \, \xi^2 &= 1,
\end{align}
and $f^\star_\RS(\alpha) = \lim_{\beta \to \infty} [-\Phi_\RS(\alpha,\beta) / \beta]$ is given as (recall $H(x) = \int_x^\infty \mcD u$):
\begin{align}\label{eq:fstar_RS}
    f^\star_\RS(\alpha) &= \alpha H[\sqrt{2 \chi_\RS}].
\end{align}

\myskip
\textbf{Replica-symmetric prediction for $\alpha_\inj$ --}
Recall the criterion of eq.~\eqref{eq:criterion_alphainj} for the injectivity threshold.
Eqs.~\eqref{eq:chi_RS} and \eqref{eq:fstar_RS} 
are easy to analyze numerically, and they yield that $f^\star_\RS(\alpha) = 1$ for:
\begin{align}\label{eq:alpha_inj_RS}
    \alpha_\inj^\mathrm{RS} &\simeq 7.64769, 
\end{align}
in which $\mathrm{RS}$ stands for the replica-symmetric assumption.

\myskip
\textbf{Instability of the replica-symmetric solution and the need for a different ansatz --} 
An important check of the validity of the replica-symmetric ansatz is that it indeed is a maximum of the functional given in eq.~\eqref{eq:Phir_final} (or a minimum when $r < 1$ as we discussed).
This can be verified locally, by considering the Hessian of this function, and looking at the sign of its eigenvalues when $r \to 0$. The stability criterion
is called the \emph{de Almeida-Thouless} (dAT) condition \cite{de1978stability}, and we derive it in Appendix~\ref{subsec_app:stability_rs} for any inverse temperature $\beta \geq 0$, cf.\ eq.~\eqref{eq:AT_explicit}. 
However, we also show that this condition is never satisfied in the limit $\beta \to \infty$, for any $\alpha > 2$.
This suggests that the correct solution actually breaks the replica symmetry! 
Formally, the functional of eq.~\eqref{eq:Phir_final} exhibits a well-known physical phenomenon known as 
spontaneous symmetry breaking: while the function to maximize is invariant under the group of permutations of the $r$ replicas, any particular maximum is not invariant under this symmetry.

\myskip\textbf{A replica-symmetric lower bound --}
In Appendix~\ref{sec_app:rs_lower_bound}, we detail a way to use the replica-symmetric prediction at finite $\beta \geq 0$, combined with the stability 
analysis of Appendix~\ref{subsec_app:stability_rs}, to obtain a lower bound on $\alpha_\inj$:
\begin{align}\label{eq:lower_bound_RS_stability}
    \alpha_\inj \geq \alpha_\AT\simeq 5.3238,
\end{align}
in which the definition and calculation of $\alpha_\AT$ can be deduced solely from the replica-symmetric calculation at high enough temperature (low enough $\beta$).
High temperature replica-symmetric regimes in spin glasses are notoriously easier to analyze mathematically than low-temperature settings~\cite{talagrand2010mean}:
for this reason, while we do not provide here a mathematical proof of eq.~\eqref{eq:lower_bound_RS_stability}, we expect 
it to be easier to establish rigorously than bounds from replica symmetry breaking theory, and we leave a proof of eq.~\eqref{eq:lower_bound_RS_stability} as an interesting open problem.
We refer to Appendix~\ref{sec_app:rs_lower_bound} for more details on this bound, which is 
shown as a light green area in Fig.~\ref{fig:chi_estar_T0}.

\subsection{The overlap distribution and replica symmetry breaking}\label{subsec:rsb_discussion}

Since we must go beyond replica symmetry, one has to understand what could happen if the overlap concentration of eq.~\eqref{eq:overlap_concentration}
is not satisfied.
We define $q \equiv \bx \cdot \bx'$, in which $\bx,\bx'$ are independent samples under the Gibbs measure of eq.~\eqref{eq:def_Gibbs}, with the \emph{same quenched noise} $\bW$, 
and we will study the law of $q$ \emph{averaged over $\bW$}, which we will denote $\rho_n(q)$.

\myskip
A natural possibility is that, while the random variable $q$ no longer concentrates, its average distribution $\rho_n(q)$ still 
converges (weakly) to an asymptotic law $\rho(q)$ (for $q \in [0,1]$) as $n \to \infty$.
Replica-symmetry then corresponds to the case $\rho(q) = \delta(q - q_0)$.
But how does an arbitrary $\rho(q)$ transfers to a $r \times r$ overlap matrix $\bQ$ maximizing eq.~\eqref{eq:Phir_final}?
Actually, the other way (going from $\bQ$ to $\rho(q)$) is easier to formalize. Indeed, for the same $\bW$, let us draw two independent samples $\bx, \bx'$ under the Gibbs measure (two ``replicas'').
On average, their overlap is distributed as the off-diagonal elements of the overlap matrix, i.e.\ we have (one can formalize this argument, see e.g.\ \cite{montanari2024friendly})
\begin{align*}
    \rho(q) \simeq \frac{1}{r(r-1)}\sum_{a \neq b} \delta(q - Q^{ab}).
\end{align*}
However, recall that our physical system is not represented by the overlap matrix $\bQ$ at finite $r$, but rather by its $r \to 0$ limit, 
so we should take this limit as well to get the $\rho(q)$ that describes our original physical system (even though taking the $r \to 0$ limit of a $r \times r$ matrix shatters much of our intuition!).
More concretely, the overlap distribution $\rho(q)$ is related to the overlap matrix $\bQ$ by:
\begin{align}\label{eq:rho_q_general}
    \rho(q) &= \lim_{r \to 0} \frac{1}{r(r-1)}\sum_{a \neq b} \delta(q - Q^{ab}).
\end{align}
\textbf{One-step replica symmetry breaking --}
To build back our intuition a bit, let us look at the simplest possible $\rho(q)$ beyond the RS ansatz, that is,  let us assume that 
$\rho(q) = m \delta(q-q_0) + (1-m) \delta(q-q_1)$, with $m \in [0,1]$, and $q_0 \leq q_1$.
One brilliant realization of Parisi \cite{parisi1979infinite,parisi1980order,parisi1980sequence} was that this distribution arises from an \emph{ultrametric} overlap matrix $\bQ$, i.e.\ 
that has the following form: 
\begin{align}\label{eq:rhoq_Q_1RSB}
    \rho(q) = m \delta(q - q_0) + (1-m) \delta(q-q_1) \quad ``\Longleftrightarrow" \quad \bQ = \begin{pmatrix}
      1 & q_1 & q_1 && & & \\ 
      q_1 & 1 & q_1 && \cdots & q_0 & \cdots &\\ 
      q_1 & q_1 & 1 && & &\\
       &  & &\ddots & & &\\
       & & &  &1 &q_1 & q_1 \\  
      \cdots & q_0 & \cdots && q_1 & 1 & q_1 \\  
      & &  && q_1 & q_1 & 1 
    \end{pmatrix}.
\end{align}
Let us detail how to go from the $\bQ$ shown in eq.~\eqref{eq:rhoq_Q_1RSB} to the $\rho(q)$ that we want.
We denote $x \in \{1,\cdots,r\}$ the size of the diagonal blocks in this matrix $\bQ$. Then:
\begin{align}\label{eq:Pq_1RSB_finite_r}
    \frac{1}{r(r-1)}\sum_{a \neq b} \delta(q - Q^{ab}) &= \frac{x-1}{r-1} \delta(q-q_1) + \frac{r-x}{r-1} \delta(q-q_0).
\end{align}
Now arises an issue: since we take the $r \downarrow 0$ limit, and $x \in \{1,\cdots,r\}$ is an integer, how should we proceed? 
Comparing eq.~\eqref{eq:Pq_1RSB_finite_r} with our target $\rho(q)$ gives us a possible answer (which turns out to be the correct one \cite{mezard1987spin}): relaxing the constraint that $x \in \{1,\cdots,r\}$, and taking the limit $r \downarrow 0$ independently of $x$, we reach: 
\begin{align*}
    \rho(q) &= (1-x) \delta(q-q_1) + x \delta(q-q_0),
\end{align*}
i.e.\ exactly the $\rho(q)$ we wanted to build, with $x = m \in [0,1]$ which now became a real parameter in $[0,1]$. 
This type of distribution $\rho(q)$ (and by extension the corresponding $\bQ$ in eq.~\eqref{eq:rhoq_Q_1RSB}) 
is called One-Step Replica Symmetry Breaking (1RSB).

\myskip
\textbf{General replica symmetry breaking --}
More generally, one can represent a distribution with a finite support of $(k+1)$ elements as $\rho(q) = \sum_{i=0}^k (m_i - m_{i-1}) \delta(q-q_i)$, with weights 
$m_0 \leq m_1 \leq \cdots \leq m_{k-1} \leq m_k$, using the conventions $m_{-1} = 0, m_k = 1$.
This distribution is called ``$k$-step replica symmetry breaking'' ($k$-RSB),
and in this ansatz, the overlap matrix $\{Q_{ab}\}$ can be written as a hierarchical generalization of eq.~\eqref{eq:rhoq_Q_1RSB} (with the convention $q_{-1} = 0$ and $q_{k+1} = 1$):
\begin{align*}
    \bQ &= \sum_{i=0}^{k+1} (q_i - q_{i-1}) \bJ^{(r)}_{m_{i-1}}, 
\end{align*}
with $\bJ_m^{(r)}$ the block-diagonal matrix with $r/m$ blocks of size $m$, each diagonal block being the all-ones matrix.
Once again, the integers $\{m_i\}_{i=0}^k$ become elements of $[0,1]$ in the $r \downarrow 0$ limit. 
As in the replica-symmetric case discussed above, the limit $r \downarrow 0$ also turns the maximum over $\{m_i,q_i\}$ into an infimum \cite{mezard1987spin}.
In the end, the $k$-RSB prediction for the free entropy is of the form:
\begin{align}\label{eq:Phi_kRSB_general}
    \Phi_{k-\mathrm{RSB}}(\alpha,\beta) &= \inf_{0\leq q_0 \leq \cdots \leq q_k\leq 1}  \, \inf_{0 < m_0 < \cdots m_{k-1} < m_k = 1} \mcP[\{m_i\}, \{q_i\};\alpha,\beta].
\end{align}
It is common to represent the right hand side as a function of a step function $q(x)$ for $x \in [0,1]$, uniquely defined by $\{m_i\}$ and $\{q_i\}$, see Fig.~\ref{fig:q_frsb} (left, blue curve). 
We write then the argument of the RHS of eq.~\eqref{eq:Phi_kRSB_general} as $\mcP[\{q(x)\};\alpha,\beta]$.
\begin{figure}[t]
    \centering
	\includegraphics[width=\textwidth]{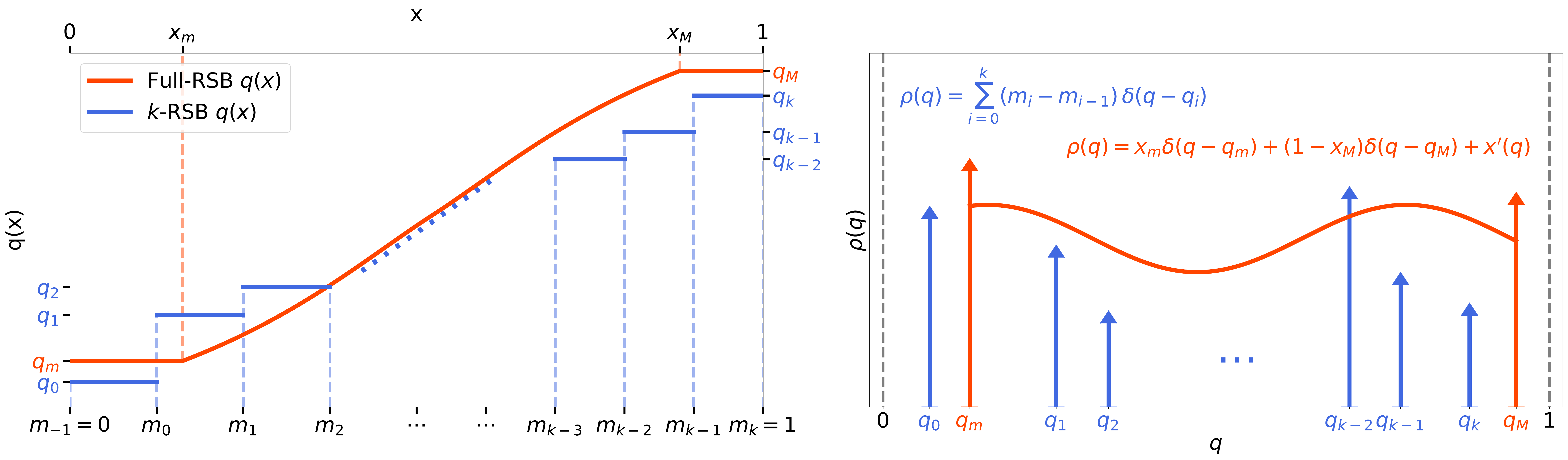}
	\caption{
        Illustration of the finite-RSB and full-RSB structure in the functions $\rho(q)$ (right) and $q(x)$ (left). 
        We use the convention $m_{-1} = 0$.
        In terms of the overlap distribution, the $k$-RSB ansatz (in blue) corresponds to $\rho(q) = \sum_{i=0}^k (m_{i} - m_{i-1}) \delta(q - q_i)$, with the convention $m_{-1} = 0$.
        In orange, the full-RSB distribution is $\rho(q) = \int_0^1 \delta(q-q(x)) \, \rd x$, and is assumed to have two delta peaks at the edges of its support
        $q  \in \{q_m, q_M\}$, with masses $\{x_m,(1-x_M)\}$ (see the equation on the right figure). 
        $x(q)$ is the functional inverse of $q(x)$.
        \label{fig:q_frsb}}
\end{figure}
This allows to consider completely generic distributions $\rho(q)$ (or equivalently functions $q(x)$), by taking the $k \to \infty$ limit of eq.~\eqref{eq:Phi_kRSB_general}. 
This generic procedure is called ``Full Replica Symmetry Breaking'' (Full RSB), and was introduced by Parisi in \cite{parisi1979infinite}.
It yields for the free entropy a formula of the type:
\begin{align}\label{eq:Phi_FRSB_general}
    \Phi_\FRSB(\alpha,\beta) &= \inf_{\{q(x)\}} \{\mcP[\{q(x)\};\alpha,\beta]\}.
\end{align}
Such formulas are usually called Parisi formulas in the spin glass literature.
Note that in many disordered models, the overlap distribution $\rho(q)$ has been observed to have two points with positive mass, at the edges of its bulk (see Fig.~\ref{fig:q_frsb}, right). 
This leads to generically characterize the function $q(x)$ as (see Fig.~\ref{fig:q_frsb} right, red curve):
\begin{align*}
    \begin{cases}
    q(x) = q_m & \textrm{ if } x \in [0,x_m], \\
    q(x) & \textrm{ if } x \in [x_m,x_M], \\
    q(x) = q_M & \textrm{ if } x \in [x_M, 1].
    \end{cases}
\end{align*}
This is purely a convention that often turns out to be convenient
and does not remove any generality as one can always set $x_m = 0$ and $x_M = 1$.

\myskip
\textbf{Relation between $\rho(q)$ and $q(x)$ --}
For an overlap distribution with a well-defined density $\rho(q)$, one has the relation $\rho(q) = x'(q)$, with $x(q) \in [0,1]$ the CDF of the overlap, and $x \mapsto q(x)$ is then the functional inverse of $q \mapsto x(q)$.

\medskip\noindent
\textbf{RSB and the form of the Gibbs measure --}
Interestingly, one can interpret the level of RSB as an assumption on the structure of the level sets of the Gibbs measure (or the global minima of the energy, when $\beta = \infty$).
Roughly speaking, 1-RSB corresponds to an organization of the mass of the Gibbs measure into clusters.
Inside each cluster two solutions typically have overlap $q_1$, while solutions belonging to two different clusters have a typical overlap $q_0$.
This hierarchy can be iterated inside each cluster, which gives rise to the 2-RSB structure.
Iterating even further, the level of RSB corresponds to the depth of this hierarchical structure of clusters, which is known as \emph{ultrametric} \cite{mezard1984nature,panchenko2013parisi}.
Ultrametricity and RSB is a beautiful mathematical representation of the free energy landscape of spin glass models, 
which also allows creating efficient algorithms~\cite{alaoui2021algorithmic,alaoui2021optimization,subag2021following,montanari2021optimization,auffinger2023optimization}.

\myskip
A thorough description of all the consequences of replica symmetry breaking would be beyond our scope:
the major reference on this topic is \cite{mezard1987spin}, and we invite the reader to read as well \cite{talagrand2010mean}, and the very recent lecture notes \cite{montanari2024friendly}, for discussions in a more mathematically-friendly language.

\subsection{One-step replica symmetry breaking}\label{subsec:1rsb}

We start by generalizing the calculation we made in Section~\ref{subsec:rs} to the more general one-RSB ansatz we described above.
We give the results here, while the calculation is detailed in Appendix~\ref{sec_app:1rsb}.
The final result is given as an infimum over three parameters $\{m, q_0, q_1\}$ (see Fig.~\ref{fig:q_frsb} for their interpretation):
\begin{align}\label{eq:phi_1rsb}
   \Phi_\ORSB(\alpha,\beta) &= \inf_{m,q_0,q_1} \Bigg[
   \frac{m - 1}{2m} \log (1-q_1) + \frac{1}{2m} \log [1-mq_0 + (m-1)q_1] + \frac{q_0}{2[1-mq_0 + (m-1)q_1]} \nonumber  \\
   &+ \frac{\alpha}{m} \int \mcD \xi_0 \log \Bigg\{ \int \mcD \xi_1 \Bigg[1- (1-e^{-\beta}) H\Big(- \frac{\sqrt{q_0} \xi_0 + \sqrt{q_1 - q_0} \xi_1}{\sqrt{1-q_1}}\Big) \Bigg]^m \Bigg\} \Bigg].
\end{align}
Note that when $q_1 = q_0$ or when $m = 1$, the overlap distribution $\rho(q)$ reduces to a single delta peak, and we consistently retrieve the replica-symmetric solution of eq.~\eqref{eq:phi_RS}.

\myskip
\textbf{The zero-temperature limit and the injectivity threshold --}  
In Appendix~\ref{subsec_app:zerotemp_1rsb}, we detail how to take the $\beta \to \infty$ limit in $\Phi_\ORSB(\alpha,\beta)$, 
and to obtain the function $f^\star_\ORSB(\alpha) \coloneqq \lim_{\beta \to \infty} [-\Phi_\ORSB(\alpha, \beta)/\beta]$. In Appendix~\ref{subsec_app:1rsb_numerical} we present the numerical procedure we used to solve the resulting equations. 
We reach the light blue curve in Fig.~\ref{fig:chi_estar_T0} for $f^\star_\ORSB(\alpha)$, and in particular 
we have 
\begin{align}\label{eq:alphainj_1RSB}
   \alpha_\inj^\ORSB \simeq 6.7157. 
\end{align}

\myskip
\textbf{Validity of the 1-RSB assumption --}
While the 1-RSB ansatz is a natural extension of the previous replica symmetric assumption,
the results of \cite{franz2017universality} (which study the same model with a slightly different energy function) strongly suggest that for any $\alpha > 2$, at low enough temperatures
the system undergoes a continuous transition from a RS to a Full RSB phase, without any finite level of RSB at intermediate temperatures\footnote{
In particular, a stability analysis of the 1-RSB ansatz, similar to what we did in Appendix~\ref{subsec_app:stability_rs}, would yield that it becomes unstable at the same temperature as the RS 
ansatz.
}. 
This motivates us to compute the complete Full RSB picture in Section~\ref{sec:full_rsb}.
Nevertheless, we will see the 1-RSB prediction of eq.~\eqref{eq:alphainj_1RSB} is already very accurate.

\section{The full-RSB solution: exact injectivity threshold}\label{sec:full_rsb}
\subsection{The full-RSB prediction for the free entropy}

The full-RSB calculation is detailed in Appendix~\ref{sec_app:frsb}, and quite closely follows a similar derivation 
presented in \cite{franz2017universality,urbani2018statistical}. 

\myskip
\textbf{Notations --}
Before stating the result, let us introduce some notation.
For any $\sigma \geq 0$, we let $\gamma_{\sigma^2}(h) = \exp\{-h^2/(2\sigma^2)\} / \sqrt{2 \pi \sigma^2}$ the PDF of $\mcN(0,\sigma^2)$.
For two functions $a, b : \bbR \to \bbR$, we denote $(a \star b)(h) = \int \rd u \, a(u) b(h-u)$ their convolution.
For a function $f(x,h)$ with $x \in [0,1]$ and $h \in \bbR$, we always consider convolutions in the $h$ variable, e.g.\ the notation 
$\gamma_{\sigma^2} \star f(x,h)$ denotes the function $(\gamma_{\sigma^2} \star f)(x,h) = \int \rd u \, \gamma_{\sigma^2}(u) f(x, h - u)$.
Moreover, we denote with a dot derivatives in the $x$ variable, and with a prime derivatives in the $h$ variable, e.g.\ $\dot{f} = \partial_x f$ and $f''= \partial^2_h f$.

\myskip
Let us now state the results of the full-RSB calculation.
We obtain the following formula for the free entropy: 
\begin{align}\label{eq:phi_frsb}
   \Phi_\FRSB(\alpha,\beta) &= \inf_{\{q(x)\}} \Bigg\{\frac{1}{2} \log (1-q(1)) + \frac{q(0)}{2(1-\langle q \rangle)} + \frac{1}{2} \int_0^1 \rd u \frac{\dot{q}(u)}{\lambda(u)} + \alpha (\gamma_{q(0)} \star f) (0, 0)\Bigg\}.
\end{align}
Here, we denoted $\langle q \rangle = \int_0^1 \rd u \, q(u)$ and we defined the auxiliary function: 
\begin{align}\label{eq:def_lambdax}
    \lambda(x) &\coloneqq 1 - x q(x) - \int_x^1 \rd y \, q(y).
\end{align}
Moreover, $f(x,h)$ is taken to be the solution of the \emph{Parisi PDE}:
\begin{align}\label{eq:Parisi_PDE}
    \begin{dcases}
        f(1,h) &= \log{\gamma_{1-q(1)} \star e^{-\beta \theta}(h)}, \\ 
       \dot{f}(x,h) &= - \frac{\dot{q}(x)}{2} \big[f''(x,h) + x f'(x,h)^2\big], \hspace{0.5cm} x \in (0,1).
    \end{dcases}
\end{align}
Similar equations were derived and analyzed in \cite{franz2017universality,urbani2018statistical}. 
These works followed a long series of important papers on 
the spherical perceptron and its connection to the packing of hard spheres \cite{charbonneau2014fractal,franz2015universal,rainone2015following}.
Note that these works consider a shift $\sigma$ in the perceptron activation, so that here we are in the $\sigma = 0$ setting of their results.
Moreover, their energy function is slightly different from eq.~\eqref{eq:def_energy}, as it contains a multiplicative quadratic term.

\myskip
\textbf{The positive-temperature FRSB equations --}
In order to impose the Parisi PDE constraint on the function $f(x,h)$ in eq.~\eqref{eq:phi_frsb}, we use a functional Lagrange multiplier $\Gamma(x,h)$. 
This yields that the free entropy $\Phi_\FRSB(\alpha,\beta)$ is given by the extremization with respect to $q(x), \Lambda(x,h), f(x,h)$ of:
\begin{align}\label{eq:frsb_fentropy_2}
  \Phi_\FRSB(\alpha,\beta) &= \frac{1}{2} \log (1-q(1)) + \frac{q(0)}{2(1-\langle q \rangle)} + \frac{1}{2} \int_0^1 \rd u \frac{\dot{q}(u)}{\lambda(u)} + \alpha \gamma_{q(0)} \star f (0,0) \nonumber \\ 
  & - \alpha \int \rd h \, \Lambda(1,h) [f(1,h) - \log \gamma_{1-q(1)} \star e^{-\beta \theta}(h)] \nonumber  \\ 
  & +  \alpha \int_{0}^1 \rd x \int \rd h \, \Lambda(x,h) [\dot{f}(x,h) + \frac{\dot{q}(x)}{2} (f''(x,h) + x f'(x,h)^2)].
\end{align}
Differentiating these equations with respect to $\Lambda(x,h)$ yields the Parisi PDE of eq.~\eqref{eq:Parisi_PDE} (as it should), 
while differentiation w.r.t.\ $q(x)$ and $f(x,h)$ respectively yield:
\begin{subnumcases}{\label{eq:frsb_eqs}}
\label{eq:frsb_eq_1}
   \frac{q(0)}{\lambda(0)^2} + \int_0^x \rd u \frac{\dot{q}(u)}{\lambda(u)^2} = \alpha \int \rd h \, \Lambda(x,h)  f'(x,h)^2, &\\ 
\label{eq:frsb_eq_2}
   \dot{\Lambda}(x,h) = \frac{\dot{q}(x)}{2} \Big[\Lambda''(x,h) - 2 x (f'(x,h) \Lambda(x,h))'\Big], &\\ 
\label{eq:frsb_eq_3}
   \Lambda(0, h) = \gamma_{q(0)}(h). &
\end{subnumcases}
Finally, differentiation w.r.t.\ $\beta$ yields the average energy: 
\begin{align}\label{eq:estar_frsb}
   e^\star_\FRSB(\alpha, \beta) &\coloneqq - \Phi_\FRSB'(\alpha,\beta) = \alpha \int \rd h \Lambda(1,h) \frac{\big[\gamma_{1-q(1)} \star \theta e^{-\beta \theta}\big](h)}{\big[\gamma_{1-q(1)} \star e^{-\beta \theta}\big](h)}.
\end{align}

\myskip 
\textbf{A sanity check: the RS solution --}
In the RS assumption, we have $q(x) = q_0$ for all $x$. In particular, this implies that 
$\dot{q}(x) = 0$, and $q(0) = \langle q \rangle = q_0$. Moreover, it is easy to see that in this case, 
since $\dot{q}(x) = 0$, we have $f(x,h) = f(1,h) = \log \gamma_{1-q_0} \star e^{-\beta \theta(h)}$ for all $x$, and similarly $\Lambda(x,h) = \gamma_{q_0}(h)$. 
Therefore, eq.~\eqref{eq:frsb_eq_1} becomes:
\begin{align}\label{eq:frsb_eq_rs}
   \frac{q_0}{(1-q_0)^2} &= \alpha \int \rd h \gamma_{q_0}(h) \Bigg[\frac{\big(\gamma_{1-q_0} \star e^{-\beta \theta}\big)'(h)}{\gamma_{1-q_0} \star e^{-\beta \theta}(h)}\Bigg]^2.
\end{align}
One can check (the derivation is presented in Appendix~\ref{subsec_app:rs_from_rsb}) that this equation is equivalent to eq.~\eqref{eq:q_RS_eq_new}: we found back the RS solution!

\subsection{Zero-temperature limit and algorithm for the injectivity threshold}
\label{subsec:zero_temp_algorithmic_frsb}

\textbf{The zero-temperature limit --}
In the zero temperature limit, the scaling of the FRSB equations in the ``UNSAT'' phase of a slightly different spherical perceptron
has been worked out in \cite{franz2017universality}. 
The scaling with $\beta$ of the solution to eqs.~\eqref{eq:Parisi_PDE} and \eqref{eq:frsb_eqs}, as $\beta \to \infty$, can be deduced by transposing their arguments to our model.
More precisely, in the $\beta \to \infty$ limit, letting $\lambda(q) \coloneqq \lambda[x(q)]$ and $f(q,h) \coloneqq f(x(q), h)$, 
one can show that the ``rescaled'' variables $(\beta(1-q_M) ; \beta x(q) ; \beta \lambda(q) ; \beta^{-1} f(q,h))$ 
satisfy non-trivial limiting equations as $\beta \to \infty$, and we define $(\chi; x_\infty; \lambda_\infty; f_\infty)$ as the limiting values of these variables:
\begin{align}\label{eq:frsb_zerotemp_scaling}
   \begin{cases}
       \beta(1 - q_M) &\to \chi, \\
       \beta x(q) &\to x_\infty(q), \\
       \beta \lambda(q) &\to \lambda_\infty(q), \\
       \beta^{-1} f(q,h) &\to f_\infty(q, h).
   \end{cases}
\end{align}
Moreover, $\Lambda(q,h) \coloneqq \Lambda(x(q),h)$ remains finite as $\beta \to \infty$.
In particular, since $x(q = 1) = 1$ by definition (see Fig.~\ref{fig:q_frsb}), we have that $x_\infty(q)$ now extends up to $+\infty$.
We define $q_\infty(x)$ as the inverse function to $x_\infty(q)$, and then we can define 
all functions in terms of $x$, e.g.\ $f_\infty(x, h) \coloneqq f_\infty(q_\infty(x), h)$.
In this limit, all eqs.~\eqref{eq:frsb_eq_1},\eqref{eq:frsb_eq_2},\eqref{eq:frsb_eq_3} scale very naturally, and the 
Parisi PDE of eq.~\eqref{eq:Parisi_PDE} as well. The only non-trivial part is the boundary condition at $x = 1$, 
which becomes 
\begin{align*}
   f_\infty(x \to +\infty, h) &= \frac{1}{\beta}\log \gamma_{\chi/\beta} \star e^{-\beta \theta}(h) + \smallO(1). 
\end{align*}
The scaling of the right hand-side can be worked out exactly:
\begin{align}\label{eq:f_xinfty}
 f_\infty(x \to +\infty, h) &=  \begin{cases}
     0 & \textrm{ if } h < 0, \\ 
     -1 & \textrm{ if } h > \sqrt{2\chi}, \\ 
     -\frac{h^2}{2\chi} & \textrm{ otherwise }.
 \end{cases}
\end{align}
Similarly, we can work out the zero-temperature limit of eq.~\eqref{eq:estar_frsb}, and we get: 
\begin{align*}
   f^\star_\FRSB(\alpha) &= \lim_{\beta \to \infty} e^\star_\FRSB(\alpha, \beta) = \alpha \int_{\sqrt{2\chi}}^\infty \Lambda_\infty(x \to \infty,h)  \rd h.
\end{align*}

\myskip
\textbf{Algorithmic procedure --}
In this paragraph, for the clarity of the presentation, all quantities are considered in the zero-temperature limit, and we drop the $\infty$ subscripts.

\myskip
The procedure we use is relatively similar to the finite-temperature one described in Appendix~B of \cite{franz2017universality}, but is done at zero temperature,
and at fixed $x$ rather than fixed $q$ (as we found this choice to be numerically more stable).
In order to increase numerical precision, we rescale $h$ and use $t = h / \sqrt{2\chi}$, allowing to handle small values of the susceptibility $\chi$.
Our algorithmic procedure is as follows:

\myskip
\begin{itemize}
   \item \textbf{Before starting --} Pick $k$ large enough, $x_\mathrm{max} \gg 1$ large enough, and a grid $0 < x_0 < x_1 < \cdots < x_{k-1} = x_\mathrm{max} < x_k = \infty$.
   \item \textbf{Initialization --} Start from a guess $\chi > 0$ and $0 \leq q_0 \leq q_1 \leq \cdots \leq q_{k-1} < q_k = 1$.
   \item[$(i)$] Find the functions $f(x_i,t)$ via the procedure: 
   \begin{align}\label{eq:frsb_procedure_i}
      f(x_k = \infty, t) &= 
      \begin{dcases}
         0 & \textrm{ if } t \leq 0 , \\    
         -1 & \textrm{ if } t \geq 1, \\    
         - t^2 & \textrm{ if } t \in (0,1).
      \end{dcases}, \\
      f(x_i,t) &= \frac{1}{x_i} \log \Big[\gamma_{\frac{q_{i+1} - q_i}{2\chi}} \star e^{x_i f}(x_{i+1},t)\Big].
   \end{align}
   \item[$(ii)$] Find $\Lambda(q_i,t)$ via the procedure:
   \begin{align}\label{eq:frsb_procedure_ii}
    \begin{dcases}
        \Lambda(x_0, t) &= \gamma_{q_0}(\sqrt{2\chi}t), \\
        \Lambda(x_i,t) &= e^{x_{i-1} f(x_i,t)}\, \gamma_{\frac{q_i - q_{i-1}}{2\chi}} \star \Big[\Lambda \cdot e^{-x_{i-1} f}\Big](x_{i-1},t).
    \end{dcases}
   \end{align}
   \item[$(iii)$] Compute $q^{-1}_i$ (the hierarchical elements of $\bQ^{-1}$, not $1/q_i$) using, for all $i \in \{0,\cdots,k\}$:
   \begin{align}\label{eq:frsb_procedure_iii}
      q_i^{-1} &= - \frac{\alpha}{\sqrt{2\chi}} \int \rd t \, \Lambda(x_i,t) \, f'(x_i, t)^2.
   \end{align}
   \item[$(iv)$] Update $\lambda_i = \lambda(q_i)$ via 
   \begin{align}\label{eq:frsb_procedure_iv}
      \begin{dcases}
      \lambda_0 &= \sqrt{-\frac{q_0}{q_0^{-1}}}, \\    
      \frac{1}{\lambda_i} &= \frac{1}{\lambda_{i-1}} - x_{i-1} (q_i^{-1} - q_{i-1}^{-1}).
      \end{dcases}
   \end{align}
   \item[$(v)$] Update $\{q_i\}_{i=0}^k$ with $q_k = 1$ and
   \begin{align}\label{eq:frsb_procedure_v}
      q_i &= 1 - \frac{\lambda_i}{x_i} - \sum_{j=i+1}^k \Big(\frac{1}{x_j} - \frac{1}{x_{j-1}}\Big) \lambda_j.
   \end{align}
   \item[$(vi)$]
   Update $\chi$ by solving the equation (with $q_{-1} = 0$ and $x_{-1} = 0$): 
   \begin{align}\label{eq:frsb_procedure_vi}
      \sum_{i=0}^k \frac{\sqrt{\chi}(q_i - q_{i-1})}{\Big[\chi + \sum_{j=i+1}^k (q_j - q_{j-1}) x_{j-1}\Big]\Big[\chi + \sum_{j=i}^k (q_j - q_{j-1}) x_{j-1}\Big]}
      &= 2^{3/2}\alpha \int_0^{1} \rd t \, \Lambda( 1,t) \, t^2.
   \end{align}
   \item Iterate steps $(i) \to (vi)$ until convergence. 
   \item \textbf{Final value for the energy --} We then compute the ground state energy as: 
   \begin{align*}
      f^\star_\FRSB(\alpha) &= \alpha \sqrt{2\chi} \int_{1}^\infty \rd t \, \Lambda(1,t).
   \end{align*}
\end{itemize}
The procedure is done for $k$ large enough so that the result does not vary with $k$ and approaches the $k \to \infty$ limit.
Steps~$(i)$ and $(ii)$ are a discretization of the zero-temperature limits of the PDEs of eqs.~\eqref{eq:Parisi_PDE} and \eqref{eq:frsb_eq_2}, arising from the $k$-RSB ansatz (see Appendix~\ref{sec_app:frsb}).
We give more details on the derivation of steps $(iii) - (vi)$ in Appendix~\ref{subsec_app:derivation_algorithmic_frsb}, leveraging results of \cite{franz2017universality}.

\myskip
The different convolutions with Gaussians are done using an analytical formula for the Discrete Fourier Transform (DFT) of a Gaussian under a Shannon-Whittaker interpolation, 
and fast Fourier transform techniques. More details on this point are given in Appendix~\ref{subsec_app:convolutions}.

\myskip
\textbf{Implementation and results --}
We present our results for $f^\star_\FRSB$ and the zero-temperature susceptibility $\chi$ in Fig.~\ref{fig:chi_estar_T0}, and the zero-temperature overlap distribution function $q(x)$ 
for various values of $\alpha$ in Fig.~\ref{fig:q_T0}.
\begin{figure}[th]
   \centering
\includegraphics[width=0.8\textwidth]{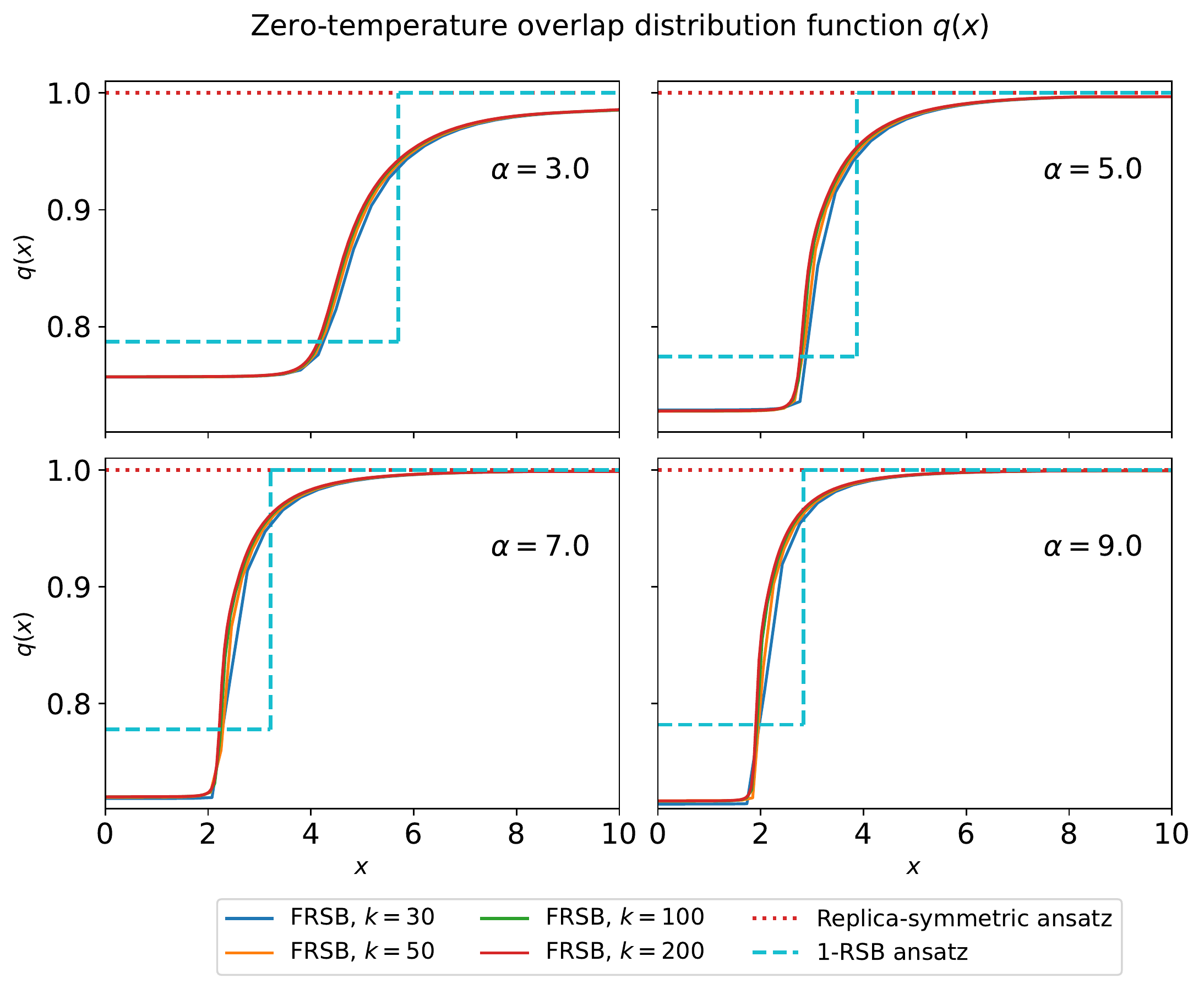}
\caption{$T=0$ limit quantities of the RS, 1RSB and FRSB solutions, as a function of $\alpha$.
   We compare the predictions for the different
   forms of the function $q(x)$ corresponding to the assumed level of replica symmetry breaking.
\label{fig:q_T0}}
\end{figure}
In particular, the full-RSB prediction for the injectivity threshold is $\alpha_\inj^\mathrm{FRSB} \simeq 6.698$.
We ran a more precise binary search procedure for computing the value of this transition, 
which we detail in Appendix~\ref{subsec_app:numerics_frsb_threshold}. A summary of its result is presented 
in Fig.~\ref{fig:alpha_inj_frsb_summary}, and it yields the bound we conjecture in Result~\ref{result:frsb_result}:
\begin{align}\label{eq:fRSB_inj_threshold}
   6.6979 \leq \alpha_\inj^\mathrm{FRSB} \leq 6.6981.
\end{align}
\begin{figure}[th]
   \centering
\includegraphics[width=1.0\textwidth]{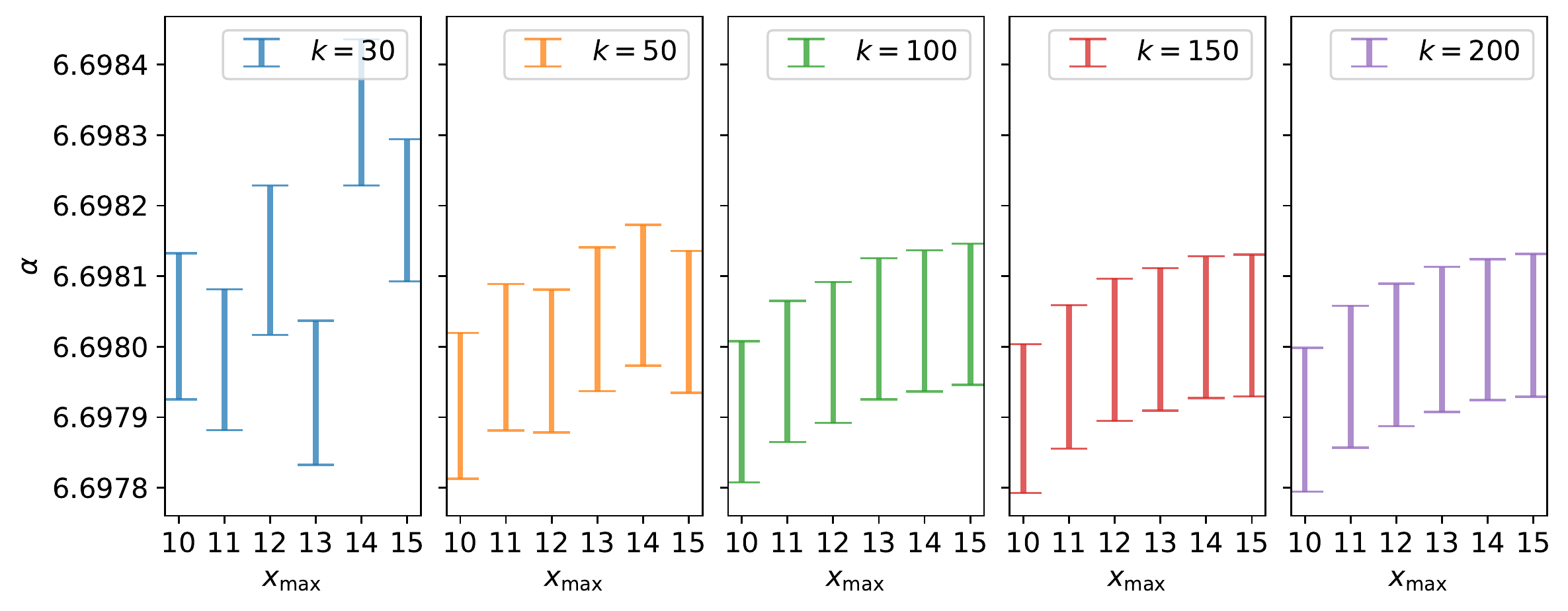}
\caption{\label{fig:alpha_inj_frsb_summary}
Computation of $\alpha_\inj^\FRSB$ using the FRSB algorithmic procedure.
For different values of $x_\mathrm{max}$ and $k$ we give an interval numerically found to contain $\alpha_\inj^\FRSB$.
In Result~\ref{result:frsb_result} we took the interval of values obtained with $k = 200$ and $x_\mathrm{max} = 15$.
We give more details on the numerical procedure in Appendix~\ref{subsec_app:numerics_frsb_threshold}.
}
\end{figure}
Note that this bound is compatible with the hierarchy described in eq.~\eqref{eq:hierarchy_upper_bounds_alpha_inj}.
Moreover, the 1-RSB predictions are found to be very close (but not equal) to the exact FRSB results.
This can be intuitively visualized by the fact that $q(x)$ is relatively well approximated by a step function, which corresponds to the 1RSB ansatz, cf.\ Fig.~\ref{fig:q_T0}.
We emphasize however that the full-RSB algorithmic procedure above does not allow to directly recover the $1$-RSB result, even used with $k = 1$: 
indeed it implicitly relies on the fact that one takes $k$ large enough so as not to have to optimize over the variables $x_1, \cdots, x_k$, so that we can take them to be fixed.

\myskip
\textbf{Remark: convexity of the Parisi functional --} In the context of mixed $p$-spin models, the so-called Parisi functional, i.e.\ the functional whose infimum we take in eq.~\eqref{eq:phi_frsb}, has been shown to be strictly convex, 
and thus to have a unique minimizer \cite{auffinger2015parisi}. 
This is conjectured to hold as well in our setting, however there is no rigorous guarantee that our iterative procedure should converge to a global minimizer.
However, our numerical simulations are compatible with this conjecture: as we detail in Appendix~\ref{subsec_app:numerical_results_frsb}, we find the iterative procedure 
to converge to a consistent solution for all initializing points.
Moreover, our procedure exhibits polynomial convergence (see Fig.~\ref{fig_app:convergence_frsb}): 
this suggests that there is an accumulation of near-zero eigenvalues in the Hessian of the Parisi functional close to the minimum (otherwise we would observe exponential convergence), and thus that the Parisi functional is strictly but not strongly convex.

\bibliographystyle{alpha}
\bibliography{refs}

\appendix 
\addtocontents{toc}{\protect\setcounter{tocdepth}{1}} 

\section{Proofs}\label{sec_app:proofs}
\subsection{Proof of Proposition~\ref{prop:injectivity_random_intersection}}\label{subsec_app:proof_random_intersection}

Note that if $m < n$, then 
$\bx \in \bbR^n \mapsto \bW \bx \in \bbR^m$ is not injective,
$C_{m,n} = \bbR^m$, and eq.~\eqref{eq:pmn_random_intersection} stands trivially.
We thus assume $m \geq n$. Since $\bx \in \bbR^n \mapsto \bW \bx \in \bbR^m$ is then a.s.\ injective, $\bW \bbR^n$ is a (random) $n$-dimensional subspace of $\bbR^m$.
Moreover, by rotation invariance of the Gaussian distribution, it is uniformly sampled.
Thus it is enough to show that $p_{m,n} = \bbP[(\bW \bbR^{n}) \cap C_{m,n} = \{0\}]$.
In the end, it suffices to show the following lemma, whose proof elements can be found in \cite{puthawala2022globally,paleka2021injectivity,clum2022topics}, which we repeat for completeness: 
\begin{lemma}\label{lemma:technical_random_intersection}
    \noindent
    Almost surely under the law of $\bW$, the following two statements are equivalent: 
    \begin{itemize}
        \item[$(i)$] $\varphi_\bW$ is injective.
        \item[$(ii)$] $(\bW \bbR^{n}) \cap C_{m,n} = \{0\}$.
    \end{itemize}
\end{lemma}
\begin{proof}[Proof of Lemma~\ref{lemma:technical_random_intersection} --]
    In the following, we assume that the following event stands: 
    \begin{align*}
        E(\bW) \coloneqq \{\textrm{All choices of $n$ distinct rows of $\bW$ are linearly independent vectors in $\bbR^n$}\}.
    \end{align*}
    It is easy to see that $\bbP[E(\bW)] = 1$, since every set of $n$ independent standard Gaussian vectors in $\bbR^n$ is linearly independent almost surely.

    \myskip
    Let us show first that $(ii) \Rightarrow (i)$.
    Recall $\mathrm{ReLU}(x) = \max(0,x)$.
        Note that for any $a \leq b \in \bbR$, $\mathrm{ReLU}(a) = \mathrm{ReLU}(b)$ 
        implies that ReLU is constant on $(a,b)$.
    Assume that $\varphi_\bW(\bx) = \varphi_\bW(\by)$. 
    Let us consider $\bz = (\bx + \by) / 2$, then $\varphi_\bW(\bz) = \varphi_\bW(\bx)$ by the note above.
    Moreover, for all $\mu \in [m]$ such that $(\bW \bz)_\mu > 0$, then $(\bW \bx)_\mu = (\bW \by)_\mu = (\bW \bz)_\mu$.
    On the other hand, if $(\bW \bz)_\mu = 0$, then necessarily $(\bW \bx)_\mu = (\bW \by)_\mu = 0$, since 
    $(\bW \bx)_\mu \leq 0 \Leftrightarrow (\bW \by)_\mu \leq 0$.
    By $(ii)$, $\bW \bz$ has at least $n$ non-negative coordinates, so the argument above implies that there exists at least 
    $n$ values of $\mu \in [m]$ such that $(\bW \bx)_\mu = (\bW \by)_\mu$.
    Since $E(\bW)$ stands, this shows that $\bx = \by$.

    \myskip 
    Let us now show $(i) \Rightarrow (ii)$.
    We divide $\bbR^n$ into equivalence classes defined by the relation 
    $\bx \sim \by \Leftrightarrow \forall \mu \in [m], (\bW \bx)_\mu > 0 \Leftrightarrow (\bW \by)_\mu > 0$.
    These equivalence classes $\mcR_S$ are defined by a subset $S$ of $[m]$, 
    so that for all $\bx \in \mcR_S$, $(\bW \bx)_\mu > 0 \Leftrightarrow \mu \in S$.
    Assume that there exists $\bx \in \bbR^n$ such that $\bW \bx \neq 0$ and $\bW \bx$ has strictly less than $n$ positive coordinates, i.e.\  
    $\bx \in \mcR_S$ with $|S| < n$. On $\mcR_S$, $\varphi_\bW$ is a linear transformation with $|S| < n$ linearly independent rows $\{\bW_\mu\}_{\mu \in S}$. 
    Its image has thus dimension smaller than $n$.
    Therefore, the following result, which implies that $\mcR_S$ has dimension $n$, then implies that $\varphi_\bW$ is not injective on $\mcR_S$. Having proved the contrapositive, we can then infer $(i) \Rightarrow (ii)$.
\end{proof}
    \begin{lemma}\label{lemma:technical_equivalence_classes}
        \noindent
        The following statement is true almost surely: 
        for all $S \subseteq [m]$, either $\mcR_S = \{0\}$, or there exists $\bx \in \mcR_S$ and $\varepsilon > 0$ such that 
        $B_2(\bx, \varepsilon) \subseteq \mcR_S$. 
    \end{lemma}

\begin{proof}[Proof of Lemma~\ref{lemma:technical_equivalence_classes} --]
    By a union bound over all $S \subseteq [m]$, it suffices to show this statement a.s.\ for any fixed $S \subseteq [m]$.
    Let us assume that $\mcR_S \neq 0$. The following statement implies the conclusion of Lemma~\ref{lemma:technical_equivalence_classes}: 
    \begin{align}\label{eq:to_show_equivalence_class}
        \bbP \Big\{\forall \bx \in \mcR_S, \, \exists \mu \notin S \textrm{ s.t. } \bW_\mu \cdot \bx = 0\Big\} = 0.
    \end{align}
    Indeed, one can then a.s.\ find an element $\bx \in \mcR_S$ such that $(\bW \bx)_\mu > 0$ for all $\mu\in S$ and $(\bW \bx)_\mu < 0$ for all $\mu\notin S$, therefore $B_2(\bx, \varepsilon) \subseteq \mcR_S$ for sufficiently small $\varepsilon$. 
    We now show eq.~\eqref{eq:to_show_equivalence_class}.
     
    \myskip
    Assume that there exists $\bx \in \mcR_S \backslash\{0\}$, with $\nu_1, \cdots, \nu_{k_\bx} \in [m]$ all indices such that
    $\bW_\nu \cdot \bx = 0$, that satisfies $k_\bx \geq 1$. Note that by $E(\bW)$ (which stands a.s.) and since $\bx \neq 0$ we must have $k_\bx < n $. Thus, since $\{\bW_{\nu_i}\}_{i=1}^{k_\bx}$ are linearly independent on $E(\bW)$, 
    we can then fix $\by \in \big(\{\bW_{\nu_i}\}_{i=1}^{k_\bx-1} \big)^\perp$ such that $\bW_{\nu_{k_\bx}} \cdot \by < 0$.
    Consider $\bx' = \bx + \delta \by$ with arbitrary $\delta > 0$. 
    By hypothesis, $\bW_{\nu_i} \cdot \bx' = 0$ for all $i \in [k_\bx-1]$. Moreover, for $\delta$ small enough, 
    $\bW_\mu \cdot \bx'$ has the same sign as $\bW_\mu \cdot \bx$ if $\mu \notin \{\nu_1, \cdots, \nu_{k_\bx}\}$.
    Finally, $\bW_{\nu_{k_\bx}} \cdot \bx' = \delta \bW_{\nu_{k_\bx}} \cdot \by < 0$. In the end, taking $\delta$ small enough, we have found 
     $\bx' \in \mcR_S$ with $k_{\bx'} = k_{\bx} - 1$. Iterating this procedure, we have shown that a.s.\ there exists a point $\bx \in \mcR_S$ such that $k_\bx = 0$, which implies eq.~\eqref{eq:to_show_equivalence_class}.
\end{proof}

\subsection{Proof of Lemma~\ref{lemma:cover}}\label{subsec_app:proof_cover}

Let us first recall Cover's theorem \cite{cover1965geometrical}. 
We use the $\mathrm{sign}(x)$ function, with the convention $\sign(0) = 0$.
We call a set of vectors $\{\bW_1,\cdots,\bW_m\}$ in $\bbR^n$ in \emph{general position} if it has no linearly independent subset of size strictly less than $n$.
Cover's theorem is an exact formula for the number of dichotomies\footnote{A dichotomy is a binary labeling of the vectors.} of this set that are realizable by a linear separation:
\begin{theorem}[Cover \cite{cover1965geometrical}]\label{thm:cover}
    \noindent
    Let $\bW_1,\cdots,\bW_m \in \bbR^n$ be in general position. Then
    \begin{align*}
        \sum_{\bepsilon \in \{\pm 1\}^m} \indi \Big\{ \exists \bx \in \mcS^{n-1} \textrm{ realizing } \forall \mu \in [m] \, : \, \mathrm{sign}(\bW_\mu \cdot \bx) = \varepsilon_\mu \Big\} = 2 \sum_{k=0}^{n-1} \binom{m-1}{k}.
    \end{align*}
\end{theorem}
Let us now show that Theorem~\ref{thm:cover} implies Lemma~\ref{lemma:cover}.
We assume $\alpha < 3$, so in particular we can fix $\delta > 0$ such that $m \leq (3-\delta) (n-1)$ for $n$ large enough.
We denote $\tilde{m} = m - (n-1) \leq (2-\delta) n$.
Since $\bW$ is a Gaussian matrix, the set $\{\bW_\mu\}_{\mu \in [\tilde{m}]}$ is a.s.\ in general position.
Moreover, by sign invariance, 
for any $\bepsilon \in \{\pm 1\}^m$ we have: 
\begin{align*}
    \bbP_\bW \Big\{\exists \bx \in \mcS^{n-1} \, : \, \forall \mu \in [\tilde{m}] , \, \mathrm{sign}(\bW_\mu \cdot \bx) = \varepsilon_\mu \Big\} &=
    \bbP_\bW \Big\{\exists \bx \in \mcS^{n-1} \, : \, \forall \mu \in [\tilde{m}] , \, \mathrm{sign}(\bW_\mu \cdot \bx) = -1 \Big\}.
\end{align*} 
For $\bepsilon$ uniformly sampled in $\{\pm 1\}^m$ (independently of $\bW$), 
we denote $\bbP_{\bW,\bepsilon}$ the joint probability law of $(\bW, \bepsilon)$, and $\bbP_\bepsilon$ the law of $\bepsilon$. 
The previous remark on sign invariance allows to deduce:
\begin{align*}
    &\bbP_\bW \Big\{\exists \bx \in \mcS^{n-1} \, : \, \forall \mu \in [\tilde{m}] , \, \mathrm{sign}(\bW_\mu \cdot \bx) = -1 \Big\} \\
    &= 
    \bbP_{\bW,\bepsilon} \Big\{\exists \bx \in \mcS^{n-1} \, : \, \forall \mu \in [\tilde{m}] , \, \mathrm{sign}(\bW_\mu \cdot \bx) = \varepsilon_\mu \Big\}, \\ 
    &= 
    \EE_{\bW} \Big[\bbP_\bepsilon\Big(\exists \bx \in \mcS^{n-1} \, : \, \forall \mu \in [\tilde{m}] , \, \mathrm{sign}(\bW_\mu \cdot \bx) = \varepsilon_\mu \Big | \bW\Big)\Big], \\
    &= 
    \frac{1}{2^{\tilde{m}}} \times 2 \sum_{k=0}^{n-1} \binom{\tilde{m}-1}{k},
\end{align*}
by Theorem~\ref{thm:cover}. Since $\tilde{m} \leq (2-\delta) n$, it is then elementary to check that this implies
\begin{align*}
    \lim_{n \to \infty} \bbP_\bW \Big\{\exists \bx \in \mcS^{n-1} \, : \, \forall \mu \in [\tilde{m}] , \, \mathrm{sign}(\bW_\mu \cdot \bx) = -1 \Big\} &= 1.
\end{align*}
The proof is then finished by noticing that if $\bx$ satisfies $\mathrm{sign}(\bW_\mu \cdot \bx) = -1$ for all $\mu \in [\tilde{m}]$, 
it must satisfy $E_\bW(\bx) \leq m - \tilde{m} < n$, and using eq.~\eqref{eq:pmn_minimum}.

\subsection{Proof of Theorem~\ref{thm:free_entropy_concentration}}\label{subsec_app:proof_thm_fe_concentration}

    It is easy to see that if $\bW, \bW'$ are two matrices for which $\bW'_{\nu} = \bW_{\nu}$ for all $\nu \in [m] \backslash \{\mu\}$, then 
    \begin{align*}
        |\Phi_n(\bW, \beta) - \Phi_n(\bW',\beta)| &\leq \frac{2 \beta}{n}.
    \end{align*}
    The theorem is then a simple consequence of McDiarmid's inequality (see e.g.\ Theorem~6.2 of \cite{boucheron2013concentration}).

\subsection{Proof of Corollary~\ref{cor:sufficient_non_injectivity}}\label{subsec_app:proof_cor_sufficient_non_inj}
    By a dominated convergence argument, we have: 
    \begin{align*}
        \partial_\beta \{\EE_\bW \Phi_n(\bW, \beta)\} &= - \frac{1}{n} \EE_\bW \Big\{\frac{\int_{\mcS^{n-1}} \mu_n(\rd \bx) \, E_\bW(\bx) \, e^{-\beta E_\bW(\bx)}}{\int_{\mcS^{n-1}} \mu_n(\rd \bx) \, e^{-\beta E_\bW(\bx)}}\Big\}.
    \end{align*}
    Therefore
    \begin{align*}
        -\beta^2 \partial_\beta \{\EE_\bW \Phi_n(\bW, \beta)/\beta\} &= \frac{1}{n}\EE_\bW \Big\{\frac{\int_{\mcS^{n-1}} \mu_n(\rd \bx) \, \beta E_\bW(\bx) \, e^{-\beta E_\bW(\bx)}}{\int_{\mcS^{n-1}} \mu_n(\rd \bx) \, e^{-\beta E_\bW(\bx)}}\Big\} + \EE_\bW \Phi_n(\bW, \beta).
    \end{align*}
    Recall the definition of the Gibbs measure $\bbP_{\beta,\bW}$ in eq.~\eqref{eq:def_Gibbs}. It is easy to see that 
    the previous equation relates directly to the entropy of $\bbP_{\beta,\bW}$, i.e.\ 
    \begin{align*}
        \beta^2 \partial_\beta \{\EE_\bW \Phi_n(\bW, \beta)/\beta\} &= \frac{1}{n} \EE_\bW \int \rd \bbP_{\beta,\bW}(\bx) \, \log \frac{\rd \bbP_{\beta,\bW}}{\rd \mu_n}(\bx) = \EE_\bW D_{\rm KL}(\bbP_{\beta,\bW} \| \mu_n) \geq 0.
    \end{align*}
    In the language of statistical physics, this is a rewriting of the fact that the temperature derivative of the free energy is given by (minus) the entropy.
    In particular, for any $n$, $\beta \mapsto -\EE_\bW \Phi_n(\bW, \beta)/\beta$ is non-increasing, and 
    in the limit this shows that $\beta \mapsto -\Phi(\alpha,\beta)/\beta$ is non-increasing.
    The positivity of this function follows from $\Phi_n(\bW, \beta) \leq 0$, since $E_\bW(\bx) \geq 0$.
    
    \myskip 
    Let us now assume that there exists some $\beta < \infty$ such that $- \Phi(\alpha,\beta) < \beta$.
    In particular, fixing $\delta > 0$, for $n$ large enough we have $\EE_{\bW} \Phi_n(\bW,\beta) \geq - \beta + \delta$.
    Using Theorem~\ref{thm:free_entropy_concentration}, for large enough $n$, this implies 
    \begin{align*}
        \bbP_\bW[\Phi_n(\bW, \beta) \leq - \beta] &\leq 2 \exp\Big\{- \frac{n [\beta + \EE_\bW \Phi_n(\bW,\beta)]^2}{2 \alpha_n \beta^2}\Big\}, \\ 
        &\leq 2 \exp\Big\{- C(\alpha,\beta) n\Big\},
    \end{align*}
    for some $C(\alpha, \beta) > 0$.
    In particular, using eq.~\eqref{eq:bound_Phi_energy} and Proposition~\ref{prop:injectivity_random_intersection}: 
    \begin{align}\label{eq:bound_pmn}
        p_{m,n} &= \bbP_\bW \Big[\min_{\bx \in \mcS^{n-1}} E_{\bW}(\bx) \geq n\Big]
        \leq \bbP_\bW \Big[\Phi_n(\bW, \beta) \leq - \beta\Big]\leq 2 \exp\Big\{-C(\alpha,\beta) n \Big\}.
    \end{align}
    The claim follows.

\subsection{Proof of Theorem~\ref{thm:bound_Gordon}}\label{subsec_app:bound_Gordon}

    \textbf{Remark --} In what follows we usually consider $m = \alpha n$ with $\alpha > 0$, 
    and the proof can be straightforwardly generalized to the original assumption $m/n \to \alpha > 0$.
    For lightness of the presentation, we assume the simplified statement we described. 

    \myskip
    First, note that given Lemma~\ref{lemma:cover}, we can assume $\alpha \geq 3$ in what follows.
    Using Proposition~\ref{prop:injectivity_random_intersection}, we want to characterize 
    \begin{align*}
        G(\bW) &\coloneqq \min_{\bx \in \mcS^{n-1}} e_\bW(\bx) = \min_{\bx \in \mcS^{n-1}} \frac{1}{n} \sum_{\mu=1}^m \indi\{\bW_\mu \cdot \bx > 0\},
    \end{align*}
    in which $\bW = \{\bW_\mu\}_{\mu=1}^m \iid \mcN(0,\Id_n)$.
    The minimum of this function is reached since it
    takes discrete values.
    Introducing an auxiliary variable $z_\mu \coloneqq \bW_\mu \cdot \bx$, and a Lagrange multiplier $\blambda \in \bbR^m$ to fix this relation, 
    the problem is equivalent by strong duality to
    \begin{align}\label{eq:ground_state_duality}
        \nonumber
        G(\bW) &= \min_{\bx \in \mcS^{n-1}} \inf_{\bz \in \bbR^m} \sup_{\blambda \in \bbR^m} \Big\{\blambda^\intercal \bW \bx - \blambda^\intercal \bz + \frac{1}{n} \sum_{\mu=1}^m \indi\{z_\mu > 0\}\Big\}, \\ 
        &=  \inf_{\bz \in \bbR^m} \min_{\bx \in \mcS^{n-1}} \sup_{\blambda \in \bbR^m} \Big\{\blambda^\intercal \bW \bx - \blambda^\intercal \bz + \frac{1}{n} \sum_{\mu=1}^m \indi\{z_\mu > 0\}\Big\}.
    \end{align}
    Note that the infimum over $\bz$ in eq.~\eqref{eq:ground_state_duality} is actually done over $\bz \in \bW \mcS^{n-1}$, since the supremum over $\blambda$ becomes $+\infty$ for $\bz \neq \bW \bx$. 
    Letting $\|\bW\|_{\rm op} \coloneqq \max_{\bx \in \mcS^{n-1}} \|\bW \bx\|_2$, 
    we know by classical concentration inequalities (see e.g.\ \cite{vershynin2018high} - Theorem 4.4.5) and since $m = \alpha n$, 
    that $\bbP[\|\bW\|_{\mathrm{op}} \geq K \sqrt{n}] \leq e^{-n}$, for some constant $K > 1$ (that might depend on $\alpha$).
    Let us denote $B(K) \coloneqq \{\bz \in \bbR^m \, : \|\bz \|_2 \leq K \sqrt{n}\}$.
    By the argument above and the law of total probability, for all $t > 0$, 
    \begin{align}\label{eq:ub_G_GK}
        \bbP[G(\bW) \leq t] \leq \bbP[G_K(\bW) \leq t] + e^{-n},
    \end{align}
    with 
    \begin{align}\label{eq:ground_state_duality_K}
        G_{K}(\bW) \coloneqq \inf_{\substack{\bz \in B(K)}} \min_{\bx \in \mcS^{n-1}} \sup_{\blambda \in \bbR^m} \Big\{\blambda^\intercal \bW \bx - \blambda^\intercal \bz + \frac{1}{n} \sum_{\mu=1}^m \indi\{z_\mu > 0\}\Big\}.
    \end{align}
    Moreover, we will approximate $\indi\{z > 0\}$ by continuous functions; we let, for any $\delta \geq 0$:
    \begin{align}\label{eq:def_gdelta}
        \ell_\delta(x) \coloneqq
        \begin{cases}
            0 &\textrm { if } x \leq 0, \\
            1 &\textrm { if } x > \delta, \\ 
            x/\delta &\textrm { if } x \in (0, \delta].
        \end{cases}
    \end{align}
    Since $\ell_\delta(x) \leq \indi\{x > 0\}$, it is clear that:
    \begin{align}\label{eq:ground_state_duality_smoothed}
        G_K(\bW) &\geq G_{\delta,K}(\bW) \coloneqq \inf_{\substack{\bz \in B(K)}} \min_{\bx \in \mcS^{n-1}} \sup_{\blambda \in \bbR^m} \Big\{\blambda^\intercal \bW \bx - \blambda^\intercal \bz + \frac{1}{n} \sum_{\mu=1}^m \ell_\delta(z_\mu)\Big\}.
    \end{align}
    We now make use of the Gaussian min-max theorem \cite{gordon1988milman,thrampoulidis2015regularized}:
    \begin{proposition}[Gaussian min-max theorem]\label{prop:gaussian_minmax}
        \noindent
        Let $\bW \in \bbR^{m \times n}$ be an i.i.d.\ standard normal matrix, and $\bg \in \bbR^m,\bh \in \bbR^n$ two independent vectors 
        with i.i.d.\ $\mcN(0,1)$ coordinates.
        Let $\mcS_\bv, \mcS_\bu$ be two compact subsets respectively of $\bbR^n$ and $\bbR^m$, and let 
        $\psi : \mcS_\bv \times \mcS_\bu \to \bbR$ be a continuous function. 
        We define the two optimization problems: 
        \begin{align*}
            \begin{dcases}
                C(\bW)  &\coloneqq \min_{\bv \in \mcS_\bv} \max_{\bu \in \mcS_{\bu}} \Big\{\bu^\intercal \bW \bv+ \psi(\bv, \bu)\Big\}, \\ 
                \mcC(\bg, \bh) &\coloneqq \min_{\bv \in \mcS_\bv} \max_{\bu \in \mcS_{\bu}} \Big\{\norm{\bu} \bh^\intercal \bv + \norm{\bv} \bg^\intercal \bu+ \psi(\bv, \bu)\Big\}.
            \end{dcases}
        \end{align*}
        Then, for all $t \in \bbR$, one has 
        \begin{align}\label{eq:gordon_bound}
            \bbP_\bW[C(\bW) \leq t] &\leq 2 \, \bbP_{\bg,\bh}[\mcC(\bg,\bh) \leq t].
        \end{align}
    \end{proposition}
    \textbf{Remark I --} 
    It is easy to see from the proof of \cite{gordon1988milman,thrampoulidis2015regularized} that
    the statement of the theorem also holds if $\bW$ is a block matrix of the form 
    \begin{align*}
        \bW = \begin{pmatrix}
            \bW_1 & 0 \\ 
            0 & 0
        \end{pmatrix},
    \end{align*}
    with $\bW_1 \in \bbR^{m_1 \times n_1}$ having i.i.d.\ $\mcN(0,1)$ elements.
    Denoting $\bu^\intercal \bW \bv = \bu_1^\intercal \bW_1 \bv_1$,
    the definition of 
    the auxiliary problem that appears in the theorem is then modified as:
    \begin{align}\label{eq:auxiliary_pb_gordon_block}
         \mcC(\bg, \bh) \coloneqq \min_{\bv \in \mcS_\bv} \max_{\bu \in \mcS_{\bu}} \Big\{\norm{\bu_1} \bh^\intercal \bv_1 + \norm{\bv_1} \bg^\intercal \bu_1+ \psi(\bv, \bu)\Big\},
    \end{align}
    for $\bg \sim \mcN(0, \Id_{m_1})$, $\bh \sim \mcN(0, \Id_{n_1})$.

    \myskip
    \textbf{Remark II --}
    The full result of \cite{thrampoulidis2015regularized} actually includes the proof of a converse bound to eq.~\eqref{eq:gordon_bound} when the function $\psi$ is convex-concave, and the sets $\mcS_\bu, \mcS_\bw$ are convex.
    Here, we do not expect such a converse bound to be true, since the solution is conjecturally described by the full-RSB equations, and we will see that 
    the upper bound of eq.~\eqref{eq:gordon_bound} corresponds to the replica-symmetric (RS) solution.

    \myskip
    Let us first state a lemma that simplifies the auxiliary problem:
    \begin{lemma}[Auxiliary problem simplification]\label{lemma:ao_simplification}
        \noindent
        For any $\delta> 0$, and any $A \in (0,\infty]$, we define the auxiliary optimization problem, for $\bg \in \bbR^m,\bh \in \bbR^n$:
        \begin{align}\label{eq:def_CAd_gh}
                \mcC_{A,\delta}(\bg, \bh) &\coloneqq \inf_{\bz \in \bbR^m} \min_{\bx \in \mcS^{n-1}} \max_{\norm{\blambda} \leq A} \Big\{\norm{\blambda} \bh^\intercal \bx + \bg^\intercal \blambda - \blambda^\intercal \bz + \frac{1}{n} \sum_{\mu=1}^m \ell_\delta(z_\mu)\Big\}.
        \end{align}
        Then $A \mapsto \mcC_{A, \delta}(\bg, \bh)$ is non-decreasing and one has:
        \begin{align}\label{eq:def_Cd_gh}
                \lim_{A \to \infty} \mcC_{A, \delta}(\bg, \bh) &= \min_{\substack{\bz \in \bbR^m \\ \norm{\bz} \leq \norm{\bh}}} \Big\{\frac{1}{n} \sum_{\mu=1}^m \ell_\delta(g_\mu - z_\mu)\Big\}.
        \end{align}
    \end{lemma}
    Note that we added a constraint over $\norm{\blambda}$ in the auxiliary problem, so that the set of $\blambda$ considered is compact. This allows to deduce, using Proposition~\ref{prop:gaussian_minmax} (and Remark~I below) in eqs.~\eqref{eq:ground_state_duality_smoothed} and \eqref{eq:ub_G_GK}:
    \begin{lemma}\label{lemma:application_minmax}
        \noindent
        For all $\delta > 0$, and all $t \in \bbR$, one has 
        \begin{align*}
            \bbP_\bW[G(\bW) \leq t] &\leq 2 \bbP_{\bg,\bh}[\mcC_{\delta}(\bg,\bh) \leq t] + e^{-n},
        \end{align*}
        with $\mcC_{\delta}(\bg, \bh)$ the RHS of eq.~\eqref{eq:def_Cd_gh}, and $\bg, \bh$ vectors with i.i.d.\ $\mcN(0,1)$ coordinates.
    \end{lemma}
    Lemmas~\ref{lemma:ao_simplification} and \ref{lemma:application_minmax} are proven
     in Section~\ref{subsubsec:proof_lemma_application_minimax}.
    We are now ready to prove Theorem~\ref{thm:bound_Gordon}.
    Note that by weak duality:
    \begin{align*}
         \mcC_{\delta}(\bg, \bh) &= \inf_{\bz \in \bbR^m} \sup_{\kappa \geq 0} \Big\{\frac{\kappa}{n} (\norm{\bz}^2 - \norm{\bh}^2) + \frac{1}{n} \sum_{\mu=1}^m \ell_\delta(g_\mu - z_\mu)\Big\}, \\ 
         & \geq \sup_{\kappa \geq 0}  \inf_{\bz \in \bbR^m} \Big\{\frac{\kappa}{n} (\norm{\bz}^2 - \norm{\bh}^2) + \frac{1}{n} \sum_{\mu=1}^m \ell_\delta(g_\mu - z_\mu)\Big\} \coloneqq \mathcal{M}_{\delta}(\bg,\bh).
    \end{align*}
    Therefore $\bbP[\mcC_{\delta}(\bg,\bh) \leq t] \leq \bbP[\mcM_{\delta}(\bg,\bh) \leq t]$.
    Moreover, by $\norm{\bz}^2=\sum z_\mu^2$, one has:
    \begin{align}\label{eq:M_gh}
      \mcM_{\delta}(\bg,\bh) &= \sup_{\kappa \geq 0} \Big\{- \frac{\kappa}{n} \norm{\bh}^2 + \frac{1}{n} \sum_{\mu=1}^m \inf_{z \in \bbR} \{\kappa z^2 + \ell_\delta(g_\mu - z) \}\Big\}.
    \end{align}
    Let us show 
    \begin{align}\label{eq:to_show_Mgh}
       \mcM_{\delta}(\bg, \bh) \pto \mcM_{\delta} \coloneqq \sup_{\kappa \geq 0} \Big\{- \kappa + \alpha \int_{\bbR} \mcD x \, \Big[\inf_{z \in \bbR} \{\kappa z^2 + \ell_\delta(x - z) \}\Big]\Big\}.
    \end{align}    
    We can assume $\norm{\bh}^2/n \geq 1/2$, an event that has probability $1-\smallO_n(1)$.
    Denoting $f(\kappa,\bg,\bh)$ the maximized function in eq.~\eqref{eq:M_gh}, we then have\footnote{Indeed, $\inf_{z \in \bbR} \{\kappa z^2 + \ell_\delta(x - z) \} \leq \ell_\delta(x) \leq 1$.}
    $f(\kappa,\bg,\bh) \leq 2 \alpha - \kappa/2$ and $f(0, \bg, \bh) \geq 0$, so that we can write $\mcM_{\delta}(\bg,\bh) = \max_{0\leq\kappa\leq 4\alpha} f(\kappa,\bg,\bh)$.
    Letting 
    \begin{align*}
        f_\infty(\kappa) \coloneqq - \kappa + \alpha \int_{\bbR} \mcD x \, \Big[\inf_{z \in \bbR} \{\kappa z^2 + \ell_\delta(x - z) \}\Big],
    \end{align*}
    we have then for all $\kappa \in [0,4\alpha]$:
    \begin{align}\label{eq:f_finfty_kappa}
        &|f(\kappa,\bg,\bh) - f_\infty(\kappa)| \leq 4 \alpha \Big|\frac{\norm{\bh}^2}{n} - 1\Big| + \frac{\alpha}{m} \Bigg|\sum_{\mu=1}^m (X_\mu - \EE[X_\mu])\Bigg|, \\
        \nonumber
        &X_\mu \coloneqq \inf_{z \in \bbR} \{\kappa z^2 + \ell_\delta(g_\mu - z) \}.
    \end{align}
    Note that $\{X_\mu\}$ are i.i.d.\ random variables, and one shows easily that $X_{\mu} \in [0,1]$, so by Hoeffding's inequality, for all $t > 0$: 
    \begin{align*}
        \bbP\Bigg[\frac{1}{m} \Bigg|\sum_{\mu=1}^m (X_\mu - \EE[X_\mu])\Bigg| \geq t\Bigg] &\leq 2 e^{- 2 m t^2}.
    \end{align*}
    Plugging it in eq.~\eqref{eq:f_finfty_kappa} and using the concentration of $\|\bh\|^2/n$, we reach that for all $t > 0$:
    \begin{align*}
        \lim_{n\to \infty}\sup_{\kappa \in [0, 4\alpha]}\bbP[|f(\kappa,\bg,\bh) - f_\infty(\kappa)| \geq t] &= 0.
    \end{align*}
    It is elementary to check that this implies 
    $\max_{0\leq\kappa\leq 4\alpha} f(\kappa,\bg,\bh) \pto \max_{0\leq\kappa\leq 4\alpha} f_\infty(\kappa)$, and therefore 
    eq.~\eqref{eq:to_show_Mgh}.
    By Lemma~\ref{lemma:application_minmax} we have then shown that for any $t, \delta > 0$, 
    \begin{align}\label{eq:implication_Md_G}
        \mcM_{\delta} > t \Rightarrow \lim_{n \to \infty} \bbP[G(\bW) > t] = 1.
    \end{align}
    We will then conclude by considering the limit $\delta \to 0$:
    \begin{lemma}\label{lemma:Md_limit}
        \noindent
        We have $\lim_{\delta \to 0} \mcM_{\delta} = \mcM$, 
        with 
        \begin{align}\label{eq:def_M}
            \nonumber
            \mcM &\coloneqq \sup_{\kappa \geq 0} \Big\{- \kappa + \alpha \int_{\bbR} \mcD x \, \Big[\inf_{z \in \bbR} \{\kappa z^2 + \indi\{x > z\} \}\Big]\Big\}, \\
            &=\sup_{\kappa \geq 0} \Big\{- \kappa + \alpha \int_{1/\sqrt{\kappa}}^\infty \mcD x + \alpha \kappa \int_0^{1/\sqrt{\kappa}} \mcD x \, x^2    \Big\}.
        \end{align}
    \end{lemma}
    Moreover, the maximum in eq.~\eqref{eq:def_M} is reached in $\kappa^\star$ such that: 
    \begin{align*}
        1 &= \alpha \int_0^{1/\sqrt{\kappa^\star}} \mcD x \, x^2.
    \end{align*}
    And the limit is then given by:
    \begin{align*}
        \mcM &= \alpha \int_{1/\sqrt{\kappa^\star}}^\infty \mcD x.
    \end{align*}
    We recognize the replica-symmetric prediction of eq.~\eqref{eq:fstar_RS}, with $\kappa = (2 \chi_\RS)^{-1}$!
    By Lemma~\ref{lemma:Md_limit} and eq.~\eqref{eq:implication_Md_G}, 
    we showed that $\mcM > t$ implies that $\bbP[G(\bW) > t] \to 1$ as $n \to \infty$.
    Applying it for $t = 1$ ends the proof of Theorem~\ref{thm:bound_Gordon}. \qed

    \subsection{Proof of Lemmas~\ref{lemma:ao_simplification}, \ref{lemma:application_minmax} and \ref{lemma:Md_limit}}\label{subsubsec:proof_lemma_application_minimax}

    \begin{proof}[Proof of Lemma~\ref{lemma:ao_simplification} --]
    First note that in eq.~\eqref{eq:def_CAd_gh}, writing $\blambda = \tau \bbe$ with $\norm{\bbe} = 1$, one can perform the supremum over $\bbe$: 
      \begin{align*}
              \mcC_{A, \delta}(\bg, \bh) &= \inf_{\bz \in \bbR^m} \min_{\bx \in \mcS^{n-1}} \max_{\tau \in [0,A]} \Big\{ \tau \bh^\intercal \bx + \tau \norm{\bg - \bz} + \frac{1}{n} \sum_{\mu=1}^m \ell_\delta(z_\mu)\Big\}.
      \end{align*}
      The maximum over $\tau \in [0,A]$ and the minimum over $\bx$ can be carried out explicitly:
      \begin{align*}
              \mcC_{A, \delta}(\bg, \bh) &= 
              \inf_{\bz \in \bbR^m} \Big\{\frac{1}{n} \sum_{\mu=1}^m \ell_\delta(z_\mu) + \min_{\bx \in \mcS^{n-1}} [A(\bh^\intercal \bx + \| \bg - \bz\|)\indi\{\bh^\intercal \bx+\| \bg - \bz\| > 0 \}] \Big\}, \\ 
            &= 
            \inf_{\bz \in \bbR^m} \Big\{\frac{1}{n} \sum_{\mu=1}^m \ell_\delta(z_\mu) +  A(-\|\bh\|  + \| \bg - \bz\|)\indi\{-\|\bh \|+\| \bg - \bz\| > 0 \} \Big\}.
      \end{align*}
    Letting $\bz' = \bg - \bz$, this yields:
    \begin{align*}
            \mcC_{A, \delta}(\bg, \bh) &= \min\Bigg\{\min_{\substack{\bz \in \bbR^m \\ \norm{\bz} \leq \norm{\bh}}} \Big[\frac{1}{n} \sum_{\mu=1}^m \ell_\delta(g_\mu - z_\mu)\Big] \, , \, \inf_{\substack{\bz \in \bbR^m \\ \norm{\bz} > \norm{\bh}}} \Big[\frac{1}{n} \sum_{\mu=1}^m \ell_\delta(g_\mu - z_\mu) + A (\norm{\bz} - \norm{\bh}) \Big]\Bigg\}.
    \end{align*}
    We now show that:
    \begin{align}\label{eq:to_show_A_infty}
        \lim_{A \to \infty} \inf_{\substack{\bz \in \bbR^m \\ \norm{\bz} > \norm{\bh}}} \Big[\frac{1}{n} \sum_{\mu=1}^m \ell_\delta(g_\mu - z_\mu) + A (\norm{\bz} - \norm{\bh}) \Big] 
        &\geq \min_{\substack{\bz \in \bbR^m \\ \norm{\bz} \leq  \norm{\bh}}} \Big[\frac{1}{n} \sum_{\mu=1}^m \ell_\delta(g_\mu - z_\mu)\Big],
    \end{align}
    which ends the proof.
    Notice that the LHS of eq.~\eqref{eq:to_show_A_infty} is obviously a non-decreasing function of $A$, so that it indeed has a limit (possibly $+\infty$).
    Moreover, we can restrict the infimum to $\norm{\bz} \leq \norm{\bh} + \alpha / A$, since trivially for all $A$ one has (recall $\ell_\delta \leq 1$): 
    \begin{align*}
        \inf_{\substack{\bz \in \bbR^m \\ \norm{\bh}<\norm{\bz} \leq  \norm{\bh} + \alpha / A}} \Big[\frac{1}{n} \sum_{\mu=1}^m \ell_\delta(g_\mu - z_\mu) &+ A (\norm{\bz} - \norm{\bh}) \Big] 
        \overset{(\rm a)}{\leq}  \inf_{\substack{\bz \in \bbR^m \\ \norm{\bh}<\norm{\bz} \leq  \norm{\bh} + \alpha / A}} \Big[\alpha + A (\norm{\bz} - \norm{\bh}) \Big], \\
        &\leq \alpha \overset{(\rm b)}{\leq}
        \inf_{\substack{\bz \in \bbR^m \\ \norm{\bz} > \norm{\bh} + \alpha / A}} \Big[\frac{1}{n} \sum_{\mu=1}^m \ell_\delta(g_\mu - z_\mu) + A (\norm{\bz} - \norm{\bh}) \Big],
    \end{align*}
    in which we used in $(\rm a)$ and $(\rm b)$ that $\ell_\delta(x)\in[0,1]$.
    We let $\varepsilon > 0$, and for all $A > 0$ we fix $\tilde{\bz}^{(A)} \in \bbR^m$ with $\|\tilde{\bz}^{(A)}\| \in (\|\bh\| , \|\bh\| + \alpha/A]$ such that:
    \begin{align*}
        \frac{1}{n} \sum_{\mu=1}^m \ell_\delta(g_\mu - \tilde{z}^{(A)}_\mu) + A (\norm{\tilde{\bz}^{(A)}} - \norm{\bh}) 
        &\leq 
        \inf_{\substack{\bz \in \bbR^m \\ \norm{\bz}>\norm{\bh} }} \Big[\frac{1}{n} \sum_{\mu=1}^m \ell_\delta(g_\mu - z_\mu) + A (\norm{\bz} - \norm{\bh}) \Big] + \varepsilon.
    \end{align*}
    Since $\norm{\tilde{\bz}^{(A)}} \leq \norm{\bh} + \alpha / A$, we can extract a converging subsequence $\bz^{(k)}=\tilde{\bz}^{A(k)}$ such that $\exists\lim_{k\to\infty}\bz^{(k)} =: \bz^*$, 
    and $\norm{\bh} < \norm{\bz^{(k)}} \leq \norm{\bh} + \alpha/A(k)$, with $A(k) \to \infty$.
    Therefore $\norm{\bz^*} = \norm{\bh}$. Moreover:
    \begin{align*}
        \sum_{\mu=1}^m \ell_\delta(g_\mu - z^*_\mu) &= \lim_{k \to \infty} \sum_{\mu=1}^m \ell_\delta(g_\mu - z^{(k)}_\mu), \\ 
        &\leq \liminf_{k \to \infty} \Big[\frac{1}{n} \sum_{\mu=1}^m \ell_\delta(g_\mu - z^{(k)}_\mu) + A(k) (\norm{\bz^{(k)}} - \norm{\bh}) \Big], \\
        &\leq \liminf_{k \to \infty}\Big\{\inf_{\substack{\bz \in \bbR^m\\  \norm{\bz}>\norm{\bh}}} \Big[\frac{1}{n} \sum_{\mu=1}^m \ell_\delta(g_\mu - z_\mu) + A(k) (\norm{\bz} - \norm{\bh}) \Big]\Big\} + \varepsilon, \\
        &\leq 
        \lim_{A \to \infty} \inf_{\substack{\bz \in \bbR^m \\ \norm{\bz} > \norm{\bh}}} \Big[\frac{1}{n} \sum_{\mu=1}^m \ell_\delta(g_\mu - z_\mu) + A (\norm{\bz} - \norm{\bh}) \Big] + \varepsilon.
    \end{align*}
    Letting $\varepsilon > 0$ be arbitrarily small, the claim of eq.~\eqref{eq:to_show_A_infty} follows.
    \end{proof}

    \myskip 
    \begin{proof}[Proof of Lemma~\ref{lemma:application_minmax} --]
        Recall eq.~\eqref{eq:ground_state_duality_smoothed}. In particular, for any $A, \delta, K > 0$ we have: 
        \begin{align}\label{eq:def_GAdK}
            G_K(\bW) \geq G_{A, \delta, K}(\bW) &\coloneqq  \inf_{\bz \in B(K)} \min_{\bx \in \mcS^{n-1}} \max_{\norm{\blambda} \leq A} \Big\{\blambda^\intercal \bW \bx - \blambda^\intercal \bz + \frac{1}{n} \sum_{\mu=1}^m \ell_\delta(z_\mu)\Big\}.
        \end{align}
        By eq.~\eqref{eq:ub_G_GK} and eq.~\eqref{eq:def_GAdK}, we have:
        \begin{align}\label{eq:domination_G_GAdK}
           \bbP_\bW[G(\bW) \leq t] &\leq \bbP_\bW[G_K(\bW) \leq t] + e^{-n} \leq \bbP_\bW[G_{A, \delta, K}(\bW)\leq t ] + e^{-n}. 
        \end{align}
        Using Proposition~\ref{prop:gaussian_minmax} (since all sets are compact and functions involved are continuous, see in particular Remark~I below it)
        we have, for all $t \in \bbR$:
        \begin{align*}
           \bbP_\bW[G_{A, \delta, K}(\bW, \bz)\leq t ] &\leq 2 \bbP_{\bg,\bh}[\mcC_{A, \delta, K}(\bg,\bh,\bz) \leq t],
        \end{align*}
        in which $\mcC_{A, \delta, K}(\bg, \bh)$ is defined as in eq.~\eqref{eq:def_CAd_gh}, restricting furthermore the infimum to $\bz \in B(K)$.
        In particular, $\mcC_{A, \delta, K}(\bg, \bh) \geq \mcC_{A, \delta}(\bg, \bh)$.
        Therefore by eq.~\eqref{eq:domination_G_GAdK}:
        \begin{align}
            \label{eq:domination_G_CAd}
           \bbP_\bW[G(\bW) \leq t]&\leq 2 \bbP_{\bg,\bh}[\mcC_{A, \delta}(\bg,\bh) \leq t] + e^{-n}.
        \end{align}
        Note that $\bbP_{\bg,\bh}[\mcC_{A, \delta}(\bg,\bh) \leq t] = \EE_{\bg,\bh} [\indi\{\mcC_{A, \delta}(\bg,\bh) \leq t\}]$, 
        and moreover by Lemma~\ref{lemma:ao_simplification}\footnote{We use there the fact that $A \mapsto \mcC_{A, \delta}(\bg, \bh)$ is non-decreasing}: 
        \begin{align*}
            \lim_{A \to \infty} \indi\{\mcC_{A, \delta}(\bg,\bh) \leq t\} = \indi\{\mcC_\delta(\bg,\bh) \leq t\}.
        \end{align*}
        Taking the $A \to \infty$ limit in eq.~\eqref{eq:domination_G_CAd} and using the dominated convergence theorem 
        ends the proof of Lemma~\ref{lemma:application_minmax}.
    \end{proof}

    \myskip 
    \begin{proof}[Proof of Lemma~\ref{lemma:Md_limit} --]
        For $\delta\geq 0$, we define
        \begin{align*}
            f_\delta(\kappa) &\coloneqq - \kappa + \alpha \int_{\bbR} \mcD x \, \Big[\inf_{z \in \bbR} \{\kappa z^2 + \ell_\delta(x - z) \}\Big],  
        \end{align*}
        so that $\mcM_\delta = \sup_{\kappa \geq 0} f_\delta(\kappa)$ for $\delta>0$, and 
        $\mcM = \sup_{\kappa \geq 0} f_0(\kappa)$.
        Notice first that eq.~\eqref{eq:def_M} follows from the following identity, that can be easily checked:
        \begin{align*}
            \inf_{z \in \bbR} \{\kappa z^2 + \indi\{z < x\} \} &= \indi\{\sqrt{\kappa} x \geq 1\} + \indi\{\sqrt{\kappa} x \in (0,1)\} \kappa x^2.
        \end{align*}
        Lemma~\ref{lemma:Md_limit} will follow if we can show: 
        \begin{align}\label{eq:to_show_Md_M}
            \lim_{\delta \to 0} \sup_{\kappa \geq 0} |f_\delta(\kappa) - f_0(\kappa)| = 0.
        \end{align}
        Notice that for all $\delta > 0$ and all $x \in \bbR$, we have
        $\indi\{x > \delta \} \leq  \ell_\delta(x) \leq \indi\{x > 0\}$.
        In particular, 
        \begin{align*}
            \inf_{z \in \bbR} \{\kappa z^2 + \indi\{z < x - \delta\} \} \leq \inf_{z \in \bbR} \{\kappa z^2 + \ell_\delta(x - z) \} \leq \inf_{z \in \bbR} \{\kappa z^2 + \indi\{z < x\} \}.
        \end{align*}
        One computes easily the left and right sides of this inequality:
        \begin{align*}
            \begin{dcases}
                \inf_{z \in \bbR} \{\kappa z^2 + \ell_\delta(x - z) \} &\leq \indi\{\sqrt{\kappa} x \geq 1\} + \indi\{\sqrt{\kappa} x \in (0,1)\} \kappa x^2, \\
                \inf_{z \in \bbR} \{\kappa z^2 + \ell_\delta(x - z) \} &\geq \indi\{\sqrt{\kappa} (x-\delta) \geq 1\} + \indi\{\sqrt{\kappa} (x-\delta) \in (0,1)\} \kappa (x-\delta)^2.
            \end{dcases}
        \end{align*}
        Therefore we reach:
        \begin{align*}
            |f_\delta(\kappa) - f_0(\kappa)| &\leq \alpha \int_{1/\sqrt{\kappa}}^{1/\sqrt{\kappa} + \delta} \mcD x + \alpha \kappa \int_{0}^{1/\sqrt{\kappa}} \frac{\rd x}{\sqrt{2 \pi}} \, x^2 \, \Big[e^{-x^2/2} - e^{-(x+\delta)^2/2}\Big], \\ 
            &\leq \alpha \int_0^\delta \mcD x + \alpha \int_{0}^{1/\sqrt{\kappa}} \frac{\rd x}{\sqrt{2 \pi}} \, \Big[e^{-x^2/2} - e^{-(x+\delta)^2/2}\Big], \\
            &\leq \alpha \int_0^\delta \mcD x + \alpha \int_{0}^{\infty} \frac{\rd x}{\sqrt{2 \pi}} \, \Big[e^{-x^2/2} - e^{-(x+\delta)^2/2}\Big], \\
            &\leq 2\alpha \int_0^\delta \mcD x,
        \end{align*}
        which goes to $0$ as $\delta \to 0$, uniformly in $\kappa$. This ends the proof.
    \end{proof}

\subsection{Improving over a standard use of Gordon's inequality?}\label{subsec_app:improvement_gordon}

We give here a brief and informal description of the improvements made in \cite{stojnic2013negative,montanari2024tractability} over the standard use of Gordon's inequality.
The starting point of this improvement is to use Gordon's inequality in the form of stochastic domination, cf. e.g.\ Theorem~1 of \cite{thrampoulidis2015regularized}.
This form implies that for $\bG$ a Gaussian i.i.d.\ matrix and $z \sim \mcN(0,1)$, $\bg, \bh \sim \mcN(0, \Id_d)$, defining the min-max problems
\begin{align*}
    \begin{dcases}
    \hat{\xi} &\coloneqq \min_\bw \max_\bu [\bu^\T \bG \bw + \psi(\bw, \bu)], \\
    \xi &\coloneqq \min_\bw \max_\bu [\bu^\T \bG \bw + \psi(\bw, \bu) + z \|\bu\| \|\bw\|], \\
    \xi_\mathrm{lin.} &\coloneqq \min_\bw \max_\bu [\|\bg\|\bu^\T \bh + \|\bh\| \bw^\T \bg + \psi(\bw, \bu)],
    \end{dcases}
\end{align*}
then for any $c \geq 0$, we have
\begin{equation*}
    \EE[\exp\{-c \xi\}] \leq \EE[\exp\{- c \xi_\mathrm{lin.}\}].
\end{equation*}
For our purposes, we would like to obtain an upper bound for $\bbP[\hat{\xi} \leq t]$, cf.\ eq.~\eqref{eq:ground_state_duality_smoothed}.
In the settings of~\cite{stojnic2013negative,montanari2024tractability}, the vectors $\bw, \bu$ are unit-normed, so that 
$\xi = \hat{\xi} + z$, and therefore
\begin{equation*}
    \EE[e^{-c \hat{\xi}}] = \frac{\EE[e^{- c\xi}]}{\EE[e^{-cz}]} \leq \frac{\EE[\exp\{- c \xi_\mathrm{lin.}\}]}{\EE[e^{-cz}]}.
\end{equation*}
Markov's inequality yields then:
\begin{equation*}
    \bbP[\hat{\xi} \leq t] \leq e^{ct} \EE[e^{-c \hat{\xi}}] \leq e^{ct} \frac{\EE[\exp\{- c \xi_\mathrm{lin.}\}]}{\EE[e^{-cz}]}.
\end{equation*}
The authors of \cite{stojnic2013negative,montanari2024tractability} use this bound and optimize it over $c \geq 0$.
However, in our setting the vectors are not unit-normed, cf eq.~\eqref{eq:ground_state_duality_smoothed}.
Instead, we use in Proposition~\ref{prop:gaussian_minmax} a form of Gordon's inequality that is weaker than stochastic domination, but that allows to directly compare $\hat{\xi}$ and $\xi_{\mathrm{lin.}}$:
    \begin{equation*}
    \bbP[\hat{\xi} \leq t] \leq 2 \bbP[\xi_\mathrm{lin.} \leq t].
    \end{equation*}
This technicality prevents us from directly applying the 
    inequality used in~\cite{stojnic2013negative,montanari2024tractability}: nevertheless, similar improvements over the standard application of the min-max inequality are likely 
    possible, and we leave it as an open question for future work.

\section{Replica-symmetric supplementary calculations}\label{sec_app:technicalities}
\subsection{Zero-temperature limit of the replica-symmetric solution}\label{subsec_app:zerotemp_RS}

In this section, we derive eqs.~\eqref{eq:chi_RS} and \eqref{eq:fstar_RS}. 
Our arguments will sometimes be informal, and a rigorous treatment would demand more care.

\myskip
Recall that we have the expansion of eq.~\eqref{eq:q_RS_zerotemp}, with 
$\chi_\RS$ the zero-temperature susceptibility of the system.
In this section, we often drop the $\RS$ subscript on quantities to lighten the notations.
We use the expansion of $H(x) = \int_x^\infty \mcD u$ for large $x \gg 1$:
\begin{align}\label{eq:expansion_H}
    H(x) &= \frac{e^{-\frac{x^2}{2}}}{\sqrt{2 \pi}} \Big[\frac{1}{x} + \mathcal{O}_{x \to \infty}\Big(\frac{1}{x^3}\Big)\Big].
\end{align}
\textbf{Computation of $f^\star_\RS(\alpha)$ --}
We start by deriving eq.~\eqref{eq:fstar_RS}.
As one can check from eq.~\eqref{eq:phi_RS} that $\Phi_\RS(\alpha,\beta)$ is a differentiable function of $\beta$, 
by L'Hospital's rule we have $f^\star_\RS(\alpha) = \lim_{\beta\to\infty} e^\star_\RS(\alpha, \beta)$, 
with $e^\star_\RS(\alpha,\beta) \coloneqq - \partial_\beta \Phi(\alpha,\beta)$.
We have from eq.~\eqref{eq:q_RS_zerotemp}:
\begin{align*}
    \sqrt{\frac{q}{1-q}} &= \sqrt{\frac{\beta}{\chi}} + \mathcal{O}(\beta^{-1/2}),
\end{align*}
We compute the limit of the integrand in eq.~\eqref{eq:estar_beta} (changing variables $\xi \to - \xi$):
\begin{align*}
    \frac{e^{-\beta} H \Big(-\xi \sqrt{\frac{q}{1-q} }\Big)}{1 - (1 - e^{-\beta}) H \Big(-\xi \sqrt{\frac{q}{1-q} }\Big)} 
    &\simeq 
    \frac{e^{-\beta} H \Big(-\xi \sqrt{\frac{\beta}{\chi} }\Big)}{1 - (1 - e^{-\beta}) H \Big(-\xi \sqrt{\frac{\beta}{\chi} }\Big)}.
\end{align*}
We separate three cases, and use the expansion of eq.~\eqref{eq:expansion_H} to reach that at leading order 
in $\beta$: 
\begin{align}\label{eq:expansion_estar_RS}
    \frac{e^{-\beta} H \Big(-\xi \sqrt{\frac{\beta}{\chi} }\Big)}{1 - (1 - e^{-\beta}) H \Big(-\xi \sqrt{\frac{\beta}{\chi} }\Big)} 
    &\simeq \begin{dcases}
        \frac{\sqrt{\chi}}{\sqrt{2 \pi \beta} |\xi|}e^{-\beta - \frac{\beta \xi^2}{2 \chi}} \to_{\beta \to \infty} 0 & \textrm{ if } \xi < 0, \\
        \frac{\xi \sqrt{2 \pi \beta}}{\sqrt{\chi}}  e^{-\beta + \frac{\beta \xi^2}{2 \chi}} \to_{\beta \to \infty} 0 & \textrm{ if } \xi \in (0, \sqrt{2 \chi}), \\
        1 & \textrm{ if } \xi > \sqrt{2 \chi}.
    \end{dcases}
\end{align}
Using the pointwise limit above, we reach (as we mentioned above, a more careful argument would need to be carried out to make this expansion rigorous)
\begin{align*}
    e^\star(\alpha,\beta)
    &\simeq \alpha \int_{\sqrt{2\chi}}^\infty \, \mcD \xi = \alpha H[\sqrt{2\chi}].
\end{align*}
In the end, we reach eq.~\eqref{eq:fstar_RS}:
\begin{align*}
    f^\star_\RS(\alpha) = \lim_{\beta \to \infty} e^\star_\RS(\alpha,\beta) &= \alpha H[\sqrt{2 \chi_\RS}].
\end{align*}
\textbf{Computing $\chi$ --}
There now remains to find $\chi$ as a function of $\alpha$, from eq.~\eqref{eq:q_RS_eq_new}.
Plugging in the expansion of eq.~\eqref{eq:q_RS_zerotemp} we find (changing $\xi \to - \xi$):
\begin{align}\label{eq:chi_RS_1}
   \frac{1}{\sqrt{\chi}} 
&= -\alpha \lim_{\beta \to \infty} \frac{1}{\sqrt{\beta}} \int \mcD \xi \frac{(1-e^{-\beta}) \xi H'\Big(-\xi \sqrt{\frac{q}{1-q}} \Big)}{1-(1-e^{-\beta}) H\Big(-\xi \sqrt{\frac{q}{1-q}} \Big)}.
\end{align}
In the same way as in eq.~\eqref{eq:expansion_estar_RS}, we can show: 
\begin{align}\label{eq:expansion_chi_RS}
- \frac{1}{\sqrt{\beta}} \frac{(1-e^{-\beta}) \xi H'\Big(-\xi \sqrt{\frac{q}{1-q}} \Big)}{1-(1-e^{-\beta}) H\Big(-\xi \sqrt{\frac{q}{1-q}} \Big)}
    &\simeq \begin{dcases}
        \frac{\xi}{\sqrt{2 \pi\beta}} \, e^{-\frac{\beta \xi^2}{2 \chi}} \to_{\beta \to \infty} 0 & \textrm{ if } \xi < 0, \\
        \frac{\xi^2}{\sqrt{\chi}} & \textrm{ if } \xi \in (0, \sqrt{2 \chi}), \\
        \frac{\xi}{\sqrt{2 \pi\beta}} \, e^{\beta-\frac{\beta \xi^2}{2 \chi}} \to_{\beta \to \infty} 0 & \textrm{ if } \xi > \sqrt{2 \chi}.
    \end{dcases}
\end{align}
Therefore, we reach from eq.~\eqref{eq:chi_RS_1} that, as $\beta \to \infty$: 
\begin{align*}
   \alpha \int_0^{\sqrt{2\chi}} \mcD \xi \, \xi^2 &= 1, 
\end{align*}
which is eq.~\eqref{eq:chi_RS}.

\subsection{Stability of the replica-symmetric solution}\label{subsec_app:stability_rs}

In this section we follow Appendix~4 of \cite{engel2001statistical} (see also e.g.\ \cite{urbani2018statistical}) to characterize the stability of the RS solution in replica space. 
This gives rise to the so-called de Almeida-Thouless conditions \cite{de1978stability,gardner1988optimal}, which is a criterion for stability 
expressed in terms of so-called \emph{replicon eigenvalues}.

\myskip
We start again from the general expression of eq.~\eqref{eq:Phir_final}: 
$\Phi(\alpha,\beta;r) = \sup_{\bQ} G_r(\bQ)$, with
\begin{align}\label{eq:Gr}
   G_r(\bQ) \coloneqq \frac{1}{2} \log \det \bQ + \alpha \log \int_{\bbR^r} \frac{\rd \bz}{(2\pi)^{r/2} \sqrt{\det \bQ}} e^{-\frac{1}{2} \bz^\intercal \bQ^{-1} \bz -\beta \sum_{a=1}^r \theta(z^a)} = G_{1,r}(\bQ) + \alpha G_{2,r}(\bQ).
\end{align}
In what follows, we compute the Hessian of $G_r(\bQ)$ taken at the replica-symmetric point.

\subsubsection{The derivatives of \texorpdfstring{$G_{1,r}$}{G1r}}

The derivatives of $G_{1,r}(\bQ)$ can be worked out in terms of the matrix elements of $\bQ^{-1}$ (here $a<b$ and $c < d$):
\begin{align*}
    \begin{dcases}
        \frac{\partial G_{1,r}}{\partial Q_{ab}} &= Q^{-1}_{ab}, \\ 
        \frac{\partial^2 G_{1,r}}{\partial Q_{ab} \partial Q_{cd}} &= -[Q^{-1}_{ac} Q^{-1}_{bd} + Q^{-1}_{ad} Q^{-1}_{bc}].
    \end{dcases}
\end{align*}
Recall that at the replica symmetric point with $Q_{ab} = q$ and $Q_{aa} = 1$ we have
\begin{align}\label{eq:Qm1_RS}
    \begin{dcases}
        Q^{-1}_{aa} &= \frac{1 + (r-2)q}{(1-q) [1+(r-1)q]}, \\
        Q^{-1}_{ab} &= -\frac{q}{(1-q) [1+(r-1)q]}.
    \end{dcases}
\end{align}
Therefore (taking the notations of \cite{engel2001statistical}):
\begin{align}\label{eq:HessG1_RS}
    \Bigg[\frac{\partial^2 G_{1,r}}{\partial Q_{ab} \partial Q_{cd}}\Bigg]_{\mathrm{RS}} &=
    \begin{dcases}
        P_1 & \textrm{ if } a = c \, ; b = d , \\
        Q_1 & \textrm{ if } a = c \, ; b \neq d \textrm{ or } b = c \textrm{ or } a \neq c \, ; b = d \textrm{ or } a = d , \\
        R_1 & \textrm{ if all indices are distinct},
    \end{dcases}
\end{align}
in which $P_1,Q_1,R_1$ are defined as: 
\begin{align*}
    \begin{dcases}
        P_1 &\coloneqq -\Bigg(\frac{1 + (r-2)q}{(1-q) [1+(r-1)q]}\Bigg)^2 - \Bigg(\frac{q}{(1-q) [1+(r-1)q]}\Bigg)^2, \\ 
        Q_1 &\coloneqq -\Bigg(\frac{1 + (r-2)q}{(1-q) [1+(r-1)q]}\Bigg)\Bigg(-\frac{q}{(1-q) [1+(r-1)q]}\Bigg) - \Bigg(\frac{q}{(1-q) [1+(r-1)q]}\Bigg)^2, \\ 
        R_1 &\coloneqq - 2 \Bigg(\frac{q}{(1-q) [1+(r-1)q]}\Bigg)^2.
    \end{dcases}
\end{align*}
We now take the limit $r \downarrow 0$. With an abuse of notation, we still denote the limits $P_1,Q_1,R_1$:
\begin{align}\label{eq:PQR_1}
    \begin{dcases}
        P_1 &= \frac{-1 + 4 q(1-q) - q^2}{(1-q)^4}, \\ 
        Q_1 &= \frac{q (1-q) - 2 q^2}{(1-q)^4}, \\ 
        R_1 &= -\frac{2 q^2}{(1-q)^4}.
    \end{dcases}
\end{align}

\subsubsection{The derivatives of \texorpdfstring{$G_{2,r}$}{G2r}}

We now turn to $G_{2,r}(\bQ)$, that we rewrite using a Gaussian transformation:
\begin{align}\label{eq:G2_alternate}
   G_{2,r}(\bQ) &= \log \int_{\bbR^r} \frac{\rd \bu \rd \bv}{(2\pi)^{r}} e^{-\frac{1}{2} \sum_{a,b} Q^{ab} v^a v^b -\beta \sum_{a=1}^r \theta(u^a) + i \sum_{a=1}^r u^a v^a}.
\end{align}
This form is more suitable for computing the Hessian with respect to $\bQ$.
In order to write the formulas compactly, we introduce the following average for any function of $\{v^a\}$:
\begin{align*}
   \langle g(\{v^a\}) \rangle_r &\coloneqq \Bigg\{\frac{\int_{\bbR^r} \rd \bu \, \rd \bv\, g(\{v^a\}) \, e^{-\frac{1}{2} \sum_{a,b} Q^{ab} v^a v^b -\beta \sum_{a=1}^r \theta(u^a) + i \sum_{a=1}^r u^a v^a}}{
\int_{\bbR^r} \rd \bu \, \rd \bv \, e^{-\frac{1}{2} \sum_{a,b} Q^{ab} v^a v^b -\beta \sum_{a=1}^r \theta(u^a) + i \sum_{a=1}^r u^a v^a}}\Bigg\}_{\mathrm{RS}}.
\end{align*}
With this definition, we have from eq.~\eqref{eq:G2_alternate}:
\begin{align}\label{eq:HessG2_1}
    &\Bigg[\frac{\partial^2 G_{2,r}}{\partial Q_{ab} \partial Q_{cd}}\Bigg]_{\mathrm{RS}} =
    \langle v^a v^b v^c v^d \rangle - \langle v^a v^b \rangle \langle v^c v^d \rangle,
\end{align}
in which $a < b$ and $c < d$.
One can easily see that this Hessian has the same ``replica-symmetric'' structure as the one of $G_{1,r}$:
\begin{align}\label{eq:HessG2_RS}
   \Bigg[\frac{\partial^2 G_{2,r}}{\partial Q_{ab} \partial Q_{cd}}\Bigg]_{\mathrm{RS}} &=
   \begin{dcases}
      P_2 & \textrm{ if } a = c \, ; b = d , \\
      Q_2 & \textrm{ if } a = c \, ; b \neq d \textrm{ or } b = c \textrm{ or } a \neq c \, ; b = d \textrm{ or } a = d , \\
      R_2 & \textrm{ if all indices are distinct}.
   \end{dcases}
\end{align}
We compute these three terms separately in the limit $r \to 0$. 
In order to simplify the results, we introduce the notation $\EE \langle g(v) \rangle$, with $\EE$ the expectation over $\xi \sim \mcN(0,1)$, and 
\begin{align*}
   \langle g(v) \rangle \coloneqq \frac{\int \rd u \, \rd v\,g(v) \, e^{-\frac{1-q}{2} v^2 - \beta \theta(u) + i v [u - \sqrt{q}\xi]}}{\int \rd u \, \rd v\, e^{-\frac{1-q}{2} v^2 - \beta \theta(u) + i v [u - \sqrt{q}\xi]}}.
\end{align*}
From this definition and eq.~\eqref{eq:HessG2_1}, one can check (using the same trick to decouple the replicas we used in the RS calculation, cf.\ Section~\ref{subsec:rs}) that we have, 
as $r \to 0$:
\begin{align}\label{eq:PQR_2}
    \begin{dcases}
        P_2 &= \EE[\langle v^2\rangle^2] - \EE[\langle v\rangle^2]^2, \\ 
        Q_2 &= \EE[\langle v^2\rangle \langle v\rangle^2] - \EE[\langle v\rangle^2]^2, \\ 
        R_2 &= \EE[ \langle v\rangle^4] - \EE[\langle v\rangle^2]^2.
    \end{dcases}
\end{align}

\subsubsection{de Almeida-Thouless condition for replica-symmetric stability}\label{subsubsec_app:AT_condition}

Classical replica studies \cite{engel2001statistical} show that for a Hessian having the form of eqs.~\eqref{eq:HessG1_RS} or \eqref{eq:HessG2_RS},
the linear stability of the RS local maximum is given by the sign of the ``replicon'' eigenvalue $P-2Q+R$.
More precisely, the AT condition for the stability of the RS solution in replica space reads here:
\begin{align*}
    \lambda_3 = [P_1 - 2Q_1+R_1] + \alpha [P_2 - 2Q_2 +R_2]\leq 0.
\end{align*}
By eqs.~\eqref{eq:PQR_1} and \eqref{eq:PQR_2} we get:
\begin{align}\label{eq:AT_condition}
    \frac{1}{(1-q)^2} \geq \alpha \EE\Big[\big(\langle v^2 \rangle - \langle v \rangle^2\big)^2\Big].
\end{align}
In order to make eq.~\eqref{eq:AT_condition} more explicit, we compute the right-hand side using the identity $\langle v^2 \rangle - \langle v \rangle^2 = - q^{-1} \partial^2_\xi \log \mcZ(\xi)$,
with 
\begin{align*}
    \mcZ(\xi) &\coloneqq \int \frac{\rd u \, \rd v}{2\pi} \, e^{-\frac{1-q}{2} v^2 - \beta \theta(u) + i v [u - \sqrt{q}\xi]}.
\end{align*}
This integral is easy to work out: 
\begin{align*}
   \log \mcZ(\xi) &= \log  \int \frac{\rd u}{\sqrt{2\pi(1-q)}} \, e^{- \beta \theta(u) - \frac{1}{2 (1-q)} (u - \sqrt{q}\xi)^2} 
   = \log \Big[1 - (1-e^{-\beta}) H \Big(- \xi \sqrt{\frac{q}{1-q} }\Big)\Big].
\end{align*}
Let us define $f_\beta(h) \coloneqq \log (1 - (1-e^{-\beta}) H[-h/\sqrt{1-q}])$, 
so that $\log \mcZ(\xi) = f_\beta(\sqrt{q} \xi)$.
Then $\langle v^2 \rangle - \langle v \rangle^2 = - f_\beta''(\sqrt{q}\xi)$.
The AT condition for the stability of the replica-symmetric solution is then expressed easily as a function of $(\alpha,q)$
at any $\beta \geq 0$ as
\begin{align}\label{eq:AT_explicit}
    \frac{1}{\alpha} \geq (1-q)^2 \int \mcD\xi f_\beta''(\sqrt{q} \xi)^2.
\end{align}

\subsubsection{The \texorpdfstring{$\beta \to \infty$}{hightemp} limit}

We now take the limit $\beta \to \infty$ in eq.~\eqref{eq:AT_explicit}, introducing the zero-temperature susceptibility $\chi_\RS = \chi$ (cf.\ eq.~\eqref{eq:q_RS_zerotemp}).
Using the same expansions as in eqs.~\eqref{eq:expansion_estar_RS} and \eqref{eq:expansion_chi_RS}
we have as $\beta \to \infty$:
\begin{align}\label{eq:expansion_fbeta}
    \frac{1}{\beta} f_\beta(h) &\simeq \begin{cases}
     0 & \textrm{ if } h < 0, \\ 
     -1 & \textrm{ if } h > \sqrt{2\chi}, \\ 
     -\frac{h^2}{2\chi} & \textrm{ otherwise }.
    \end{cases}
\end{align}
Therefore, we have at large $\beta$, that $f_\beta''(h) \simeq - \beta \chi^{-1} \indi\{h \in (0,\sqrt{2\chi})\}$.
Since $(1-q)^2 \simeq \chi^2 / \beta^2$, the RS stability condition~\eqref{eq:AT_explicit} becomes, in the $\beta \to \infty$ limit:
\begin{align}\label{eq:AT_explicit_zero_temp}
   \alpha \int_0^{\sqrt{2\chi}} \mcD \xi &\leq 1.
\end{align}
However, recall that in the zero-temperature limit, the RS susceptibility $\chi$ is given by the solution to eq.~\eqref{eq:chi_RS}: 
\begin{align*}
   \alpha \int_0^{\sqrt{2\chi}} \mcD \xi \, \xi^2 &= 1,
\end{align*}
which can be turned easily by integration by parts into:
\begin{align*}
   \alpha \int_0^{\sqrt{2\chi}} \mcD \xi  &= 1 + \alpha \sqrt{\frac{\chi}{\pi}} e^{-\chi} > 1,
\end{align*}
in which the inequality holds in all the ``UNSAT'' phase $\alpha > 2$ for which $\chi < \infty$.
Therefore, eq.~\eqref{eq:AT_explicit_zero_temp} is \emph{never} satisfied for any $\alpha > 2$: at zero-temperature, the replica-symmetric solution is never linearly stable!

\section{A replica-symmetric lower bound}\label{sec_app:rs_lower_bound}
\noindent
\textbf{High-dimensional concentration on energy level sets --}
Let us first describe physical reasons (at a heuristic level) for the concentration of the energy under the Gibbs measure of eq.~\eqref{eq:def_Gibbs}, 
for any $\beta \geq 0$. 
As we mentioned in Section~\ref{sec:introduction},
proving this property is highly non-trivial.
While we did not need to assume this concentration to hold in the rest of the paper, it will be important in 
this part to describe the derivation of a replica-symmetric lower bound for the injectivity threshold.

\myskip
At fixed $\bW$, the distribution of intensive energies is described by a probability density 
$P_\beta(e)$ given by: 
\begin{align}\label{eq:Pe}
    P_\beta(e) &\coloneqq \int \rd \bbP_{\beta,\bW}(\bx) \, \delta\Big(\frac{E_\bW(\bx)}{n} - e\Big) = \frac{e^{-n \beta e}}{\mcZ_n(\bW, \beta)}\int \mu_n(\rd \bx) \, \delta\Big(\frac{E_\bW(\bx)}{n} - e\Big).
\end{align}
One can show (using properties of the uniform measure $\mu_n$ and the fact that the energy is extensive) that the ``entropic'' term on the right scales exponentially with $n$, that is that 
for any $e \in [0, \alpha]$, one has a well-defined $F(e) \coloneqq \lim_{n \to \infty} (1/n) \log \int \mu_n(\rd \bx) \, \delta(E_\bW(\bx)/n - e) \in [-\infty, 0]$.
In mathematical terms, for $\bx \sim \mu_n$, $E_\bW(\bx)/n$ satisfies a large deviation principle in the scale $n$, with rate function $-F(e)$.
Therefore, by eq.~\eqref{eq:Pe}, $P_\beta(e)$ has large deviations in the scale $n$ around a value $e^\star(\beta)$, i.e.\ we have the following behavior: 
\begin{align}\label{eq:def_estar_ldp}
    \begin{dcases}
        e^\star(\beta) &\coloneqq \argmax_e [-\beta e + F(e)], \\
        P_\beta(e) &\simeq \exp\Big\{n \Big(- \beta e + F(e) - [-\beta e^\star(\beta) + F(e^\star(\beta))]\Big)\Big\}.
    \end{dcases}
\end{align}
In particular, the probability (under the Gibbs measure) of having a configuration with energy $e$ such that $|e - e^\star(\beta)| > \varepsilon$ is exponentially small in $n$ for any $\varepsilon > 0$.
Therefore, we expect that at any $\beta \geq 0$, all the mass of the Gibbs measure concentrates (as $n \to \infty$) around the level set with 
intensive energy $e^\star(\beta)$, given thus also by the mean energy under the Gibbs measure \cite{ellis2006entropy}.
Note that we discarded the dependency of $e^\star$ on $\bW$:
the concentration with respect to $\bW$ can be justified (but not proven!) using the concentration of $\Phi_n(\bW,\beta)$ 
in Theorem~\ref{thm:free_entropy_concentration}.
Indeed, note that the average energy under the Gibbs measure is precisely given by a derivative of the free energy: 
\begin{align*}
    \int \rd \bbP_{\beta,\bW}(\bx) \, E_\bW(\bx) &= -\frac{\partial}{\partial \beta} \Big[\log \mcZ_n(\bW, \beta)\Big].
\end{align*}
Therefore, one expects that the concentration of the free energy 
transfers to the derivatives, and thus that 
the energy level also concentrates as a function of $\bW$.
Summing up, the function $e^\star(\alpha,\beta)$ (we explicit its dependency on $\alpha$ and $\beta$) is -- conjecturally -- equal to the following limit: 
\begin{align*}
    e^\star(\alpha, \beta) &= \plim_{n \to \infty} \frac{1}{n} \int \rd \bbP_{\beta,\bW}(\bx) \, E_\bW(\bx) = -\partial_\beta\Phi(\alpha,\beta).
\end{align*}
in which the limit is again in probability over the randomness induced by $\bW$.
Furthermore, in the limit $\beta \to \infty$ -- as we argued in the main text -- we expect the Gibbs measure to concentrate its mass around the global minima of $E_\bW$, 
and therefore that
\begin{align}\label{eq:estar_GS}
    \lim_{\beta \to \infty} e^\star(\alpha, \beta) &= \plim_{n \to \infty} \Big\{\frac{1}{n}  \min_{\bx \in \mcS^{n-1}} E_\bW(\bx) \Big\}.
\end{align}
\textbf{The lower bound --}
By eq.~\eqref{eq:pmn_minimum} and eq.~\eqref{eq:estar_GS}, $e^\star(\alpha, \beta = \infty) \geq 1 \Leftrightarrow \alpha \geq \alpha_\inj$\footnote{Assuming that $\alpha \mapsto e^\star(\alpha, \beta = \infty)$ is continuous and strictly increasing, which we always observe, see
Fig.~\ref{fig:chi_estar_T0}.}.
Since $e^\star(\alpha,\beta)$ is a decreasing function of $\beta$ and $e^\star(\alpha,\beta = 0) = \alpha/2$, this implies that for any $\alpha \in (2,\alpha_\inj]$
there exists $\beta^\star(\alpha) \in [0,\infty]$ such that 
\begin{align*}
    \beta^\star(\alpha) &\coloneqq \sup\{\beta \geq 0 \, : \, e^\star(\alpha,\beta) \geq 1\}.
\end{align*}
Moreover, it is easy to see that $\beta^\star(\alpha)$ is a non-decreasing function of $\alpha$.
In particular, if for all $\beta < \beta^\star(\alpha)$ the RS solution is stable (in the sense of the dAT condition described above, and derived in Appendix~\ref{subsec_app:stability_rs}), 
then the replica-symmetric ansatz will yield the exact solution for all $\beta \in [0,\beta^\star(\alpha))$.
We denote $\alpha_\AT$ the largest such $\alpha$:
\begin{align*}
    \alpha_\AT &\coloneqq \sup \{\alpha \, | \, \alpha > 2 \textrm{ and the RS solution is stable for all } \beta < \beta^\star(\alpha) \}.
\end{align*}
Recall that $\alpha_\inj$ is, according to our criterion, equal to:
\begin{align*}
   \alpha_\inj &= \inf \{\alpha \, : \, \lim_{\beta\to \infty} e^\star(\alpha,\beta) \geq 1\} = \inf\{\alpha \, : \, \beta^\star(\alpha) = \infty\}.
\end{align*}
However, we know that for $\beta \to \infty$ the RS solution is never stable for $\alpha > 2$ (see Appendix~\ref{subsec_app:stability_rs}), 
so for all $\alpha > 2$, if $\beta^\star(\alpha) = \infty$ then $\alpha > \alpha_\AT$.
In particular, we get the lower bound $\alpha_\AT \leq \alpha_\inj$.
We stress that an important property of $\alpha_\AT$ is that it can be computed solely from the RS solution (and its stability analysis),
and might thus be more amenable to a rigorous analysis than the replica symmetry breaking upper bounds.
Recall that the stability condition for the RS solution is given by eq.~\eqref{eq:AT_explicit}, in which $q$ is the overlap given by the RS calculation, i.e.\ by eq.~\eqref{eq:q_RS_eq_new}.
A numerical evaluation of eq.~\eqref{eq:AT_explicit}, available in the attached code \cite{github_repo}, yields $\alpha_\AT \simeq 5.3238$, 
which implies the lower bound presented in eq.~\eqref{eq:lower_bound_RS_stability}.

\section{One-step replica symmetry breaking}\label{sec_app:1rsb}
\subsection{Derivation of the 1-RSB free entropy}\label{subsec_app:1rsb_derivation}

We perform here, for completeness of our presentation, the textbook calculation of the spherical perceptron free entropy at the one-RSB level.
We start again from eq.~\eqref{eq:Phir_final}, which we rewrite using a Gaussian transformation as:
\begin{align}\label{eq:Phir_1rsb_1}
   \Phi(\alpha,\beta,r) &= \sup_{\bQ}\Big[\frac{1}{2} \log \det \bQ + \alpha \log \int_{\bbR^r} \frac{\rd \bu \rd \bv}{(2\pi)^{r}} e^{-\frac{1}{2} \sum_{a,b} Q^{ab} v^a v^b -\beta \sum_{a=1}^r \theta(u^a) + i \sum_{a=1}^r u^a v^a}\Big].
\end{align}
We assume a 1RSB ansatz given in eq.~\eqref{eq:rhoq_Q_1RSB}, with $q_1 > q_0$, and $m \in \{1,\cdots,r\}$ with $m \, | \, r$ the Parisi parameter (i.e.\ the size of the diagonal blocks
in the ultrametric $\bQ$).
More precisely we have, with $k \coloneqq r / m$:
\begin{align}\label{eq:1RSB_explicit}
    \begin{dcases}
        Q_{aa} &= 1, \\
        Q_{ab} &= q_1 \textrm{ if } \Big\lfloor \frac{a}{k} \Big\rfloor = \Big\lfloor \frac{b}{k} \Big\rfloor, \\
        Q_{ab} &= q_0 \textrm{ otherwise.}
    \end{dcases}
\end{align}
\textbf{The entropic contribution --} We focus on the first term of eq.~\eqref{eq:Phir_1rsb_1}.
It is elementary algebra to check that under the ansatz of eq.~\eqref{eq:1RSB_explicit}, the spectrum of $\bQ$ is: 
\begin{align*}
    \mathrm{Sp}(\bQ) &= \{1-q_1\}^{r-k} \cup \{1-mq_0 + (m-1)q_1\}^{k-1} \cup \{1+(r-m)q_0 + (m-1)q_1\}.
\end{align*}
In particular, this yields:
\begin{align*}
    \log \det \bQ &= r \frac{m - 1}{m} \log (1-q_1) + \frac{r-m}{m} \log [1-mq_0 + (m-1)q_1] + \log[1+(r-m)q_0 + (m-1)q_1].
\end{align*}
And thus: 
\begin{align}\label{eq:entropic_1rsb}
    \partial_r[\log \det \bQ]_{r = 0 } &= \frac{m - 1}{m} \log (1-q_1) + \frac{1}{m} \log [1-mq_0 + (m-1)q_1] + \frac{q_0}{[1-mq_0 + (m-1)q_1]}.
\end{align}
\textbf{The interaction contribution --}
We focus now on the second term $\alpha G_{2,r}(\bQ)$ in eq.~\eqref{eq:Phir_1rsb_1}, with:
\begin{align*}
    G_{2,r}(\bQ) &\coloneqq \log \int_{\bbR^r} \frac{\rd \bu \rd \bv}{(2\pi)^{r}} e^{-\frac{1}{2} \sum_{a,b} Q^{ab} v^a v^b -\beta \sum_{a=1}^r \theta(u^a) + i \sum_{a=1}^r u^a v^a}.
\end{align*}
Under the 1-RSB ansatz, it becomes:
\begin{align*}
    G_{2,r}(\bQ) &= \log \int_{\bbR^r} \frac{\rd \bu \rd \bv}{(2\pi)^{r}} e^{-\frac{1-q_1}{2} \sum_{a} (v^a)^2 - \frac{q_0}{2} \big(\sum_a v^a\big)^2 
    - \frac{q_1 - q_0}{2} \sum_{x=0}^{k-1} \big(\sum_{l=1}^m v^{mx + l}\big)^2 - \sum_{a} [\beta\theta(u^a) - i u^a v^a]}.
\end{align*}
Introducing Gaussian transformations based on the formula $e^{-x^2/2} = \int \mcD z \, e^{-i z x}$ to decouple the replicas, we obtain:
\begin{align*}
    &G_{2,r}(\bQ) = \log \int \mcD \xi \prod_{x=0}^{k-1} \int \mcD z_x \int_{\bbR^r} \frac{\rd \bu \rd \bv}{(2\pi)^{r}} \exp\Big\{-\frac{1-q_1}{2} \sum_{a} (v^a)^2 - i \sqrt{q_0} \xi \sum_a v^a \\
    &- i\sqrt{q_1 - q_0} \sum_{x=0}^{k-1} z_x \sum_{l=1}^m v^{mx + l} -\beta \sum_{a} \theta(u^a) + i \sum_{a} u^a v^a\Big\}, \\ 
    &= \log \int \mcD \xi \Bigg\{ \int \mcD z \Bigg[\int \frac{\rd u \rd v}{2\pi} \exp\Big\{-\frac{1-q_1}{2} v^2 - i v [\sqrt{q_0} \xi + \sqrt{q_1 - q_0} z] -\beta \theta(u) + i u v\Big\} \Bigg]^m \Bigg\}^{\frac{r}{m}}.
\end{align*}
Using this Gaussian transformation trick, we were able to decouple replicas and therefore obtain an expression that is analytic in $r$.
This allows to take the $r \downarrow 0$ limit (keeping $m$ fixed), and to reach:
\begin{align*}
    &\partial_r \big[G_{2,r}(\bQ)\big]_{r=0} \\
    &= \frac{1}{m} \int \mcD \xi \log \Bigg\{ \int \mcD z \Bigg[\int \frac{\rd u \rd v}{2\pi} \exp\Big\{-\frac{1-q_1}{2} v^2 + i v [u - \sqrt{q_0} \xi - \sqrt{q_1 - q_0} z] -\beta \theta(u)\Big\} \Bigg]^m \Bigg\}.
\end{align*}
Performing the Gaussian integrals, and recall the definition of $H(x) \coloneqq \int_x^\infty \mcD u$, we reach:
\begin{align}\label{eq:energetic_1rsb}
    \partial_r \big[G_{2,r}(\bQ)\big]_{r=0} &= \frac{1}{m}  \int \mcD \xi \log \Bigg\{ \int \mcD z \Bigg[1- (1-e^{-\beta}) H\Big(- \frac{\sqrt{q_0} \xi + \sqrt{q_1 - q_0} z}{\sqrt{1-q_1}}\Big) \Bigg]^m \Bigg\}.
\end{align}
Combining eq.~\eqref{eq:entropic_1rsb} and eq.~\eqref{eq:energetic_1rsb}, we reach 
eq.~\eqref{eq:phi_1rsb}.

\subsection{Zero-temperature limit}\label{subsec_app:zerotemp_1rsb}

Recall that in the $\beta \to \infty$ limit we have the scaling (see e.g.\ \cite{franz2017universality}) 
\begin{align}\label{eq:scaling_1rsb_zerotemp}
        m \sim \frac{c_m}{\beta} \hspace{1cm} \textrm{and} \hspace{1cm}
        1 - q_1 \sim \frac{\chi_\ORSB}{\beta},
\end{align}
while $q_0$ has a limit in $(0,1)$ as $\beta \to \infty$.
In the following of this section, we write $\chi_\ORSB = \chi$ to lighten the notations.
The asymptotics of the determinant term in eq.~\eqref{eq:phi_1rsb} can be worked out:
\begin{align*}
   &\frac{m - 1}{2m} \log (1-q_1) + \frac{1}{2m} \log [1-mq_0 + (m-1)q_1] + \frac{q_0}{2[1-mq_0 + (m-1)q_1]} \nonumber \\
   &\simeq \frac{\beta}{2} \Big[\frac{q_0}{\chi+c_m(1-q_0)} + \frac{1}{c_m} \log \Big(\frac{\chi + c_m(1-q_0)}{\chi}\Big)\Big].
\end{align*}
The limit of the other term in eq.~\eqref{eq:phi_1rsb} can also be computed, using that:
\begin{align*}
    \Bigg[1- (1-e^{-\beta}) H\Big(- \frac{\sqrt{q_0} \xi + \sqrt{q_1 - q_0} \xi_1}{\sqrt{1-q_1}}\Big) \Bigg]^m &= \exp\{m f_\beta(u)\},
\end{align*}
with $u \coloneqq \sqrt{q_0} \xi_0 + \sqrt{1-q_0} \xi_1$, and $f_\beta(h) \coloneqq \log (1 - (1-e^{-\beta}) H[-h/\sqrt{1-q_1}])$. 
We described the expansion of $f_\beta(h)/\beta$ for large $\beta$ in eq.~\eqref{eq:expansion_fbeta} (simply replacing $\chi_\RS$ by $\chi_\ORSB$).
We reach then:
\begin{align}\label{eq:expansion_1rsb_basis_term}
    \Bigg[1- (1-e^{-\beta}) H\Big(- \frac{\sqrt{q_0} \xi + \sqrt{q_1 - q_0} \xi_1}{\sqrt{1-q_1}}\Big) \Bigg]^m &\simeq 
    \begin{dcases}
        1 &\textrm{ if } u < 0, \\
        e^{-c_m} &\textrm{ if } u > \sqrt{2\chi}, \\
        e^{-\frac{c_m u^2}{2 \chi}} &\textrm{ if } u \in (0,\sqrt{2\chi}).
    \end{dcases}
\end{align}
Anticipating on what follows, we introduce the auxiliary functions (in which $u = u(\xi_0,\xi_1) \coloneqq \sqrt{q_0}\xi_0 + \sqrt{1-q_0} \xi_1$): 
\begin{align}\label{eq:def_nabc}
    \begin{dcases}
        n(\xi_0) &\coloneqq \int_{u \leq 0} \mcD \xi_1 + e^{-c_m} \int_{u > \sqrt{2 \chi}} \mcD \xi_1 +  \int_{0<u< \sqrt{2 \chi}} \mcD \xi_1 e^{- \frac{c_m}{2\chi} u^2}, \\
        a(\xi_0) &\coloneqq \int_{0< u < \sqrt{2\chi}} \mcD \xi_1 \, e^{-\frac{c_m}{2\chi}u^2} \, u, \\
        b(\xi_0) &\coloneqq \int_{0< u < \sqrt{2\chi}} \mcD \xi_1 \, e^{-\frac{c_m}{2\chi}u^2} \, u^2, \\
        c(\xi_0) &\coloneqq e^{-c_m} \int_{u > \sqrt{2\chi}} \mcD \xi_1 + \frac{1}{2 \chi} \int_{0 < u < \sqrt{2\chi}} \mcD \xi_1 \, u^2 \, e^{-\frac{c_m}{2 \chi} u^2}.
    \end{dcases}
\end{align}
By integration by parts, all these functions can be expressed in terms of elementary functions and $H(x) \coloneqq \int_x^\infty \mcD \xi$.
Moreover, note that we have the identities: 
\begin{subnumcases}{\label{eq:dn}}
    \label{eq:dn_dq}
    \partial_{q_0}n(\xi_0) =  -\frac{c_m}{2\chi(1-q_0)} \Big[\frac{\xi_0}{\sqrt{q_0}} a(\xi_0) - b(\xi_0)\Big],& \\ 
    \label{eq:dn_dchi}
    \partial_{\chi}n(\xi_0) = \frac{c_m}{2\chi^2} b(\xi_0), & \\ 
    \label{eq:dn_dcm}
    \partial_{c_m}n(\xi_0) = -c(\xi_0). &
\end{subnumcases}
Eqs.~\eqref{eq:dn_dchi} and \eqref{eq:dn_dcm} can be obtained directly from the definition of eq.~\eqref{eq:def_nabc}. 
For eq.~\eqref{eq:dn_dq}, we found more convenient to differentiate the finite-$\beta$ integral one can write for $n(\xi_0)$ using eq.~\eqref{eq:expansion_1rsb_basis_term}, and then take its large $\beta$ limit. 
We leave the derivation of these equations to the reader.
Using the expansion of eq.~\eqref{eq:expansion_1rsb_basis_term}, we obtain the limit of the second term of eq.~\eqref{eq:phi_1rsb}:
\begin{align*}
    &\int \mcD \xi_0 \log \Bigg\{ \int \mcD \xi_1 \Bigg[1- (1-e^{-\beta}) H\Big(- \frac{\sqrt{q_0} \xi_0 + \sqrt{q_1 - q_0} \xi_1}{\sqrt{1-q_1}}\Big) \Bigg]^m \Bigg\} \Bigg] \simeq \int \mcD \xi_0 \log  n(\xi_0).
\end{align*}
In the end, we have computed the limit of the free energy at the 1-RSB level, i.e.\ 
$f^\star_\ORSB(\alpha) \coloneqq - \lim_{\beta \to \infty} \Phi_\ORSB(\alpha,\beta)/\beta$:
\begin{align}\label{eq:phi_1rsb_zero_temp}
   f^\star_\ORSB(\alpha) &=
    -\frac{q_0}{2[\chi+c_m(1-q_0)]} - \frac{1}{2c_m} \log \Big(\frac{\chi + c_m(1-q_0)}{\chi}\Big) - \frac{\alpha}{c_m} \int \mcD \xi_0 \log n(\xi_0),
\end{align}
in which one must implicitly maximize over $(c_m, q_0, \chi)$.
Note that an equivalent expression can be obtained using the limit of the average energy, 
since $e^\star_\ORSB(\alpha,\beta = \infty) = \lim_{\beta \to \infty} [-\partial_\beta \Phi_\ORSB(\alpha, \beta)] = f^\star_\ORSB(\alpha)$. 
Performing expansions in a similar way to the RS computations described in Appendix~\ref{subsec_app:zerotemp_RS},
we reach:
\begin{align}\label{eq:fstar_1rsb_app}
    f^\star_\ORSB(\alpha) &= e^\star_\ORSB(\alpha,\beta = \infty) = \alpha \int \mcD \xi_0 \frac{1}{n(\xi_0)} e^{-c_m} \int_{u > \sqrt{2\chi}} \mcD \xi_1.
\end{align}
Let us emphasize that eq.~\eqref{eq:fstar_1rsb_app} is an identity involving the parameters $(c_m, q_0, \chi)$, which have to be found by maximizing eq.~\eqref{eq:phi_1rsb_zero_temp}.

\subsection{Numerical procedure}\label{subsec_app:1rsb_numerical}

Let us summarize here the equations that allow to find the $1$-RSB prediction for the injectivity threshold, using the set of auxiliary functions of eq.~\eqref{eq:def_nabc}.
One simply proceeds by derivation of the limit of the free energy functional given in eq.~\eqref{eq:phi_1rsb_zero_temp} with respect to $(q_0, \chi,c_m)$, 
using eq.~\eqref{eq:dn}.
More precisely, at a given value of $\alpha > 2$, one must find $q_0 \in (0,1)$ and $\chi,c_m > 0$ satisfying the following set of three equations:
\begin{align}\label{eq:final_eq_1rsb_zerotemp}
    \begin{dcases}
        &\frac{q_0}{[\chi + (1-q_0)c_m]^2} = \frac{\alpha}{c_m \chi (1-q_0)} 
        \int \mcD \xi_0 \frac{1}{n(\xi_0)} \Big[\frac{\xi_0}{\sqrt{q_0}} a(\xi_0) - b(\xi_0) \Big], \\
        &\frac{\chi + (1-q_0)^2 c_m}{\chi[(1-q_0) c_m + \chi]^2} 
        = \frac{\alpha}{\chi^2}\int \mcD \xi_0 \frac{b(\xi_0)}{n(\xi_0)}, \\
        &\frac{c_m (1-q_0) [\chi + c_m (1-2q_0)]}{2[\chi + c_m(1-q_0)]^2} -\frac{1}{2} \log\frac{\chi+c_m(1-q_0)}{\chi}
        = \alpha \int \mcD \xi_0 \, \Big\{ \log n(\xi_0) + c_m \frac{c(\xi_0)}{n(\xi_0)}\Big\}.
    \end{dcases}
\end{align}
Once one has found the solution to eq.~\eqref{eq:final_eq_1rsb_zerotemp}, we can obtain the large-$\beta$ limit of the energy either from eq.~\eqref{eq:phi_1rsb_zero_temp} or eq.~\eqref{eq:fstar_1rsb_app}. 

\myskip
Following the statistical physics folklore, in order to implement an iterative scheme to solve eq.~\eqref{eq:final_eq_1rsb_zerotemp}, 
we use auxiliary variables.
Namely, we iterate the first two equations of eq.~\eqref{eq:final_eq_1rsb_zerotemp} as:
\begin{align}\label{eq:final_eq_1rsb_zerotemp_what}
    \begin{dcases}
        A_0^t &= \frac{\alpha}{c_m \chi^t (1-q_0^t)}
        \int \mcD \xi_0 \frac{1}{n_t(\xi_0)} \Big[\frac{\xi_0}{\sqrt{q_0^t}} a_t(\xi_0) - b_t(\xi_0) \Big], \\
        A_1^t &= \frac{\alpha}{(\chi^t)^2}\int \mcD \xi_0 \frac{b_t(\xi_0)}{n_t(\xi_0)}, \\
        q_0^{t+1} &= F_1(A_0^t, A_1^t, c_m^t), \\ 
        \chi^{t+1} &= F_2(A_0^t,A_1^t, c_m^t),
    \end{dcases}
\end{align}
in which we added a time subscript for the auxiliary functions to highlight their dependency on $q_0^t,\chi^t,c_m^t$. 
Moreover, the functions $F_1,F_2$ are defined as the only roots (in $q_0,\chi$) of the equations 
\begin{align*}
      A_0 = \frac{q_0}{[\chi + (1-q_0)c_m]^2} \hspace{1cm} \textrm{and} \hspace{1cm}
      A_1 = \frac{\chi + (1-q_0)^2 c_m}{\chi[(1-q_0) c_m + \chi]^2},
\end{align*}
such that $q_0 \in (0,1)$ and $\chi \geq 0$.
Note that this implies that 
\begin{align}\label{eq:chi_from_q0_1rbs}
    \chi = \frac{A_0 c_m (1-q_0)^2}{q_0 A_1 - A_0}.
\end{align}
Therefore, in order for the solution to exist we must have $A_0 < A_1$, and then the solution satisfies $q_0 > A_0 / A_1$.
The remaining equation on $q_0$ can be written as:
\begin{align}\label{eq:1rsb_zerotemp_remaining_q0}
    A_0 &= \frac{(A_1 q_0 - A_0)^2}{q_0 c_m^2 (1-q_0)^2 (A_1 - A_0)^2}.
\end{align}
We solve eq.~\eqref{eq:1rsb_zerotemp_remaining_q0} on $q_0$ with a polynomial equations solver, and consider the unique solution in $(0,1)$ such that the 
corresponding $\chi$ in eq.~\eqref{eq:chi_from_q0_1rbs} satisfies $\chi \geq 0$, i.e.\ such that $q_0 > A_0 / A_1$.

\myskip
At a given iteration $t$, we iterate eq.~\eqref{eq:final_eq_1rsb_zerotemp_what} for the value $c_m = c_m^t$. 
We then do a binary search to solve the last equation of eq.~\eqref{eq:final_eq_1rsb_zerotemp} and find $c_m^{t+1}$.
We found this procedure to converge very quickly (see the attached code \cite{github_repo}), and it yields the 1RSB curves in Fig.~\ref{fig:chi_estar_T0}
and the prediction of eq.~\eqref{eq:alphainj_1RSB}.

\section{Details of the FRSB computation}\label{sec_app:frsb}
In this section we derive the full RSB conjecture for the free entropy $\Phi(\alpha,\beta)$. 
Our computation is extremely close to the one of \cite{franz2017universality}, and we refer to this work (and the lecture notes \cite{urbani2018statistical}) for more details on 
the technicalities of the derivation.

\myskip
Recall the form of the $r$-th moment of the partition function, written as a function of the overlap matrix (without any assumption on the form of the saddle point), 
that is eq.~\eqref{eq:Phir_final}. 
Note that by using the Gaussian integration formula, we can rewrite $I_\beta(\bQ)$ so as to obtain:
\begin{align}\label{eq:Phir_frsb_1}
   \Phi(\alpha,\beta;r) &= \sup_{\bQ} \Big[\frac{1}{2} \log \det \bQ + \alpha \log \int_{\bbR^r} \frac{\rd \bu \rd \bv}{(2\pi)^{r}} e^{-\frac{1}{2} \sum_{a,b} Q^{ab} v^a v^b -\beta \sum_{a=1}^r \theta(u^a) + i \sum_{a=1}^r u^a v^a} \Big].
\end{align}
Let us now perform the replica method under the full-RSB ansatz described in Fig.~\ref{fig:q_frsb}.

\subsection{Entropic contribution}\label{subsec_app:entropic}

We start with the first ``entropic'' term in eq.~\eqref{eq:Phir_frsb_1}.
Its expression under a full-RSB ansatz is given in eq.~(23) of \cite{franz2017universality}, 
itself taken from Appendix~A.II of \cite{mezard1991replica}. 
However
the derivation is itself very interesting and will be useful for the other term in eq.~\eqref{eq:Phir_frsb_1}, so we first detail it here.

\myskip\textbf{Derivation for $r > 0$ --}
We focus on the entropic term, which we may write as: 
\begin{align}\label{eq:entropic_term_1}
    \frac{1}{2} \log \det \bQ &= -\log \int_{\bbR^r} \frac{\rd \bu}{(2\pi)^{r/2}} \exp\Big\{-\frac{1}{2} \bu^\intercal \bQ \bu\Big\}.
\end{align}
We fix a $k$-RSB ansatz, cf.\ Fig.~\ref{fig:q_frsb}, and we will take in the end the limit $k \to \infty$. 
We have in the $r \downarrow 0$ limit $m_{-1} \coloneqq r \leq m_0 \leq m_1 \leq \cdots \leq m_{k-1} \leq m_k = 1$, 
and the parameters $q_0 \leq q_1 \leq q_k \leq q_{k+1} = 1$.
Recall that in this ansatz, the hierarchical overlap matrix $\{Q_{ab}\}$ can be written as:
\begin{align*}
    \bQ &= \sum_{i=0}^{k+1} (q_i - q_{i-1}) \bJ^{(r)}_{m_{i-1}}, 
\end{align*}
with $J_m^{(r)}$ the block-diagonal matrix with $r/m$ blocks of size $m$, each diagonal block being the all-ones matrix.
In order to compute the integral of eq.~\eqref{eq:entropic_term_1}, we use a simple yet very powerful identity introduced in \cite{duplantier1981comment}, 
and valid for any matrix (not necessarily a hierarchical RSB matrix) $\{Q_{ab}\}$:
\begin{align}\label{eq:duplantier}
    \exp\Big\{-\frac{1}{2} \sum_{a,b=1}^r Q_{ab} u_a u_b\Big\} 
    &= \exp\Bigg(\frac{1}{2} \sum_{a,b=1}^r Q_{ab} \frac{\partial^2}{\partial h_a \partial h_b}\Bigg)\Bigg[\prod_{c=1}^r \exp(-i u_c h_c)\Bigg]_{\bh = 0}.
\end{align}
This identity can be shown by Taylor-expanding the exponential involving the differential operator.
Using it in eq.~\eqref{eq:entropic_term_1} we get:
\begin{align}\label{eq:entropic_term_2}
    - \frac{1}{2} \log \det \bQ &= \log \int_{\bbR^r} \frac{\rd \bu}{(2\pi)^{r/2}} \exp\Bigg(\frac{1}{2} \sum_{a,b=1}^r Q_{ab} \frac{\partial^2}{\partial h_a \partial h_b}\Bigg)\Bigg[\prod_{c=1}^r \exp(-i u_c h_c)\Bigg]_{\bh = 0}.
\end{align}
Note that $\bu$ does not appear in the differential operator, so that one can exchange the differential operator and the integral over $\bu$.
Integrating with respect to $\bu$ yields then, using the Fourier representation of the delta distribution (we denote $\partial_a = \partial/\partial_{h_a}$):
\begin{align}\label{eq:entropic_term_3}
    &\hspace{-0.3cm}- \frac{1}{2} \log \det \bQ = \log  \Bigg\{ (2\pi)^{r/2} \exp\Bigg(\frac{1}{2} \sum_{a,b=1}^r Q_{ab} \partial_a \partial_b\Bigg)\Bigg[\prod_{c=1}^r \delta(h_c)\Bigg]_{\bh = 0}\Bigg\}, \nonumber \\
    &\hspace{-0.3cm}= \log  \Bigg\{ (2\pi)^{r/2} \exp\Bigg(\frac{1}{2}\sum_{i=0}^{k+1} (q_i - q_{i-1}) \sum_{a,b=1}^r (\bJ^{(r)}_{m_{i-1}})_{ab} \partial_a \partial_b \Bigg)\Bigg[\prod_{c=1}^r \delta(h_c)\Bigg]_{\bh = 0}\Bigg\}, \nonumber \\
    &\hspace{-0.3cm}= \log  \Bigg\{ (2\pi)^{r/2} \exp\Bigg(\frac{1}{2}\sum_{i=0}^{k} (q_i - q_{i-1}) \sum_{a,b=1}^r (\bJ^{(r)}_{m_{i-1}})_{ab} \partial_a \partial_b\Bigg) \exp\Bigg(\frac{1-q_k}{2} \sum_{a=1}^r \partial^2_a\Bigg)\Bigg[\prod_{c=1}^r \delta(h_c)\Bigg]_{\bh = 0}\Bigg\}.
\end{align}
We use now the crucial identity, for any $\omega \geq 0$ and smooth function $f$, and which can be shown by Taylor-expanding $f$ around $h$ inside the integral on the right hand side:
\begin{align}\label{eq:identity_convolution}
    \exp \Big(\frac{\omega}{2} \partial_h^2\Big) f[h] &= [\gamma_\omega \star f](h) = \int \frac{\rd z}{\sqrt{2 \pi \omega}} e^{-\frac{z^2}{2\omega}} f(h-z).
\end{align}
Here we denoted $\gamma_\omega(x) = e^{-x^2/(2\omega)} / \sqrt{2\pi \omega}$, and $\gamma_0(x) = \delta(x)$.
Using eq.~\eqref{eq:identity_convolution} inside eq.~\eqref{eq:entropic_term_3} we reach:
\begin{align}\label{eq:entropic_term_4}
    &- \frac{1}{2} \log \det \bQ = \log  \Bigg\{ (2\pi)^{r/2} \exp\Bigg(\frac{1}{2}\sum_{i=0}^{k} (q_i - q_{i-1}) \sum_{a,b=1}^r (\bJ^{(r)}_{m_{i-1}})_{ab} \partial_a \partial_b\Bigg) \Bigg[\prod_{c=1}^r \gamma_{1-q_k}(h_c)\Bigg]_{\bh = 0}\Bigg\}.
\end{align}
We will iteratively apply the differential operator in the exponential, starting from $i = 0$ up to $i = k$.
We will make use of another important identity, which is just a consequence of simple differential calculus combined with eq.~\eqref{eq:identity_convolution}, 
and valid for any $p,n \in \bbN$ and smooth $R(h_1, \cdots, h_n)$:
\begin{align}\label{eq:identity_diff_calculus}
    \begin{dcases}
   \Bigg[\Big(\sum_{a=1}^n \frac{\partial}{\partial h_a}\Big)^p R(h_1, \cdots, h_n)\Bigg]_{h_a = h} &= \frac{\partial^p}{\partial h^p} [h \mapsto R(h,h,\cdots,h)], \\
    \exp\Bigg(\frac{\omega}{2} \Big(\sum_{a=1}^n \partial_a\Big)^2 \Bigg)[R(h_1, \cdots, h_n)]_{h_a = h} \!\!\!\! &=e^{\frac{\omega}{2} \frac{\partial^2}{\partial h^2}} R(h,\cdots,h) = \gamma_\omega \star [h \mapsto R(h,\cdots,h)].
    \end{dcases}
\end{align}
Let us now come back to eq.~\eqref{eq:entropic_term_4}. We separate the term $i = 0$, and we have, 
with eq.~\eqref{eq:identity_convolution}:
\begin{align}\label{eq:entropic_term_5}
     \nonumber
    &\exp\Bigg(\frac{1}{2}\sum_{i=0}^{k} (q_i - q_{i-1}) \sum_{a,b=1}^r (\bJ^{(r)}_{m_{i-1}})_{ab} \partial_a \partial_b\Bigg) \Bigg[\prod_{c=1}^r \gamma_{1-q_k}(h_c)\Bigg]_{\bh = 0} \\
    &= \exp\Bigg(\frac{q_0}{2} \Big(\sum_a \partial_a\Big)^2\Bigg) [\Xi(\bh)]_{\bh = 0}
    = \gamma_{q_0} \star [h \mapsto \Xi(h, \cdots, h)]_{h = 0},
\end{align}
with $\Xi(\bh)$ defined as: 
\begin{align*}
   \Xi(\bh) &\coloneqq \exp\Bigg(\frac{1}{2}\sum_{i=1}^{k} (q_i - q_{i-1}) \sum_{a,b=1}^r (\bJ^{(r)}_{m_{i-1}})_{ab} \partial_a \partial_b\Bigg)\Bigg[\prod_{c=1}^r \gamma_{1-q_k}(h_c)\Bigg]. 
\end{align*}
Note that $\Xi(\bh)$ factorizes over the inner diagonal blocks of size $m_0$, and we have
$\Xi(h,\cdots,h) = \zeta(h)^{r/{m_0}}$, with
\begin{align}
   \label{eq:def_f_entropic_contribution}
   \zeta(h)&\coloneqq  \exp\Bigg(\frac{1}{2}\sum_{i=1}^{k} (q_i - q_{i-1}) \sum_{a,b=1}^{r/m_0} (\bJ^{(r/m_0)}_{m_{i-1}/m_0})_{ab} \partial_a \partial_b\Bigg)\Bigg[\prod_{c=1}^{m_0} \gamma_{1-q_k}(h_c)\Bigg]_{h_c = h}.
\end{align}
Therefore, putting it back into eq.~\eqref{eq:entropic_term_5} and using eq.~\eqref{eq:identity_diff_calculus}, we have:
\begin{align*}
    &\exp\Bigg(\frac{1}{2}\sum_{i=0}^{k} (q_i - q_{i-1}) \sum_{a,b=1}^r (\bJ^{(r)}_{m_{i-1}})_{ab} \partial_a \partial_b\Bigg) \Bigg[\prod_{c=1}^r \gamma_{1-q_k}(h_c)\Bigg]_{\bh = 0}
    \! \! \! = [\gamma_{q_0} \star \zeta^{r/m_0}](h = 0),
\end{align*}
with $\zeta(h)$ defined in eq.~\eqref{eq:def_f_entropic_contribution}.
This procedure can then be repeated iteratively on the diagonal blocks, all the way to the innermost ones. 
Eq.~\eqref{eq:entropic_term_4} then becomes:
\begin{align}\label{eq:entropic_term_7}
    &- \frac{1}{2} \log \det \bQ = \log\Big[ (2\pi)^{r/2} \, \gamma_{q_0} \star g^{r/m_0} (m_0, h=0) \Big],
\end{align}
with the functions $g(m_i,h)$ iteratively defined as:
\begin{align}\label{eq:iterative_construction_g}
    \begin{dcases}
        g(m_k = 1, h) &= \gamma_{1-q_k}(h), \\ 
        g(m_{i-1}, h) &= \gamma_{q_i - q_{i-1}} \star g^{m_{i-1}/m_i}(m_i, h).
    \end{dcases}
\end{align}
We now take the $k \to \infty$ (Full RSB) limit in eq.~\eqref{eq:iterative_construction_g}.
In this limit we can approximate any function $q(x)$ (see Fig.~\ref{fig:q_frsb}),
and taking for $(m_i)_{i=0}^k$ a regular grid on $x \in [0, 1]$,
we have $m_0 \to 0$, $m_{k-1} \to 1$, and for all $i = 0, \cdots, k-1$, we have $m_i \to x$ and $m_{i} - m_{i-1} = \rd x$. 
Moreover $q_k \to q(1)$, $q_0 \to q(0)$, and $q_{i+1} - q_{i} = \dot{q}(x) \rd x$.
We sometimes also use the notation $q_m = q(0)$, $q_M = q(1)$.
To make things clearer, we will denote derivatives w.r.t.\ $x$ with dots, and the ones w.r.t.\ $h$ with the usual prime.
The second line of eq.~\eqref{eq:iterative_construction_g} becomes, at first order in $\rd x$ (recall the crucial eq.~\eqref{eq:identity_convolution}), 
for $x \in (0,1)$:
\begin{align*}
    g(x, h) - \rd x \, \dot{g}(x,h) &= e^{\frac{\dot{q}(x)}{2} \rd x \, \partial_h^2} \Big[g - \frac{\rd x}{x} g \log g\Big](x,h) = \Big(1 + \frac{\dot{q}(x)}{2} \rd x \, \partial_h^2\Big) \Big[g - \frac{\rd x}{x} g \log g\Big](x,h).
\end{align*}
Comparing the terms at first order in $\rd x$, we reach the PDE: 
\begin{align}\label{eq:Parisi_PDE_g}
    \dot{g}(x,h) &= -\frac{\dot{q}(x)}{2} g''(x,h) + \frac{1}{x} g \log g (x,h), \hspace{0.5cm} x \in (0, 1).
\end{align}
It is convenient to rewrite eq.~\eqref{eq:Parisi_PDE_g} in terms of $f(x,h) \coloneqq (1/x) \log g(x,h)$, which yields the \emph{Parisi PDE}: 
\begin{align}\label{eq:Parisi_PDE_f}
    \begin{dcases}
        f(1,h) &= \log \gamma_{1-q(1)}(h), \\ 
        \dot{f}(x,h) &= - \frac{\dot{q}(x)}{2} \big[f''(x,h) + x f'(x,h)^2\big], \hspace{0.5cm} x \in (0,1).
    \end{dcases}
\end{align}
The boundary condition in the first line was given by eq.~\eqref{eq:iterative_construction_g}: 
$g(1, h) = \gamma_{1-q(1)}(h)$.

\myskip 
\textbf{Remark: universality of the Parisi PDE --}
As can be already hinted by the calculation above and the method of \cite{duplantier1981comment}, the Parisi PDE 
described in eq.~\eqref{eq:Parisi_PDE_f} is actually extremely general: the specificities of the term that we wish to compute only appear 
in the boundary conditions at $x = 1$, while the evolution equation is only dependent on the ultrametric structure of the problem. 
We will see a clear example of this when computing the energetic contribution to the free entropy.

\myskip
\textbf{The $r \to 0$ limit --}
Taking the $r \to 0$ limit in eq.~\eqref{eq:entropic_term_7} yields finally: 
\begin{align*}
    - \frac{1}{2} \partial_r [\log \det \bQ]_{r = 0} = \frac{\log 2\pi}{2} + \gamma_{q_m} \star f (0,h = 0).
\end{align*}
\textbf{Solution to the Parisi PDE for the entropic contribution --}
Fortunately, with the boundary condition that we have here, the Parisi PDE of eq.~\eqref{eq:Parisi_PDE_f} is analytically solvable. 
Indeed $g(x,h)$ always remains (up to a scaling) a centered Gaussian function of $h$, or equivalently we can look for a solution in the form 
\begin{align*}
    f(x,h) &= \frac{1}{x} \log C(x) + \frac{1}{x} \log \gamma_{\omega(x)}(h),
\end{align*}
with $\omega(1) = 1 - q(1)$ and $C(1) = 1$.
This yields after some algebra simple ODEs on $\omega, C$ that are easily verified to be solved by: 
\begin{align*}
    \begin{dcases}
        \omega(x) &= \frac{1 - x q(x) - \int_x^1 q(u) \rd u}{x} = \frac{\lambda(x)}{x}, \\ 
        \log C(x) &= \frac{x}{2} \int_x^1 \frac{\rd u}{u^2} [1 + \log 2\pi \omega(u)].
    \end{dcases}
\end{align*}
We took the notation $\lambda(x)$ defined in eq.~\eqref{eq:def_lambdax}.
In particular, for every $x \in (0,1)$, we have: 
\begin{align*}
    \gamma_{q_m} \star f(x,h=0) &= \frac{1}{2} \int_x^1 \frac{\rd u}{u^2} \Big[1 + \log 2\pi \frac{\lambda(u)}{u}\Big] - \frac{1}{2 x} \log \Big[2 \pi \frac{\lambda(x)}{x}\Big] - \frac{q_m}{2 \lambda(x)}.
\end{align*}
We can take the limit of this equation as $x \to 0$.
With our notations, we have $\lambda(0) = 1 - \langle q \rangle$, $\lambda(1) = 1 - q_M$ and $\dot{\lambda}(q) = - u \dot{q}(u)$.
By integration by parts, we reach: 
\begin{align*}
    \gamma_{q_m} \star f(x,h=0) &= \frac{1}{2} \int_x^1 \frac{\rd u}{u} \Big[\frac{- u \dot{q}(u)}{\lambda(u)} - \frac{1}{u}\Big]- \frac{1}{2} [1 + \log 2\pi(1-q_M)] + \frac{1}{2 x} - \frac{q_m}{2 \lambda(x)}, \nonumber \\ 
    &=  -\frac{1}{2} \int_x^1 \rd u \frac{\dot{q}(u)}{\lambda(u)} - \frac{1}{2} \log 2\pi(1-q_M) - \frac{q_m}{2 \lambda(x)}.
\end{align*}
\noindent
\textbf{Final result for the entropic contribution --}
Therefore, taking the limit $x \to 0$, we have
\begin{align}\label{eq:entropic_frsb}
    &\partial_r [\log \det \bQ]_{r = 0} = -\log 2\pi - 2 \gamma_{q_m} \star f (0,h = 0) = \log (1-q_M) + \frac{q_m}{1-\langle q \rangle} + \int_0^1 \rd u \frac{\dot{q}(u)}{\lambda(u)}.
\end{align}
Note that eq.~\eqref{eq:entropic_frsb} is also equivalent to a formula given in Appendix II of \cite{mezard1991replica}
as can be seen by integration by parts: 
\begin{align}\label{eq:entropic_frsb_2}
   \partial_r \big[\log \det \bQ \big]_{r=0} &= \log (1-\langle q \rangle) + \frac{q_m}{1 - \langle q \rangle} - \int_0^1 \frac{\rd x}{x^2} \log \frac{\lambda(x)}{ 1- \langle q \rangle}.
\end{align}
In the IPP, one uses $\lambda(0) = 1 - \langle q \rangle$, $\lambda(1) = 1 - q_M$, and $\lambda(u) = \lambda(0) + \mathcal{O}(u^2)$.

\subsection{Energetic contribution}\label{subsec_app:energetic}

The second part of the free entropy is the energetic contribution, i.e.\ $\alpha G_{2,r}(\bQ)$, with
\begin{align*}
    G_{2,r}(\bQ) &\coloneqq \log \int_{\bbR^r} \frac{\rd \bu \rd \bv}{(2\pi)^{r}} e^{-\frac{1}{2} \sum_{a,b} Q^{ab} v^a v^b -\beta \sum_{a=1}^r \theta(u^a) + i \sum_{a=1}^r u^a v^a}.
\end{align*}
Again using the identity of eq.~\eqref{eq:duplantier}, we have:
\begin{align*}
    G_{2,r}(\bQ) &= \log \int_{\bbR^r} \frac{\rd \bu \rd \bv}{(2\pi)^{r}} e^{-\beta \sum_{a=1}^r \theta(u^a) + i \sum_{a=1}^r u^a v^a} e^{\frac{1}{2} \sum_{a,b} Q^{ab} \partial_a \partial_b} \Bigg[\prod_{a=1}^r e^{-i v^a h_a}\Bigg]_{h=0}, \\ 
    &= \log  e^{\frac{1}{2} \sum_{a,b} Q^{ab} \partial_a \partial_b} \Bigg[\prod_{a=1}^r e^{-\beta \sum_{a=1}^r \theta(h_a)}\Bigg]_{h=0}.
\end{align*}
One can notice that this equation is extremely similar to eq.~\eqref{eq:entropic_term_3}, 
but the function $\delta(h)$ has been replaced with $e^{-\beta \theta(h)}$. However, the whole procedure that we described above to obtain the Parisi PDE does not change at all, since it did not depend on the specifics of this function:
the PDE itself remains the same, only the boundary condition at $x = 1$ will be different.
In the end, this yields:
\begin{align*}
    \partial_r [G_{2,r}(\bQ)]_{r = 0} &= \gamma_{q(0)} \star f(x = 0, h = 0),
\end{align*}
with $f(x,h)$ given as the solution to the Parisi PDE with specific boundary condition at $x=1$:
\begin{align}\label{eq:Parisi_PDE_f_energetic}
    \begin{dcases}
        f(1,h) &= \log [\gamma_{1-q(1)} \star e^{-\beta \theta}](h), \\ 
        \dot{f}(x,h) &= - \frac{\dot{q}(x)}{2} \big[f''(x,h) + x f'(x,h)^2\big], \hspace{0.5cm} x \in (0,1).
    \end{dcases}
\end{align}
Note that one can equivalently write this PDE in terms of the parameter $q$ rather than $x$ by a change of variable $q = q(x)$, as described e.g.\ in \cite{urbani2018statistical}.

\subsection{Recovering the RS result from the full RSB equations}\label{subsec_app:rs_from_rsb}

In this paragraph we show that eq.~\eqref{eq:frsb_eq_rs} is equivalent to eq.~\eqref{eq:q_RS_eq_new}.
We denote $q_0 = q$ coherently with the RS computation.
One computes easily that  
\begin{align*}
    \gamma_{1-q} \star e^{-\beta \theta} (h) &= 1 - (1-e^{-\beta}) H \Big(\frac{-h}{\sqrt{1-q}}\Big).
\end{align*}
In particular, we have:
\begin{align*}
    \Big[\gamma_{1-q} \star e^{-\beta \theta}\Big]' (h) &= \frac{1-e^{-\beta}}{\sqrt{1-q}} H' \Big(\frac{-h}{\sqrt{1-q}}\Big).
\end{align*}
Therefore eq.~\eqref{eq:frsb_eq_rs} reads: 
\begin{align*}
    \frac{q}{(1-q)^2} &= \alpha \int \rd h \frac{e^{-\frac{h^2}{2q}}}{(1-q)\sqrt{2 \pi q}}
    \Bigg\{\frac{(1-e^{-\beta}) H' \Big(\frac{-h}{\sqrt{1-q}}\Big)}{1 - (1-e^{-\beta}) H \Big(\frac{-h}{\sqrt{1-q}}\Big)}\Bigg\}^2,
\end{align*}
or equivalently:
\begin{align}\label{eq:frsb_to_rs_1}
    \frac{q}{1-q} &= \alpha \int \mcD \xi
    \Bigg\{\frac{(1-e^{-\beta}) H' \Big(\xi \sqrt{\frac{q}{1-q}}\Big)}{1 - (1-e^{-\beta}) H \Big(\xi \sqrt{\frac{q}{1-q}}\Big)}\Bigg\}^2.
\end{align}
Since $H'(x) = - e^{-x^2/2}/\sqrt{2\pi}$, letting 
$f(\xi) \coloneqq 1 - (1-e^{-\beta}) H [\xi \sqrt{q/(1-q)}]$ we can rewrite eq.~\eqref{eq:frsb_to_rs_1}
and use an integration by parts:
\begin{align*}
    \frac{q}{1-q} &= -\alpha (1-e^{-\beta}) \sqrt{\frac{1-q}{q}} \int \rd \xi \frac{e^{-\frac{\xi^2}{2(1-q)}}}{2\pi} \Big[-\frac{f'(\xi)}{f(\xi)^2}\Big], \\ 
    &=  -\alpha (1-e^{-\beta}) \sqrt{\frac{1-q}{q}} \int \rd \xi \frac{e^{-\frac{\xi^2}{2(1-q)}}}{2\pi} \frac{\xi}{1-q} \frac{1}{f(\xi)}, \\ 
    &=  \alpha (1-e^{-\beta}) \sqrt{\frac{1}{q(1-q)}} \int \mcD \xi \frac{\xi H' \Big(\xi \sqrt{\frac{q}{1-q}}\Big)}{1 - (1-e^{-\beta}) H \Big(\xi \sqrt{\frac{q}{1-q}}\Big)},
\end{align*}
which is equivalent to eq.~\eqref{eq:q_RS_eq_new}.

\section{Technicalities of the algorithmic FRSB procedure}\label{sec_app:numerics_frsb}
\subsection{Technicalities of the derivation of the algorithmic procedure}\label{subsec_app:derivation_algorithmic_frsb}

\noindent
We give here some details on the arising of eqs.~\eqref{eq:frsb_procedure_iii}-\eqref{eq:frsb_procedure_vi}.
Recall that here all the quantities are considered at zero-temperature, with the scaling of eq.~\eqref{eq:frsb_zerotemp_scaling}.
\begin{itemize}[leftmargin=*]
    \item Eq.~\eqref{eq:frsb_procedure_iii} is a general relation between $q^{-1}$, $f$ and $\Lambda$ when eq.~\eqref{eq:frsb_eqs} is satisfied. It is explained for instance in \cite{franz2017universality}, 
    see eq.~(B.4).
    \item Eq.~\eqref{eq:frsb_procedure_iv} is a consequence of the general relation between $q^{-1}(x)$ (the function corresponding to the overlap matrix $\bQ^{-1}$) and $q(x)$ in the full RSB ansatz,
    which is (see e.g.\ eq.~(B.9) in \cite{franz2017universality}): 
    \begin{align*}
        \frac{1}{\lambda(x)} - \frac{1}{\lambda(0)} &= - x q^{-1}(x) + \int_0^x \rd y \, q^{-1}(y) \hspace{1cm} \textrm{and} \hspace{1cm} \lambda(0) = \sqrt{-\frac{q(0)}{q^{-1}(0)}}.
    \end{align*}
    Discretization of this relation yields eq.~\eqref{eq:frsb_procedure_iv}.
    \item One can invert eq.~\eqref{eq:def_lambdax} to obtain $q(x)$ as a function of $\lambda(x)$ via:
    \begin{align*}
        q(x) &= 1 - \frac{\lambda(x)}{x} + \int_x^1 \frac{\rd y}{y^2} \lambda(y).
    \end{align*}
    It is the discretization of this equation that yields eq.~\eqref{eq:frsb_procedure_v}.
    \item Eq.~\eqref{eq:frsb_procedure_vi} is a consequence of the boundary condition of eq.~\eqref{eq:frsb_eq_1} taken at $x = 1$, followed by a change of variable from $x$ to $q$ in the parameters 
    of the functions $\Lambda,f,\lambda$. After these procedures, eq.~\eqref{eq:frsb_eq_1} becomes, for the \emph{unrescaled variables} and any $\beta \geq 0$: 
    \begin{align*}
        \frac{q(0)}{\lambda(q(0))^2} + \int_{q(0)}^{q(1)} \frac{\rd p}{\lambda(p)^2} &= \alpha \int \rd h \Lambda(q(1), h) f'(q(1), h)^2.
    \end{align*}
    After taking the $\beta \to \infty$ limit, this yields for the variables that are rescaled as $\beta \to \infty$ according to eq.~\eqref{eq:frsb_zerotemp_scaling} (dropping the $\infty$ subscript):
    \begin{align*}
        \frac{q_0}{\lambda(q_0)^2} + \int_{q_0}^{1} \frac{\rd p}{\lambda(p)^2} &= \alpha \int \rd h \Lambda(1, h) f'(1, h)^2.
    \end{align*}
    Moreover, $f'(1,h) = - (h / \chi) \indi \{h \in (0,\sqrt{2\chi})\}$. 
    Therefore, rescaling then $t = h / \sqrt{2\chi}$ (and using the abusive notation $\Lambda(1,h) = \Lambda(1,t)$) we have:
    \begin{align}\label{eq_app:frsb_iii_1}
        \frac{q_0}{\lambda(q_0)^2} + \int_{q_0}^{1} \frac{\rd p}{\lambda(p)^2} &= \frac{2^{3/2}\alpha}{\sqrt{\chi}} \int_0^{1} \rd t \Lambda(1, t) \, t^2.
    \end{align}
    We focus on the left-hand side of this last equation, in the $k$-RSB ansatz. 
    We first use that $\lambda(q) = \chi + \int_q^1 \rd p \, x(p)$, a simple consequence of eq.~\eqref{eq:def_lambdax}, after change of variables and rescaling.
    Therefore, we have:
    \begin{align*}
       \int_{q_0}^{1} \frac{\rd p}{\lambda(p)^2} &= \sum_{i=0}^{k-1} \int_{q_i}^{q_{i+1}} \frac{\rd p}{\Big[\chi + \sum_{j=i+1}^{k-1} (q_{j+1} - q_j) x_j + (q_{i+1} - p) x_i\Big]^2}, \\ 
        &= \sum_{i=0}^{k-1} \frac{(q_{i+1} - q_i)}{\Big[\chi + \sum_{j=i+1}^{k-1} (q_{j+1} - q_j) x_j\Big] \Big[\chi + \sum_{j=i}^{k-1} (q_{j+1} - q_j) x_j\Big]}.
    \end{align*}
    Using the convention $q_{-1} = 0$ and $x_{-1} = 0$, we therefore reach from eq.~\eqref{eq_app:frsb_iii_1} that:
    \begin{align*}
      \sum_{i=0}^k \frac{(q_i - q_{i-1})}{\Big[\chi + \sum_{j=i+1}^k (q_j - q_{j-1}) x_{j-1}\Big]\Big[\chi + \sum_{j=i}^k (q_j - q_{j-1}) x_{j-1}\Big]}
      &= \frac{2^{3/2}\alpha}{\sqrt{\chi}} \int_0^{1} \rd t \, \Lambda( 1,t) \, t^2,
    \end{align*}
    which is precisely eq.~\eqref{eq:frsb_procedure_vi}.
\end{itemize}

\subsection{Numerical results of the procedure}\label{subsec_app:numerical_results_frsb}

In Fig.~\ref{fig_app:convergence_frsb} we present the results of typical iterations of the algorithmic procedure described above.
For different values of $\alpha$ and the RSB parameter $k$ we show the evolution of the estimates of $f^\star(\alpha)$, the susceptibility $\chi$, and the function $q(x)$, along the iterations.
In all the cases implemented we see power-law convergence to the solution, and very consistent results when varying the parameters used in the algorithm (in particular increasing $k$).

\begin{figure}
  \centering
\includegraphics[width=0.9\textwidth]{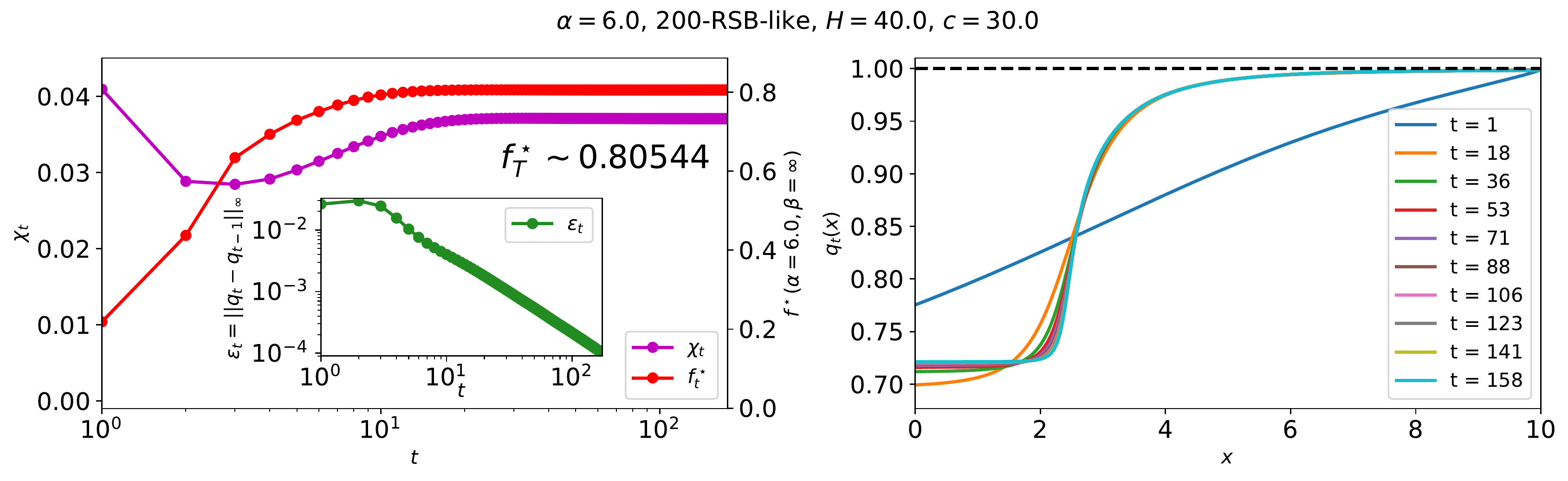}
\includegraphics[width=0.9\textwidth]{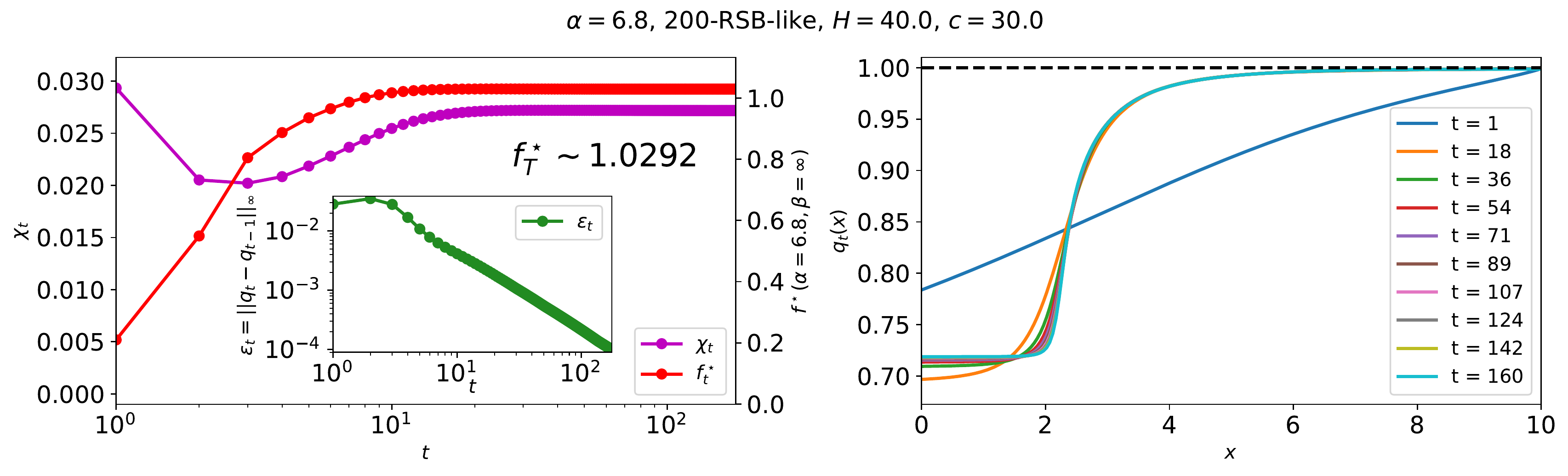}
\caption{
      Illustration of the convergence of the Full RSB procedure for different values of $\alpha$, 
      for $k =200, c = 30, H = 40$ (see Section~\ref{subsec_app:convolutions} for the definitions of $H,c$).
      On the left we show the convergence of $\chi$ and $f^\star(\alpha)$ along the iterations, as well as (in inset) the evolution of the error $\|q_t - q_{t-1}\|_\infty$ up to the threshold $10^{-4}$ we took for convergence. 
      On the right, we show the evolution of $q(x)$ along the iterations. We find very consistent behaviors when varying the parameter $k$, indicating that our simulations indeed capture well the Full RSB limit.
      We use $x_\mathrm{max} = 10$, well validated by the functions $q(x)$ we obtain.
      \label{fig_app:convergence_frsb}}
\end{figure}

\subsection{Some details on the implementation of convolutions}\label{subsec_app:convolutions}

\subsubsection{Convolutions via DFTs}

In order to implement the algorithmic procedure of Section~\ref{subsec:zero_temp_algorithmic_frsb}, we use a discrete Fourier transform approach.
We refer to \cite{getreuer2013survey} for a review on Gaussian convolution algorithms.
The goal is to compute the convolution of a centered Gaussian $\gamma_\omega$ with variance $\omega > 0$ and a function 
$f(h)$: 
\begin{align*}
    \gamma_\omega \star f(h) &= \int \rd z \, \gamma_\omega(z) \, f(h-z).
\end{align*}
We fix $N \in \bbN^\star$ and $H > 0$, and we consider a grid $h_\mu = \mu H / N$, with $\mu \in \{-N, \cdots, N\}$.
In order to leverage analytical formulas for the DFT of the Gaussian, we use a Shannon-Whittaker interpolation for $f$, i.e.\  
we approximate $f$ as: 
\begin{align*}
   f(h) &\simeq \sum_{\nu=-N}^N f_\nu \, \varphi\Big[\frac{Nh}{H} - \nu\Big],
\end{align*}
with $\varphi(x) = \mathrm{sinc}(x) = \sin(\pi x)/ (\pi x)$. Since $\varphi(\nu) = 0$ for all $\nu \in \bbZ^\star$, and $\varphi(0) = 1$, we have 
$f_\mu = f(h_\mu)$.
This approximation transfers into an approximation for $\gamma_\omega \star f$ as: 
\begin{align*}
    \gamma_\omega \star f(h) \simeq \sum_{\nu=-N}^N f_\nu \, (\gamma_\omega \star \varphi_\nu)(h),
\end{align*}
with $\varphi_\nu(h) \coloneqq \varphi(Nh / H - \nu)$.
Thus, we have, with $(\gamma_\omega \star f)_\mu = \gamma_\omega \star f(h_\mu)$, 
and using $\varphi_\nu(x) = \varphi_0(x - h_\nu)$:
\begin{align}\label{eq:dft_primal}
    (\gamma_\omega \star f)_\mu \simeq \sum_{\nu=-N}^N f_\nu (\gamma_\omega \star \varphi_\nu)_\mu = \sum_{\nu=-N}^N f_\nu (\gamma_\omega \star \varphi_0)_{\mu-\nu}.
\end{align}
Note that we naturally extended $(\gamma_\omega \star \varphi_0)_\mu$ to all $\mu \in \bbZ$, since these coefficients have an analytic expression. 
In the same way, we extend $f_\nu = 0$ if $|\nu| > N$.
For a general sequence $(f_\nu)_{\nu=-N}^N$, we define its Discrete Fourier Transform (DFT) as, for $k \in \{0, \cdots, 2N\}$:
\begin{align}
    \label{eq:def_dft}
    \begin{dcases}
        \hat{f}_k &= \sum_{\mu=-N}^N e^{-\frac{2\pi ik (\mu + N)}{2N+1}} f_\mu = e^{-\frac{2\pi i k N}{2N+1}} \sum_{\mu=-N}^N e^{-\frac{2\pi i k\mu }{2N+1}} f_\mu, \\ 
        f_\mu &= \frac{1}{2N+1} e^{\frac{2\pi i k N}{2N+1}} \sum_{k=0}^{2N} e^{\frac{2\pi i k \mu}{2N+1}} \hat{f}_k.
    \end{dcases}
\end{align}
Taking the DFT of eq.~\eqref{eq:dft_primal}, one finds:
\begin{align}\label{eq:dft_convolution_1}
    \widehat{\gamma_\omega \star f}_k &\simeq e^{\frac{2\pi i k N}{2N+1}} \, \hat{f}_k \, (\widehat{\gamma_\omega \star \varphi_0})_k.
\end{align}
Moreover, we define the Fourier transform as $\tilde{f}(\xi) \coloneqq \int \rd x \, f(x) \, e^{- 2 i \pi x \xi}$, and have
then easily $\tilde{\varphi}_0(\xi) = (H/N) \indi\{|\xi| \leq N/(2H)\}$.
The Fourier transform of the convolution is $\widetilde{f \star g}(\xi) = \tilde{f}(\xi) \tilde{g}(\xi)$.
This yields, via inverse Fourier transformation:
\begin{align*}
    (\gamma_\omega \star \varphi_0)_\mu &= \frac{H}{N} \int_{|\xi| \leq \frac{N}{2H}} \rd \xi \, e^{- 2 \pi^2 \omega \xi^2 + \frac{2i \pi H \mu \xi}{N}}
    = \int_{|\xi| \leq \frac{1}{2}} \rd \xi \, e^{- \frac{2 \pi^2 N^2 \omega \xi^2}{H^2} + 2i \pi \mu \xi}.
\end{align*}
Therefore, we have by eq.~\eqref{eq:def_dft}:
\begin{align*}
    (\widehat{\gamma_\omega \star \varphi_0})_k &= e^{-\frac{2\pi i k N}{2N+1}} \int_{|\xi| \leq \frac{1}{2}} \rd \xi \, e^{- \frac{2 \pi^2 N^2 \omega \xi^2}{H^2}}\sum_{\mu=-N}^N e^{-\frac{2\pi i k\mu }{2N+1}  + 2i \pi \mu \xi}.
\end{align*}
Taking $N \gg 1$, the term on the right is well approximated by the Dirac comb:
\begin{align*}
    \sum_{\mu=-N}^N e^{-\frac{2\pi i k\mu }{2N+1}  + 2i \pi \mu \xi} &\simeq \sum_{n \in \bbZ} \delta\Big(\xi - \frac{k}{2N+1} - n\Big).
\end{align*}
However, since $|\xi| \leq 1/2$ and $k \in \{0,\cdots,2N\}$, this implies: 
\begin{align*}
    (\widehat{\gamma_\omega \star \varphi_0})_k &\underset{N \to \infty}{\simeq} e^{-\frac{2\pi i k N}{2N+1}}\times
    \begin{dcases}
       \, e^{- \frac{2 \pi^2 N^2 \omega}{H^2} \Big[\frac{k}{2N+1}\Big]^2} \hspace{1cm} &\textrm{if } k \leq N, \\
       \, e^{- \frac{2 \pi^2 N^2 \omega}{H^2} \Big[\frac{k}{2N+1} - 1\Big]^2} \hspace{1cm} &\textrm{if } k > N.
    \end{dcases}
\end{align*}
Plugging it back into eq.~\eqref{eq:dft_convolution_1}, we finally obtain the formula we use for the DFT of the convolution $\gamma_\omega \star f$:
\begin{align*}
    \widehat{\gamma_\omega \star f}_k &\simeq \hat{f}_k \times
    \begin{dcases}
        \,e^{- \frac{2 \pi^2 N^2 \omega}{H^2} \Big[\frac{k}{2N+1}\Big]^2} \hspace{1cm} &\textrm{if } k \leq N, \\
        \,e^{- \frac{2 \pi^2 N^2 \omega}{H^2} \Big[\frac{k}{2N+1} - 1\Big]^2} \hspace{1cm} &\textrm{if } k > N.
    \end{dcases}
\end{align*}

\subsubsection{Taking a large enough value of \texorpdfstring{$N$}{N}}

Note that in order for the Gaussian convolutions to be numerically well-defined, we need the 
spacing in the grid we take on $h$ to be much smaller than the standard deviation of the Gaussians, that is we need for any $(q(x), q(x) + \rd x \, \dot{q}(x))$:
\begin{align*}
    \frac{H}{N} \ll \sqrt{\frac{\rd x \, \dot{q}(x)}{2 \chi}}.
\end{align*}
Note that, as shown in \cite{franz2017universality}, and as one can also verify from Fig.~\ref{fig:q_T0}, we have the following scaling 
as $x \to \infty$: $q(x) \sim 1 - A / x^2$, with $A > 0$. Therefore, $\dot{q}(x_\mathrm{max}) \sim 2 A / x_\mathrm{max}^3$.
Since we take $\rd x \sim x_\mathrm{max}/k$ in our numerical procedure, we have that in order for our procedure to be valid 
we need
\begin{align*}
    \frac{H}{N} \ll \sqrt{\frac{A}{k \chi x_\mathrm{max}^2}}.
\end{align*}
In practice, we find typically $\chi/A \sim 10^{-1}$, so that we will impose $N \gg N_0$, with 
\begin{align*}
    N_0 &\coloneqq \sqrt{k} \, H \, x_\mathrm{max}.
\end{align*}
In practice, we consider $N = c N_0$ (we often take $c = 30$) with a constant $c \gg 1$ in order to be well into the regime $N \gg N_0$, and still have a reasonable computational time.

\subsection{Bounds on the injectivity threshold}\label{subsec_app:numerics_frsb_threshold}

Let us detail the results of our numerical computation of $\alpha_\inj^\FRSB$, illustrated in Fig.~\ref{fig:alpha_inj_frsb_summary}.
For a given value of all parameters of the algorithm detailed in Section~\ref{subsec:zero_temp_algorithmic_frsb}, we ran 
Brent's method to find the zero of $f^\star_\FRSB(\alpha) - 1$, with a tolerance of $10^{-4}$ on the value of $\alpha$.
For all values of $\alpha$, we iterated the FRSB equations until $\norm{q^{t+1} - q_t}_\infty \leq \epsilon = 10^{-5}$.
We ran this procedure for different values of $k \in \{30, 50, 100, 200\}$, $x_\mathrm{max} \in \{10,11,12,13,14,15\}$, 
$H \in \{40,60\}$, and $c = 30$ (recall that $H$ and $c$ are defined in Section~\ref{subsec_app:convolutions}).
In Fig.~\ref{fig:alpha_inj_frsb_summary} the runs with different values of $H$ and $c$ are aggregated.

\section{``Escape through a mesh'' theorem: an alternative proof of Theorem~\ref{thm:bound_Gordon}}\label{sec_app:mesh}
In this appendix, we provide an alternative proof of Theorem~\ref{thm:bound_Gordon}
using Gordon's ``escape through a mesh'' theorem~\cite{gordon1988milman}. 
This seminal result establishes upper bounds on the probability of a random set intersecting a fixed set.
In an earlier version of this paper, we applied Gordon's ``mesh'' theorem to derive a highly suboptimal bound, $\alpha_\inj \lesssim 23.54$, due to loose estimates in our calculations. 
During the review process, however, we refined this approach, ultimately achieving the same replica-symmetric upper bound as that obtained through Gordon's min-max inequality.

\myskip 
For $m \geq 1$, we denote $a_m \coloneqq \EE[\|\bg\|_2]$, for $\bg \sim \mcN(0, \Id_m)$.
One can easily show the bound \cite{vershynin2018high}:
\begin{align}\label{eq:bound_am}
    \frac{m}{\sqrt{m+1}} \leq a_m \leq \sqrt{m}.
\end{align}
Moreover, for a closed set 
$S \subseteq \bbR^m$, we define its \emph{Gaussian width} as 
\begin{align*}
    \omega(S) \coloneqq \EE \max_{\bx \in S} [\bg \cdot \bx].
\end{align*}
We are now ready to state Gordon's ``escape through a mesh'' theorem.
\begin{theorem}[\cite{gordon1988milman}]
    \noindent
    Let $S \subseteq \mcS^{m-1}$ be a closed subset such that $\omega(S) < a_{m-n}$.
    Let $V$ be a uniformly-sampled random $n$-dimensional subspace of $\bbR^m$. 
    Then  
    \begin{align*}
        \bbP\Big[V \cap S = \emptyset\Big] \geq 1 - \frac{7}{2} \exp \Big\{-\frac{1}{18} \Big(a_{m-n} - \omega(S)\Big)^2\Big\}.
    \end{align*}
\end{theorem}
Applying this theorem to Proposition~\ref{prop:injectivity_random_intersection}, we directly reach 
\begin{corollary}\label{cor:gordon_mesh_injectivity}
    \noindent
    Assume that $a_{m-n} - \omega(C_{m,n} \cap \mcS^{m-1}) \to \infty$ as $m,n \to \infty$.
    Then $p_{m,n} \to 1$, i.e.\ $\varphi_\bW$ is injective w.h.p.
\end{corollary}
We can now state the core of our argument, which is an upper bound for the Gaussian width $\omega(C_{m,n} \cap \mcS^{m-1})$, and is proven in Appendix~\ref{subsec_app:ub_gaussian_width}.
\begin{proposition}\label{prop:ub_gaussian_width}
    \noindent
    Assume $\alpha > 2$.
    Then  
    \begin{align}\label{eq:ub_gaussian_width}
        \limsup_{m \to \infty} \frac{\omega(C_{m,n} \cap \mcS^{m-1})^2}{n}  &\leq \alpha \left(1 - \int_{0}^{t_\alpha} \mcD x \, x^2\right),
    \end{align}
    where recall that $\mcD x \coloneqq e^{-x^2/2} \rd x / \sqrt{2\pi}$ is the standard Gaussian measure. Moreover, $t_\alpha$ is the unique value of $t \geq 0$ such that 
    \begin{align}\label{eq:def_talpha}
        \int_{t_\alpha}^\infty \mcD x &= \frac{1}{\alpha}.
    \end{align}
\end{proposition}
Notably -- as discussed below -- we expect that our proof can be improved by classical concentration arguments, to yield that eq.~\eqref{eq:ub_gaussian_width} holds as an equality for $\lim(\omega^2/n)$, although we do not require it 
to prove Theorem~\ref{thm:bound_Gordon}.
Since $a_{m-n} \geq \sqrt{\alpha-1} \sqrt{n} (1 - \smallO(1))$ by eq.~\eqref{eq:bound_am}, 
Proposition~\ref{prop:ub_gaussian_width} and Corollary~\ref{cor:gordon_mesh_injectivity}
imply that $p_{m,n} \to 1$ whenever 
\begin{align*}
    \alpha - 1 > \alpha \left(1 - \int_{0}^{t_\alpha} \mcD x \, x^2\right),
\end{align*}
i.e.\ whenever $\alpha > \alpha_\inj^\mathrm{mesh}$ with $\alpha_\inj^\mathrm{mesh}$ the solution to
\begin{align}\label{eq:alpha_inj_mesh}
    \begin{dcases}
    \alpha \int_0^{t_{\alpha}} \mcD x \, x^2 &= 1, \\
    \int_{t_\alpha}^\infty \mcD x &= \frac{1}{\alpha}.
    \end{dcases}
\end{align}
We recognize in eq.~\eqref{eq:alpha_inj_mesh} the replica-symmetric threshold prediction of eqs.~\eqref{eq:chi_RS} and \eqref{eq:fstar_RS}, with $t = \sqrt{2\chi_\RS}$. 
Thus, $\alpha_\inj^\mathrm{mesh} = \alpha_\inj^\RS$, and this ends our alternative proof of Theorem~\ref{thm:bound_Gordon}.

\subsection{Proof of Proposition~\ref{prop:ub_gaussian_width}}\label{subsec_app:ub_gaussian_width}

    We denote $\omega = \omega(C_{m,n} \cap \mcS^{m-1})$.
    By weak duality, we have
    \begin{align*}
        \omega &= \EE \, \sup_{\bx \in C_{m,n}} \inf_{\lambda \in \bbR} \Big[\bg \cdot \bx - \frac{\lambda}{2} (\|\bx\|^2- 1) \Big], \\ 
        &\leq \EE \, \inf_{\lambda \in \bbR} \sup_{\bx \in C_{m,n}} \Big[\bg \cdot \bx - \frac{\lambda}{2} (\|\bx\|^2- 1) \Big], \\
        &\leq \EE \, \inf_{\lambda > 0} \sup_{\bx \in C_{m,n}} \Big[\bg \cdot \bx - \frac{\lambda}{2} (\|\bx\|^2- 1) \Big]. 
    \end{align*}
    An element $\bx \in C_{m,n}$ can be parametrized by a subset $S \subseteq [m]$ with $|S| < n$, and a set of values $\{x_\mu\}$, with $x_\mu > 0$ for $\mu \in S$ and $x_\mu \leq 0$ for $\mu \notin S$.
    This yields:
    \begin{align*}
        \omega &\leq \EE \, \inf_{\lambda > 0} \max_{\substack{S \subseteq [m] \\ |S| < n}} \Big[\frac{\lambda}{2} + \sum_{\mu \in S} \sup_{x > 0} \Big(x g_\mu - \frac{\lambda}{2} x^2\Big) 
        + \sum_{\mu \notin S} \sup_{x \leq 0} \Big(x g_\mu - \frac{\lambda}{2} x^2\Big) 
        \Big], \\ 
        &\leq \EE \, \inf_{\lambda > 0} \max_{\substack{S \subseteq [m] \\ |S| < n}} \Big[\frac{\lambda}{2} + \frac{1}{2 \lambda} \sum_{\mu \in S} g_\mu^2 \indi\{g_\mu \geq 0\} + \frac{1}{2 \lambda} \sum_{\mu \notin S} g_\mu^2 \indi\{g_\mu \leq 0\} 
        \Big], \\
        &\leq \EE \, \Bigg[\Bigg(\max_{\substack{S \subseteq [m] \\ |S| < n}} \Big[\sum_{\mu \in S} g_\mu^2 \indi\{g_\mu \geq 0\} + \sum_{\mu \notin S} g_\mu^2 \indi\{g_\mu \leq 0\}\Big]\Bigg)^{1/2} \Bigg].
    \end{align*}
    We used that $\inf_{\lambda > 0}[\lambda + T / \lambda] = 2\sqrt{T}$, for $T> 0$.
    Since the law of $\bg$ is invariant under permutation of the indices, 
    we can assume $g_1 \geq g_2 \geq \cdots \geq g_p \geq 0 > g_{p+1} \geq \cdots \geq g_m$, 
    with $p = p(\bg) \in [m]$.
    It is then straightforward to check that:
    \begin{align*}
        \max_{\substack{S \subseteq [m] \\ |S| < n}} \Big[\sum_{\mu \in S} g_\mu^2 \indi\{g_\mu \geq 0\} + \sum_{\mu \notin S} g_\mu^2 \indi\{g_\mu \leq 0\}\Big] 
        &= \sum_{\mu=1}^{n-1} g_\mu^2 \indi\{g_\mu \geq 0\} + \sum_{\mu=p+1}^m g_\mu^2.
    \end{align*}
    Therefore 
    \begin{align}\label{eq:ub_omega}
          \frac{1}{\sqrt{n}} \omega 
            \leq \EE \sqrt{\frac{1}{n} \sum_{\mu=1}^{n-1} g_\mu^2 \indi\{g_\mu \geq 0\} + \frac{1}{n}\sum_{\mu=p+1}^m g_{\mu}^2}
            \leq \sqrt{ \frac{1}{n} \EE\sum_{\mu=1}^{n-1} g_\mu^2 \indi\{g_\mu \geq 0\} + \frac{1}{n}\EE \sum_{\mu=p+1}^m g_{\mu}^2}.
    \end{align}
    Note that -- although it is not needed in what follows -- a careful analysis based on concentration would show that the first and second inequalities of eq.~\eqref{eq:ub_omega} actually hold as equalities, 
    up to a multiplicative term $1 + \smallO(1)$ as $n, m \to \infty$.
    The second term inside the square root in eq.~\eqref{eq:ub_omega} can be computed easily: 
    \begin{align}\label{eq:gmu_half}
        \frac{1}{n} \EE \sum_{\mu=p+1}^m g_{\mu}^2 &= \frac{1}{n} \EE \sum_{\mu=1}^m g_\mu^2 \indi\{g_\mu < 0\} = 
        \frac{m}{2n} \xrightarrow[m \to \infty]{} \frac{\alpha}{2}.
    \end{align}
    We now show: 
    \begin{align}\label{eq:to_show_gwidth}
        \lim_{m \to \infty} \frac{1}{m} \EE \sum_{\mu=1}^{n-1} g_\mu^2 \indi\{g_\mu \geq 0\} = \int_{t_\alpha}^\infty \mcD x \, x^2 = \frac{1}{2} - \int_0^{t_\alpha} \mcD x \, x^2,
    \end{align}
    with $t_\alpha$ defined in eq.~\eqref{eq:def_talpha}. Combining eqs.~\eqref{eq:gmu_half} and \eqref{eq:to_show_gwidth} in eq.~\eqref{eq:ub_omega} ends the proof of Proposition~\ref{prop:ub_gaussian_width}.
    To prove eq.~\eqref{eq:to_show_gwidth} we rely on the following technical lemma, proven in Appendix~\ref{subsec_app:proof_control_seq_gaussians}.
    \begin{lemma}\label{lemma:control_seq_gaussians}
        \noindent
        Let $m \geq 1$, and $z_1, \cdots, z_m \iid \mcN(0,1)$.
        Denote as $g_1 \geq \cdots \geq g_m$ the non-increasing ordering of $(z_\mu)_{\mu=1}^m$.
        For $\beta \geq 1$, let $t_\beta \in \bbR$ be the unique solution to $\beta \int_{t_\beta}^\infty \mcD x = 1$.
        Then, for any $\beta \in [1, \infty)$, if $n \in \{1, \cdots, m\}$ with $m/n \to \beta$ as $m \to \infty$:
        \begin{itemize}
            \item[$(i)$] $g_{n} \xrightarrow[m \to \infty]{a.s.} t_\beta$.
            \item[$(ii)$]
            We have
            \begin{equation*}
           \lim_{\delta \downarrow 0} \limsup_{m \to \infty} \frac{1}{m} \EE\left[\sum_{(1-\delta)n\leq \mu \leq n} g_\mu^2 \right] = 0.
            \end{equation*}
        \end{itemize}
    \end{lemma}
    By the law of large numbers and the triangular inequality, 
    \begin{align*}
         \left|\frac{1}{m} \EE \sum_{\mu=1}^{n-1} g_\mu^2 \indi\{g_\mu \geq 0\} - \int_{t_\alpha}^\infty \mcD x \, x^2\right| 
         &\leq 
         \frac{1}{m} \EE 
         \left|\sum_{\mu=1}^{n-1} g_\mu^2 \indi\{g_\mu \geq 0\} - \sum_{\mu=1}^{m} g_\mu^2 \indi\{g_\mu \geq t_\alpha\}\right| + \smallO(1)
         , \\
         &\leq 
           \underbrace{\frac{1}{m} \EE 
         \sum_{\mu=1}^{n-1} g_\mu^2 \indi\{g_\mu < t_\alpha\}}_{\coloneqq I_1} +
           \underbrace{\frac{1}{m} \EE 
        \sum_{\mu=n}^{m} g_\mu^2 \indi\{g_\mu \geq t_\alpha\}}_{\coloneqq I_2} + \smallO(1).
    \end{align*}
    We now show that $I_1, I_2 \to 0$. Letting $\delta > 0$ and $A_{\delta} \coloneqq \{g_{\lceil(1-\delta) n\rceil} < t_\alpha\}$, then $\bbP[A_\delta] \to 0$ by $(i)$ of Lemma~\ref{lemma:control_seq_gaussians}.
    By the Cauchy-Schwarz inequality:
    \begin{align*}
        I_1 &=
         \frac{1}{m}  \EE  \sum_{\mu=1}^{n-1} g_\mu^2 \indi\{g_\mu < t_\alpha\} \indi\{A_\delta\}
         + \frac{1}{m}  \EE  \sum_{\mu=1}^{n-1} g_\mu^2 \indi\{g_\mu < t_\alpha\} \indi\{g_{\lceil(1-\delta) n\rceil} \geq t_\alpha\}, \\  
         &\leq \sqrt{\bbP[A_\delta]} \frac{[\EE(\|\bg\|^4)]^{1/2}}{m} + \frac{1}{m} \EE  \sum_{\mu=\lceil(1-\delta) n\rceil}^{n} g_\mu^2.
    \end{align*}
    We then deduce that $I_1 \to 0$ as $m \to \infty$ since $\EE(\|\bg\|^4) = \mcO(m^2)$ and by taking the $\delta \downarrow 0$ limit using $(ii)$ of Lemma~\ref{lemma:control_seq_gaussians}.
    The proof that $I_2 \to 0$ follows exactly the same lines. Together, this implies eq.~\eqref{eq:to_show_gwidth}, which ends the proof of Proposition~\ref{prop:ub_gaussian_width} as detailed above.
$\qed$

\subsection{Proof of Lemma~\ref{lemma:control_seq_gaussians}}\label{subsec_app:proof_control_seq_gaussians}

We start with $(i)$.
Let $\hmu_\bg \coloneqq (1/m) \sum_{\mu=1}^m \delta_{z_\mu} = (1/m) \sum_{\mu=1}^m \delta_{g_\mu}$. By definition of $g_n$, 
$\hmu_\bg((g_n, \infty)) = (n-1)/m$.
If $\xi \sim \mcN(0,1)$, by the Glivenko-Cantelli theorem:
\begin{equation*}
    |\hmu_\bg((g_n, \infty)) - \bbP[\xi > g_n]| \leq \sup_{\theta \in \bbR} |\hmu_\bg((\theta, \infty)) - \bbP[\xi > \theta]| \xrightarrow[m \to \infty]{a.s.} 0.
\end{equation*}
Therefore since $n/m \to 1/\beta$:
\begin{align*}
    \bbP[\xi > g_n] = \int_{g_n}^\infty \mcD x \xrightarrow[m \to \infty]{a.s.} \frac{1}{\beta}.
\end{align*}
Since $t \mapsto \int_t^\infty \mcD x$ is a smooth and strictly decreasing function, this implies that $g_n \xrightarrow[m \to \infty]{a.s.} t_\beta$.

\myskip 
We now prove $(ii)$.
 For any $t \in(0,1)$ we have by Jensen's inequality and the union bound:
 \begin{align*}
    \frac{1}{m} \EE \sum_{(1-\delta) n\leq \mu \leq n} g_\mu^2 \leq \frac{2}{tm} \log 
    \EE \exp\Bigg\{\frac{t}{2} \sum_{(1-\delta) n\leq \mu \leq n} g_\mu^2\Bigg\} 
    &\leq \frac{2}{tm} \log \EE \max_{\substack{I \subseteq [m] \\ |I| = \delta n}} \exp\Bigg\{\frac{t}{2} \sum_{\mu \in I} g_\mu^2\Bigg\}, \\ 
     &\leq \frac{2}{tm} \log\sum_{\substack{I \subseteq [m] \\ |I| = \delta n}} \EE \exp\Bigg\{\frac{t}{2} \sum_{\mu \in I} g_\mu^2 \Bigg\}, \\ 
     &\leq \frac{2}{tm} \log\left[\binom{m}{\delta n} (1-t)^{-\delta n/2}\right].
 \end{align*}
 Taking $t = 1/2$ and using the bound $\binom{p}{k} \leq 2^{p H(k/p)}$ with $H(q) \coloneqq -q \log q - (1-q) \log(1-q)$ the binary entropy function, yields (recall that $m/n \to \beta$):
 \begin{align*}
    \frac{1}{m} \EE \sum_{(1-\delta) n\leq \mu \leq n} g_\mu^2 \leq 4 H \left(\frac{\delta}{\beta}\right) \log 2 + \frac{2 \delta}{\beta} \log 2 + \smallO(1).
 \end{align*}
 Taking the limit $m \to \infty$ followed by $\delta \to 0$ ends the proof.
$\qed$

\end{document}